\documentclass[12pt,a4paper,sort&compress]{article}
\pdfoutput=1

\usepackage{jheppub}
\usepackage{amsmath}
\usepackage{amssymb}
\usepackage{amsfonts}
\usepackage{amsthm}
\usepackage{shuffle}
\usepackage{enumitem}
\usepackage{color}   
\usepackage{hyperref}
\usepackage{bbm}
\usepackage{multirow}
\usepackage{dsfont}
\usepackage{faktor}
\usepackage{mathdots}

\newcommand{\eps}{\epsilon}
\newcommand{\ord}{\begin{cal}O\end{cal}}

\newcommand{\rd}{\mathrm{d}}

\def\beq{\begin{equation}}
\def\eeq{\end{equation}}
\def\bsp#1\esp{\begin{split}#1\end{split}}

\newenvironment{sloppyequation}[0]{\sloppy\begin{flushleft}\hspace*{0.75cm}\(}{\)\end{flushleft}\fussy}

\newcommand{\beqsloppy}{\begin{sloppyequation}}
\newcommand{\eeqsloppy}{\end{sloppyequation}}

\newcommand{\cD}{\begin{cal}D\end{cal}}

\newcommand{\cH}{\begin{cal}H\end{cal}}

\newcommand{\cK}{\begin{cal}K\end{cal}}

\newcommand{\cM}{\begin{cal}M\end{cal}}

\newcommand{\cT}{\begin{cal}T\end{cal}}

\newcommand{\HH}{\mathbb{H}}
\newcommand{\PP}{\mathbb{P}}
\newcommand{\ZZ}{\mathbb{Z}}
\newcommand{\QQ}{\mathbb{Q}}
\newcommand{\RR}{\mathbb{R}}
\newcommand{\CC}{\mathbb{C}}

\usepackage{hyperref}
\theoremstyle{definition}

\newtheorem{corollary}{Corollary}
\newtheorem{lemma}{Lemma}

\newcommand{\uz}{{z}}

\newcommand{\uPi}{{\Pi}}

\newcommand{\ut}{{t}}

\DeclareMathOperator{\Hom}{Hom}

\DeclareMathOperator{\Imag}{Im}

\DeclareMathOperator{\GL}{GL}
\DeclareMathOperator{\SO}{SO}

\DeclareMathOperator{\OO}{O}
\DeclareMathOperator{\SU}{SU}
\DeclareMathOperator{\NS}{NS}
\DeclareMathOperator{\T}{T}
\DeclareMathOperator{\N}{N}
\DeclareMathOperator{\Sp}{Sp}
\DeclareMathOperator{\SL}{SL}

\DeclareMathOperator{\Ker}{Ker}
\DeclareMathOperator{\so}{\mathfrak{so}}
\DeclareMathOperator{\sll}{\mathfrak{sl}}
\DeclareMathOperator{\spl}{\mathfrak{sp}}
\DeclareMathOperator{\diag}{diag}
\DeclareMathOperator{\im}{im}
\DeclareMathOperator{\him}{hm}

\DeclareMathOperator{\pf}{\mathfrak{p}}

\makeatletter
\DeclareRobustCommand*{\mfaktor}[3][]
{
   { \mathpalette{\mfaktor@impl@}{{#1}{#2}{#3}} }
}
\newcommand*{\mfaktor@impl@}[2]{\mfaktor@impl#1#2}
\newcommand*{\mfaktor@impl}[4]{
   \settoheight{\faktor@zaehlerhoehe}{\ensuremath{#1#2{#3}}}%
   \settoheight{\faktor@nennerhoehe}{\ensuremath{#1#2{#4}}}%
      \raisebox{-0.5\faktor@zaehlerhoehe}{\ensuremath{#1#2{#3}}}%
      \mkern-4mu\diagdown\mkern-5mu%
      \raisebox{0.5\faktor@nennerhoehe}{\ensuremath{#1#2{#4}}}%
}
\makeatother

\newenvironment{mymatrix}[0]{\left(\begin{smallmatrix}}{\end{smallmatrix}\right)}
\newcommand{\ph}{\phantom{-}}

\allowdisplaybreaks

\title{Modular forms for three-loop banana integrals}

\author[a]{Claude Duhr}
\emailAdd{cduhr@uni-bonn.de}
\affiliation[a]{Bethe Center for Theoretical Physics, Universität Bonn, D-53115, Germany
}

\abstract{We study periods of multi-parameter families of K3 surfaces, which are relevant to compute the maximal cuts of certain classes of Feynman integrals. We focus on their automorphic properties, and we show that generically the periods define orthogonal modular forms. Using exceptional isomorphisms between Lie groups of small rank, we show how one can use the intersection product on the periods to identify K3 surfaces whose periods can be expressed in terms of other classes of modular forms that have been studied in the mathematics literature. We  apply our results to maximal cuts of three-loop banana integrals, and we show that depending on the mass configuration, the maximal cuts define ordinary modular forms or Hilbert, Siegel or hermitian modular forms. 
}

\begin{document}

\preprint{BONN-TH-2025-04}

\maketitle


\section{Introduction}
\label{sec:intro}

Feynman integrals are the building blocks for many computations in perturbative Quantum Field Theory (QFT). Therefore, the study of their analytic properties and the development of efficient techniques for their computation is a very active area of research in modern theoretical high-energy physics. It is well known that the functions that arise from Feynman integrals are closely related to quantities of interest also in modern mathematics, specifically in number theory and algebraic geometry. The simplest class of transcendental functions one encounters are multiple polylogarithms~\cite{Mpls1,Mpls2,Remiddi:1999ew}. This class of special functions is well understood, and many Feynman integrals can expressed in terms of them.

It has been known since the early days of QFT that not all Feynman integrals can be expressed in terms of polylogarithms, and it was observed that the higher loop corrections to the electron propagator involve integrals of elliptic type~\cite{Sabry}. In particular, the maximal cuts of the two-loop sunrise integral compute the periods of a family of elliptic curves~\cite{Caffo:1998du,Laporta:2004rb,Laporta:2008sx,Muller-Stach:2011qkg,Muller-Stach:2012tgj}. However, it took until 2015 for the relevant class of transcendental functions to compute the sunrise integral to be identified as elliptic polylogarithms~\cite{ell2,LevinRacinet,MR1265553,BrownLevin,Broedel:2014vla,ell15,EnriquezZerbini}. Soon after that, it was realised that the two-loop sunrise integral with equal masses can also be expressed in terms of iterated integrals of modular forms~\cite{ManinModular,Brown2014MultipleMV,Adams:2017ejb,ell14}, and the period of the family of elliptic curves defines a modular form of weight one for a certain congruence subgroup. By now, Feynman integrals that evaluate to functions of elliptic type are starting to be relatively well understood, with numerous results published over the last few years, see for example ref.~\cite{Bourjaily:2022bwx}.

The special functions that arise from Feynman integrals are not restricted to polylogarithms and their generalisations to elliptic curves. There are two natural ways to go beyond the geometry of elliptic curves for Feynman integrals.\footnote{Recently, also a geometry of special Fano type was identified in the context of Feynman integrals, which goes beyond the cases discussed here~\cite{Schimmrigk:2024xid}.} On the one hand, elliptic curves are Riemann surfaces of genus one, and indeed there are Feynman integrals whose associated geometry are Riemann surfaces of higher genus~\cite{Huang:2013kh,Hauenstein:2014mda,Doran:2023yzu,Abreu:2024jde}. The corresponding maximal cuts then typically evaluate to the periods of the Riemann surface, which can be expressed in terms of Siegel modular forms. The full Feynman integral is expected to expressible in terms of higher-genus generalisations of polylogarithms~\cite{Enriquez_hyperelliptic,DHoker:2023vax,DHoker:2024ozn,Baune:2024biq,Baune:2024ber,DHoker:2025szl,Ichikawa} or iterated integrals that involve Siegel modular forms~\cite{Duhr:2024uid}. On the other hand, elliptic curves are one-dimensional Calabi-Yau (CY) varieties, and there are also many known examples of Feynman integrals related to higher-dimensional CY varieties, cf.~\cite{Brown:2010bw,Bourjaily:2018ycu,Bourjaily:2018yfy,Bourjaily:2019hmc}. Most notably, the $L$-loop banana integral with massive propagators is associated to a CY $(L-1)$-fold~\cite{Bloch:2014qca,MR3780269,Klemm:2019dbm,Bonisch:2020qmm,Bonisch:2021yfw}, and functions related to CY varieties also enter computations for gravitational wave physics~\cite{Driesse:2024feo,Klemm:2024wtd,Frellesvig:2023bbf,Frellesvig:2024rea,Frellesvig:2024zph,Dlapa:2024cje}, the photon self-energy in QED~\cite{Forner:2024ojj} and certain correlations functions in two-dimensional conformal field theories~\cite{Duhr:2022pch,Duhr:2023eld,Duhr:2024hjf}. However, the understanding of the classes of functions that arise from CY varieties is not at the same level as for elliptic curves, though for one-parameter families of CY varieties it is understood how to find appropriate differential equations for the Feynman integrals~\cite{Pogel:2022ken,Pogel:2022vat,Pogel:2022yat,Gorges:2023zgv}. A notable exception are one-parameter families of K3 surfaces. It was observed in ref.~\cite{Primo:2017ipr} that the maximal cuts of the equal-mass three-loop banana integral can be expressed as a product of elliptic integrals, or equivalently, as a square of a modular form. Based on this observation, it  was possible to obtain complete analytic results for the three-loop equal-mass banana integral in dimensional regularisation in terms of iterated integrals of meromorphic and magnetic modular forms~\cite{Broedel:2019kmn,Broedel:2021zij,Pogel:2022yat,magnetic1,magnetic3,Bonisch:2024nru}. 

The main goal of this paper is to take first steps towards understanding Feynman integrals associated to K3 surfaces depending on more than one parameter. Experience from the past shows that, in order to identify and understand the relevant classes of functions, it is extremely useful to start by focusing on the maximal cuts. Indeed, the maximal cuts are expected to compute the (quasi-)periods of the underlying geometry, and non-maximal cuts are typically expressed in terms of iterated integrals~\cite{ChenSymbol} over kernels involving these periods. In almost all modern techniques to derive $\eps$-factorised differential equations for Feynman integral~\cite{Henn:2013pwa} beyond the polylogarithmic case, the maximal cuts and the periods play a distinguished role~\cite{Pogel:2022ken,Pogel:2022vat,Pogel:2022yat,Gorges:2023zgv,Dlapa:2022wdu}. For this reason, we focus in this paper on understanding the classes of functions that arise as periods of families of K3 surfaces. Our main observation is that most concepts known from elliptic curves find a natural generalisation to the K3 case. In particular, the periods can always be identified with a class of special functions known as orthogonal modular forms in the literature~\cite{bruinierbook,orthogonal_PhD,WANG2020107332,Wang_2021,Schaps2022FourierCO,Schaps2023,Assaf:2022aa}. 

Recently, it has been observed that in some cases there is not a single geometry one can attach to some Feynman integral, but there may be more than one class of geometries that can be used to describe the periods or maximal cuts. For example, there can be relations between maximal cuts that compute periods of Riemann surfaces of different genera and/or CY varieties~\cite{Marzucca:2023gto,Duhr:2024hjf,Jockers:2024uan}. 
After having identified the periods of K3 surfaces as orthogonal modular forms, the natural question arises if these periods, or the maximal cuts they compute, can be expressed in terms of other classes of functions. We have already mentioned that for one-parameter families of K3 surfaces, the relevant maximal cuts can always be written as the square of a modular form. The second main result of this paper is a systematic study when K3 periods can be expressed in terms of other classes of modular forms. We show that this question has a natural answer via the well-known exceptional isomorphisms between classical Lie groups of small rank. The relevant Lie groups arise for K3 surfaces as the orthogonal groups that preserve the intersection pairing on the middle cohomology, and we identify cases where the K3 periods can be expressed in terms or ordinary modular forms or Hilbert, Siegel or hermitian modular forms. 
Finally in the last part of this paper, we apply these results to the three-loop banana integral. Remarkably, we find that the maximal cuts of the banana integral always correspond to those cases where the K3 periods can be expressed in terms of other classes of modular forms, and we identify the relevant classes for all possible mass configurations.

This paper is organised as follows: In section~\ref{sec:CYs} we provide a short review of CY varieties and their periods, and we also review the automorphic properties in the case of families of elliptic curves. In section~\ref{sec:K3_review} we focus on families of K3 surfaces, and we explain why the periods (and the mirror map) are orthogonal modular forms. In section~\ref{sec:isom_K3} we show that for families of K3 surfaces depending on a small number of parameters, the periods can be expressed in terms of other classes of modular forms, and we identify the relevant functions space in several cases. In section~\ref{sec:bananas} we apply these results to classify the classes of modular forms that arise from the maximal cuts of three-loop banana integrals. Finally in section~\ref{sec:conclusions} we draw our conclusions. We include several appendices where we provide mathematical proofs omitted throughout the main text.


\section{Calabi-Yau varieties and their periods}
\label{sec:CYs}

\subsection{Brief review of Calabi-Yau varieties}

A Calabi-Yau (CY) $n$-fold is an $n$-dimensional complex K\"ahler manifold $X$ with vanishing first Chern class. The main focus of this paper are CY twofolds, also known as K3 surfaces. We keep the discussion in this first section general, as all concepts equally apply to all values of $n$.

The cohomology groups of $X$ carry a canonical Hodge structure. We are interested in the middle cohomology, which admits the decomposition
\beq
H^n(X,\mathbb{C}) = \bigoplus_{p+q=n}H^{p,q}(X)\,, \qquad h^{p,q} = \dim H^{p,q}(X)\,,
\eeq
where $H^{p,q}(X)$ is the space of all cohomology classes of $(p,q)$-forms on $X$, i.e., the space of all differential forms involving exactly $p$ holomorphic and $q$ antiholomophic differentials. 

The vanishing of the first Chern class implies the existence of a unique holomorphic $(n,0)$-form $\Omega$, which defines a non-trivial class in $H^{n,0}(X)$. A \emph{period} of $X$ is defined by integrating $\Omega$ over an $n$-dimensional cycle in $X$. If we fix a basis $\Gamma_1,\ldots,\Gamma_{b_n}$ of $H_n(X,\mathbb{Z})$ (with $b_k=\dim H_k(X,\mathbb{Z})$ the Betti numbers), we can define a vector of periods
\beq\label{eq:uPi_def}
\uPi(\uz) = \Big(\int_{\Gamma_1}\!\!\Omega,\ldots,\int_{\Gamma_{b}}\!\!\!\Omega\Big)^T\,,
\eeq
where $b\le b_n$ is the number of non-vanishing period integrals. For $n\neq2$, we have $b=b_n$, while for $n=2$, $b$ is the number of transcendental cycles (we will come back to this point in more detail in section~\ref{sec:K3}). Periods do not only capture important information about the geometry of CY varieties, but their knowledge is also essential to compute Feynman integrals attached to CY varieties, like the banana integrals~\cite{MR3780269,Bloch:2014qca,Primo:2017ipr,Broedel:2019kmn,Broedel:2021zij,Klemm:2019dbm,Bonisch:2021yfw,Bonisch:2020qmm,Bezuglov:2021jou,Pogel:2022yat,Pogel:2022vat,Pogel:2022ken,Kreimer:2022fxm}, the icecone integrals~\cite{Duhr:2022dxb}, the traintrack integrals and their generalisations~\cite{Bourjaily:2018ycu,Bourjaily:2018yfy,Bourjaily:2019hmc,Vergu:2020uur,McLeod:2023doa}, or certain integrals that contribute to gravitational wave scattering~\cite{Dlapa:2024cje,Frellesvig:2023bbf,Frellesvig:2024zph,Frellesvig:2024rea,Klemm:2024wtd,Driesse:2024feo}. In particular, the maximal cuts of those integrals evaluate to periods of the underlying CY variety. In some instances, the periods even furnish the complete answer to the Feynman integral~\cite{Duhr:2022pch,Duhr:2023eld,Duhr:2024hjf}. 

Feynman integrals usually depend on a number of external scales.
Correspondingly, we will be interested in families of CY $n$-folds characterised by $m$ independent complex structure moduli $\uz$. The periods will then be functions of $\uz$.
The moduli space $\cM$ of complex structure deformations for families of CY varieties is always unobstructed, and one can show that for $n>2$, we always have $m=h^{n-1,1}$~\cite{MR915841,MR1027500}. The extension of this statement for $n=2$ will be discussed in section~\ref{sec:K3_review}.
Finding an explicit basis of $n$-cycles and performing the integrals in eq.~\eqref{eq:uPi_def} is usually a monumental task. It is typically easier to obtain the periods as solutions to the so-called \emph{Picard-Fuchs differential ideal}, which expresses the flatness of the Gauss-Manin connection on $X$. For one-parameter families, $m=1$, the Picard-Fuchs ideal is generated by a single ordinary differential operator, called the \emph{Picard-Fuchs operator} of $X$.

The cohomology group $H^n(X,\mathbb{Z})$ comes with an additional structure. The intersection pairing between $n$-cycles in $X$ induces a bilinear pairing 
\beq\label{eq:Q_def}
Q: H^n(X,\mathbb{Z})\times H^n(X,\mathbb{Z}) \to \mathbb{Z}\,.
\eeq
$H^n(X,\mathbb{Z})$ is a \emph{lattice}, i.e., a $\mathbb{Z}$-module together with an integer-valued bilinear pairing.
The bilinear pairing has the properties that for $\omega^{p,q}\in H^{p,q}(X)$ and $\omega^{r,s}\in H^{r,s}(X)$, with $p+q=r+s=n$, we have
\beq\bsp
Q(\omega^{p,q},\omega^{r,s})  &\,= (-1)^n\,Q(\omega^{r,s},\omega^{p,q})  \,,\\
Q(\omega^{p,q},\omega^{r,s})  &\,= 0\,, \qquad \textrm{unless } p=s \textrm{ and } q=r  \,,\\
i^{p-q}\,Q(\omega^{p,q},\overline{\omega^{p,q}})  &\,>0 \,.
\esp\eeq
The bilinear pairing $Q$ defines a canonical \emph{polarisation} on the Hodge structure carried by the cohomology group $H^n(X,\ZZ)$. For our choice of basis in eq.~\eqref{eq:uPi_def}, the existence of the polarisation implies the \emph{Hodge-Riemann bilinear relations}, which generalise the well-known Riemann bilinear relations for Riemann surfaces:
\begin{align}\label{eq:Hodge-Riemann-1}
\uPi(z)^T\Sigma\uPi(z) = 0\,,\\
\label{eq:Hodge-Riemann-2}
i^{n^2}\uPi(z)^\dagger\Sigma\uPi(z) > 0\,.
\end{align}
The matrix $\Sigma$ can be related to the Gram matrix of the intersection pairing $\cdot\cap\cdot : H_n(X,\mathbb{Z})\times H_n(X,\mathbb{Z})\to \mathbb{Z}$ between the basis cycles,
\beq\label{eq:sigma_def}
\big(\Sigma^{-1}\big)_{ij} = \Gamma_j\cap\Gamma_i\,.
\eeq
Poincar\'e duality implies that the entries of $\Sigma$ are integers. 
Note that we have $\Sigma^T = (-1)^n\Sigma$, and as a consequence eq.~\eqref{eq:Hodge-Riemann-1} is trivially satisfied for $n$ odd. 

The periods are typically not single-valued functions of the complex structure moduli $\uz$, but they develop a non-trivial monodromy as $\uz$ is varied along a closed loop $\gamma$ encircling one of the singular divisors in the moduli space. We obtain in this way an action of $\pi_1(\cM)$ on the vector of periods. In the basis in eq.~\eqref{eq:uPi_def}, this defines a representation $\rho: \pi_1(\cM) \to \GL(b,\ZZ)$
where $\GL({b},R)$ denotes the multiplicative group of invertible $b\times b$ matrices with entries in a ring $R$.  The monodromy group is then defined as the image of $\pi_1(\cM)$ in $\GL(b,\ZZ)$:
\beq
G_{\!M} := \rho\big(\pi_1(\cM)\big)\subseteq \GL(b,\ZZ)\,.
\eeq

The intersection pairing is monodromy-invariant, which puts strong constraints on the possible form of the monodromy group $G_{\!M}$. In particular, $G_{\!M}$ must be a subgroup of the orthogonal group $\OO(\Sigma,\mathbb{Z})$, where for a ring $R$ we define
\beq\label{eq:O_Sigma_R_def}
\OO(\Sigma,R) := \big\{M\in \GL(b,R): M^T\Sigma M=\Sigma\big\}\,.
\eeq
Note that all elements of $\OO(\Sigma,R)$ have determinant $\pm1$.
We define $\SO(\Sigma,R)$ as the subgroup of $\OO(\Sigma,R)$ with unit detemrinant,
\beq
\SO(\Sigma,R) := \big\{M\in \OO(\Sigma,R): \det M=1\big\}\,.
\eeq
For $n$ odd we can always find a basis in which $\Sigma = J_b$ is the Gram matrix of the standard symplectic pairing, 
\beq
J_b :=\left(\begin{smallmatrix} 0 & -\mathds{1} \\ \mathds{1} & 0\end{smallmatrix}\right)\,.
\eeq
Hence, for $n$ odd the group $\SO(\Sigma,R)$ is isomorphic the symplectic group,
\beq\label{eq:Sp_def}
\SO(\Sigma,R) \simeq \Sp(b,R) = \Big\{M\in\GL(b,R): M^TJ_bM=J_b\Big\}\,.
\eeq

The analytic structure of the periods is most easily described in the case when the moduli space of complex structure deformations has a point of \emph{maximal unipotent monodromy} (MUM). In that case there is a basis where precisely $h^{n-q,q}$ periods diverge as the $q^{\textrm{th}}$ power of a logarithm as we approach the MUM-point (which, without loss of generality, we assume to be at $\uz=0$). In particular, there is a unique distinguished solution that is holomorphic at the MUM-point, and $m$ different solutions that diverge like a single power of a logarithm. We normalise them according to
\beq
\label{eq:MUM-basis}
\Pi_0(\uz) = (2\pi i)^2\big(1+\ord(z_i)\big)\,,\qquad \Pi_k(\uz) = \Pi_0(\uz)\,\frac{\log z_k}{2\pi i} + \ord(z_l)\,,\qquad k=1,\ldots,m\,.
\eeq
Later on it will be useful to consider the vector $\uPi^{(1)}(z) := \big(\Pi_{m}(z),\ldots,\Pi_1(z)\big)^T$.
Using these distinguished periods, we can define canonical coordinates $q_k = \exp(2\pi i t_k) = z_k+\ord(z_l^2)$ on the moduli space in a neighborhood of the MUM-point by
\beq\label{eq:t_def}
t_k(\uz) = \frac{\Pi_k(\uz)}{\Pi_0(\uz)} = \frac{\log z_k}{2\pi i} + \ord(z_l)\,.
\eeq
Their inverse is the \emph{mirror map} $z_k(\ut) = q_k + \ord(q_l^2)$. Inserting the mirror map into the holomorphic period $\Pi_0(\uz)$, we can write $\Pi_0$ as a holomorphic function of the $q_k$:
\beq\label{eq:Pi0_t}
\Pi_0(\ut) := \Pi_0(\uz(\ut)) = 1 +\ord(q_k)\,.
\eeq
Note that we use the same notation to refer to the holomorphic period as a function of $z$ or $t$, because typically no confusion arises.
More generally, if $f(z)$ is a holomorphic function in a neighborhood of the MUM-point $z=0$, then we use the notation $f(t) := f(z(t))$.

\subsection{Automorphic properties of periods: a motivational example}
\label{sec:automorphic}

A main point of this paper is that we can use the structure the group $\OO(\Sigma,\ZZ)$, which contains in particular the monodromy group, to constrain the functional form for the periods.
Let us illustrate this on the simplest possible example of a CY one-fold, i.e., a family of elliptic curves described by a single modulus $z\in\PP^1(\CC)\setminus S = \cM$, with $S$ a finite set of points where the elliptic curve is singular. We can always parametrise the periods and the mirror map by modular forms. This result is not new, 
but it serves as a motivation and an illustrative example of what we want to achieve for families of K3 surfaces in the remainder of this paper.

For a family of elliptic curves the pairing $\Sigma$ is symplectic, and we can find a basis of periods such that $\Sigma=J_2$. The Hodge-Riemann bilinear relations reduce to the well-known Riemann bilinear relations for elliptic curves. For $n=1$ eq.~\eqref{eq:Hodge-Riemann-1} is trivially satisfied, while eq.~\eqref{eq:Hodge-Riemann-2} reduces to 
\beq
i\,\uPi(z)^\dagger\Sigma\uPi(z) = 2\,\big|\Pi_0(z)\big|^2\,\Imag \tau>0\,,
\eeq
We recover the well-known fact from the theory of elliptic curves that the modular parameter\footnote{For elliptic curves, it is conventional to call the modular parameter $\tau$ rather than $t_1$, cf.~eq.~\eqref{eq:t_def}.} $\tau=\tfrac{\Pi_1(z)}{\Pi_0(z)}$ lies in the complex upper half plane $\HH:=\{\tau\in\CC:\Imag\tau>0\}$.
Since the intersection pairing is symplectic, it is invariant under the symplectic group $\Sp(2,\mathbb{Z})= \SL(2,\ZZ)$, which acts via M\"obius transformations on the modular parameter:
\beq\label{eq:Moebius}
\tau = \frac{\Pi_1(z)}{\Pi_0(z)} \to \frac{a\, \Pi_1(z)+b\,\Pi_0(z)}{c \,\Pi_1(z)+d\,\Pi_0(z)} = \frac{a\, \tau+b}{c \,\tau+d} =:\gamma\cdot \tau\,.
\eeq
Note that $\SL(2,\ZZ)$ acts linearly on the periods, but it acts non-linearly on $\tau$. The quotient of $\HH$ by  $\SL(2,\ZZ)$ defines the moduli space of elliptic curves
\beq\label{eq:ell_moduli}
\cM_{\textrm{ell}}  = \mfaktor{\SL(2,\ZZ)}{\HH}\,.
\eeq
Said differently, every point in $\cM_{\textrm{ell}}$ represents a distinct elliptic curve.

The moduli space $\cM_{\textrm{ell}}$ is typically too small to coincide with the moduli space $\cM$ of our family.
The orientation-preserving part $G_{\!M}^+$ of the monodromy group $G_{\!M}$ is always a finite-index subgroup of the symplectic group $\SL(2,\mathbb{Z})$.\footnote{We recall that a subgroup $H\subseteq G$ has finite index if the number of cosets $[G:H] := \big|\faktor{G}{H}\big|$ is finite. See appendix~\ref{app:groups} for a review.} Typically, the monodromy group is a \emph{congruence subgroup} of level $N$ of $\SL(2,\ZZ)$, i.e., a subgroup of finite index which contains the principal congruence subgroup of level $N$,
\beq
\Gamma(N)=\left\{M\in\SL(2,\ZZ) : M = \mathds{1}\!\!\!\!\mod N\right\}\,.
\eeq
When expressed through the modular parameter as in eq.~\eqref{eq:Pi0_t}, the holomorphic period must transform as:
\beq
\Pi_0(\tau) \to \Pi_0\!\left(\tfrac{a\,\tau+b}{c\,\tau+d}\right) =  c\,\Pi_1(z) +  d\,\Pi_0(z) = j_{\textrm{ell}}(\gamma,\tau)\,\Pi_0(\tau)\,,
\eeq
where we defined 
\beq\label{eq:j_ell_def}
j_{\textrm{ell}}(\gamma,\tau) = (c\,\tau+d)\,.
\eeq
This last equation identifies $\Pi_0(\tau)$ as a modular form of weight one for $G_{\!M}^+$. Similarly, since $z$ is clearly monodromy-invariant, the mirror map $z(\tau)$ must be $G_{\!M}^+$-invariant, and so it defines a modular function for $G_{\!M}^+$ (in fact, a Hauptmodul if $\cM$ is a punctured Riemann sphere). Different values of $\tau\in\mathbb{H}$ may give rise to the same point $z(\tau)\in\cM$, and we have the identification
\beq\label{eq:cM_quotient_elliptic}
\cM = \PP^1(\CC)\setminus S \simeq \mfaktor{G_{\!M}^+}{\HH}\,.
\eeq
Let us comment on the relationship between the moduli space of our family and the moduli space of elliptic curves $\cM_{\textrm{ell}}$.
Since $G_{\!M}^+$ has finite index in $\SL(2,\ZZ)$, $\cM$ is a finite cover of $\cM_{\textrm{ell}}$, and the degree of the covering is precisely the index of $G_{\!M}^+$ in $\SL(2,\ZZ)$. 


Let us summarise the main points:
\begin{enumerate}
\item $\tau$ defines a canonical coordinate on the moduli space $\cM$ of the family. The Hodge-Riemann bilinear relations restrict the domain for $\tau$ to be the complex upper half-plane $\HH$. 
\item The orientation preserving part $G_{\!M}^+$ of the monodromy group is a subgroup of $\SO(\Sigma,\ZZ)\simeq \Sp(2,\ZZ)\simeq\SL(2,\ZZ)$ and acts on $\HH$ via M\"obius transformations.
\item The moduli space of all distinct elliptic curves is the quotient $\cM_{\textrm{ell}}$ in eq.~\eqref{eq:ell_moduli}.
\item Our moduli space $\cM$ of complex structure deformations can be recovered from $\HH$ via the quotient in eq.~\eqref{eq:cM_quotient_elliptic}, and it is a finite cover of $\cM_{\textrm{ell}}$.
\item The holomorphic period $\Pi_0(\tau)$ and the mirror map $z(\tau)$ are respectively a modular form of weight one and a modular function for $G_{\!M}^+$.
\end{enumerate}

This story is well known, and also plays an important role when studying Feynman integrals, in particular differential equations satisfied by Feynman integrals whose maximal cuts are attached to families of elliptic curves, cf.,~e.g.,~refs.~\cite{Adams:2017ejb,Adams:2018yfj,Broedel:2018rwm,Gorges:2023zgv}. For higher-dimensional CY varieties, the story is typically more complicated. A notable exception are families of K3 surfaces. In particular, for one-parameter families of K3 surfaces it is always possible to express the periods as products of periods of a family of elliptic curves~\cite{doran,BognerThesis,BognerCY}, and this result has played a crucial role in determining several instances of Feynman integrals attached to one-parameter families of K3 surfaces~\cite{MR3780269,Bloch:2014qca,Broedel:2019kmn,Broedel:2021zij,Pogel:2022yat,Klemm:2024wtd,Forner:2024ojj,Driesse:2024feo}. The main goal of this paper is to take first steps in laying out the roadmap of how this story extends to multi-parameter families of K3 surfaces.


\section{K3 surfaces and their periods}
\label{sec:K3_review}

From the example at the end of the previous section, it should be clear that in order to extend the discussion to K3 surfaces, we need to understand the monodromy group, and more specifically the structure of the orthogonal group $\OO(\Sigma,\ZZ)$. Since for K3 surfaces $n=2$ is even, $\Sigma=\Sigma^T$ is symmetric. We therefore start by providing a review of lattices with symmetric bilinear forms and their orthogonal groups, before we return to discussing K3 surfaces and their periods.
For a review of lattices, in particular in the context of K3 surfaces, we refer for example to ref.~\cite{Huybrechts_2016}.

\subsection{Lattices and their orthogonal groups}

\subsubsection{Lattices}
\label{sec:lattices}

A \emph{lattice} is a finitely generated free $\mathbb{Z}$-module $\Lambda$ together with an integer-valued bilinear pairing 
\beq
b: \Lambda\times \Lambda \to \mathbb{Z}\,.
\eeq
It is called \emph{even} if $b(x,x)$ is even for all $x\in\Lambda$. Otherwise it is called \emph{odd}.
We can fix a basis $e:= (e_1,\ldots,e_r)$ and we have
\beq
\Lambda = \bigoplus_{i=1}^r\mathbb{Z}\,e_i\,.
\eeq
The number $r$ of basis elements is called the \emph{rank} of the lattice. The Gram matrix of $\Lambda$ with respect to this basis is
\beq
\Sigma_{\Lambda,e} := \Big(b(e_i,e_j)\Big)_{1\le i,j\le r} = \Sigma_{\Lambda,e}^T \,.
\eeq
Any two bases $e$ and $e'$ are connected by a $\GL(r,\ZZ)$ transformation. Note that if $M\in\GL(r,\ZZ)$, then $\det M = \pm 1$ and the Gram matrices in two basis are connected by $\Sigma_{\Lambda,e'} = M^T\Sigma_{\Lambda,e}M$. We will typically drop the dependence of the Gram matrix on the choice of basis. The \emph{discriminant} of $\Lambda$ is the determinant $\det\Sigma_\Lambda$ of the Gram matrix (it does not depend on the choice of basis). The lattice is called \emph{unimodular} if its discriminant is $\pm1$.

The bilinear form extends to a bilinear form on $\Lambda\otimes \mathbb{R}\simeq \RR^r$, and every bilinear form over the real numbers is characterised by its \emph{signature}, i.e., the number of positive, negative and zero eigenvalues of $\Sigma_{\Lambda}$ (the eigenvalues do not depend on the choice of basis, and they are all real because $\Sigma_{\Lambda}$ is symmetric). In the following we only consider non-degenerate bilinear forms, i.e., without zero eigenvalues. If $\Sigma_{\Lambda}$ has $p$ positive and $q$ negative eigenvalues, we say that $\Lambda$ has signature $(p,q)$.

The dual lattice $\Lambda\!^\vee$ is defined as the free module of all integer-valued linear forms on $\Lambda$,
\beq
\Lambda\!^\vee = \Hom(\Lambda,\ZZ) = \big\{x\in \Lambda\otimes \QQ: b(x,y) \in\ZZ,\textrm{~~for all~~} y\in \Lambda\big\}\,.
\eeq
The dual lattice inherits from $\Lambda$ the ($\QQ$-valued) scalar product. Moreover, since $\Lambda$ and $\Lambda\!^\vee$ are abelian groups, we can form their quotient $A_{\Lambda}:=\faktor{\Lambda\!^\vee\!\!}{\Lambda}$, which is again an abelian group called the \emph{discriminant group}. The rank of the discriminant group is equal to the absolute value of the discriminant of $\Lambda$,
\beq\label{eq:number_A_L}
\left|A_\Lambda\right| = \left|\det\Sigma_{\Lambda}\right|\,.
\eeq
The discriminant group also inherits the scalar product. In particular, if $\Lambda$ is even, we obtain a quadratic form $q_{\Lambda}: A_\Lambda \to \faktor{\QQ}{2\ZZ}$, called the \emph{discriminant form}.

There is a set of natural operations on lattices. The direct sum of two lattices $\Lambda_1$ and $\Lambda_2$ is the lattice whose module is just the direct sum of the two modules, and the Gram matrix is block-diagonal:
\beq
\Sigma_{\Lambda_1\oplus \Lambda_2} = \left(\begin{smallmatrix} \Sigma_{\Lambda_1} &0\\0& \Sigma_{\Lambda_2}\end{smallmatrix}\right)\,.
\eeq
We will use the notation $\Lambda^{\oplus n} = \underbrace{\Lambda\oplus\cdots\oplus\Lambda}_{n}$. For a non-zero integer $n$, we define the scaled lattice $\Lambda(n)$ to be the lattice with the same underlying module $\Lambda$, but the scalar product has been scaled by $n$,
\beq
\Sigma_{\Lambda(n)} = n\,\Sigma_{\Lambda}\,.
\eeq

Let us conclude by giving some examples of lattices that we will encounter throughout this paper. 
The lattice $\langle1\rangle$ is the lattice whose module is just $\ZZ$, and the bilinear form is $b(1,1)=1$. For some non-zero integer $n$, we also define $\langle n\rangle := \langle1\rangle(n)$. the \emph{hyperbolic lattice} $H$ is the rank two lattice with Gram matrix
\beq\label{eq:Sigma_H}
\Sigma_H = \left(\begin{smallmatrix} 0&1\\1&0\end{smallmatrix}\right)\,.
\eeq
Note that $H=H(-1)$ (up to a basis change).
Finally, to every Dynkin diagram $D$ we can associate a lattice. Its underlying module is the free module generated by the simple roots of $D$, and in the basis of simple roots the Gram matrix is given by the Cartan matrix of $D$.

\subsubsection{Orthogonal groups}
\label{sec:orthogonal_groups}
An isometry is a linear map that preserves the bilinear form. The group of all isometries from a lattice $\Lambda$ to itself is the orthogonal group $\OO(\Lambda)$. If $\Sigma$ is the Gram matrix with respect to a basis, then $\OO(\Lambda)= \OO(\Sigma,\ZZ)$, with $\OO(\Sigma,\ZZ)$ defined in eq.~\eqref{eq:O_Sigma_R_def}. We denote by $\SO(\Lambda)$ the subgroup of $\OO(\Lambda)$ of matrices with unit determinant, and $\SO_0(\Lambda) = \SO_0(\Sigma,\ZZ) := \SO(\Sigma,\ZZ)\cap \SO_0(\Sigma,\RR)$ is the subgroup of elements that lie in the connected component of the Lie group $\SO_0(\Sigma,\RR)$ that contains the identity. 

Let $g\in \OO(\Lambda)$. Then $g$ also acts on the dual lattice $\Lambda\!^\vee$. Indeed, consider a dual form $\varphi\in\Lambda\!^\vee$. Then the action on $\varphi$ is given by $(g^T\varphi)(x) = \varphi(g x)$, $x\in\Lambda$. This action also preserves the scalar product on $\Lambda\!^\vee$ (because the latter was induced by the scalar product on $\Lambda$). In other words, we have an inclusion $\OO(\Lambda)\subseteq \OO(\Lambda\!^\vee)$. Moreover, every $g\in \OO(\Lambda\!^\vee)$ determines an element in the orthogonal group $\OO(A_\Lambda)$ that preserves the discriminant form $q_{\Lambda}$. Putting these two inclusions together, we obtain a group homomorphism from $\OO(\Lambda)$ to $\OO(A_{\Lambda})$. Its kernel is called the \emph{discriminant kernel} of $\Lambda$,
\beq\label{eq:disc_kern}
\widetilde{\OO}(\Lambda) := \Ker\!\big(\OO(\Lambda)\to\OO(A_{\Lambda})\big)\,.
\eeq
The discriminant kernel will play an important role in the following, so we summarise some of its properties. First, since $\widetilde{\OO}(\Lambda)$ is the kernel of a group homomorphism, it is a normal subgroup of $\OO(\Lambda)$. Moreover, it has finite index in $\OO(\Lambda)$. Indeed, from eq.~\eqref{eq:number_A_L} we know that $A_{\Lambda}$ is a finite group, and so all elements of $\OO(A_{\Lambda})$ are permutations of the elements of $A_{\Lambda}$. But there are $|\det \Sigma|!$ such permutations, and we have
\beq
[\OO(\Lambda):\widetilde{\OO}(\Lambda)] = \left|\faktor{\OO(\Lambda)}{\widetilde{\OO}(\Lambda)}\right|\le |\OO(A_{\Lambda})| \le |\det\Sigma|!\,,
\eeq
and so $[\OO(\Lambda):\widetilde{\OO}(\Lambda)]$ is finite. Finally, if we restrict to the elements that lie in the connected component of the identity, then there is a very explicit description of the discriminant kernel,
\beq\label{eq:cD_def}
\cD(\Lambda) := \widetilde{\OO}(\Lambda)\cap\SO_0(\Lambda)= \big\{\mathds{1}+M\Sigma\in\SO_0(\Lambda):M\in\ZZ^{n\times n}\big\}\,.
\eeq
In appendix~\ref{app:groups_1} we show that $\cD(\Lambda)$ always has finite index in both $\widetilde{\OO}(\Lambda)$ and $\SO_0(\Lambda)$.

\subsection{The cohomology and periods of K3 surfaces}
\label{sec:K3}

\subsubsection{The middle cohomology and the K3 lattice}
Consider a K3 surface $X$.
The structure of the middle cohomology of $X$ is very constrained. In particular, its dimension is
\beq
b_2 = \dim H^2(X,\mathbb{Z}) = 22\,.
\eeq
We have already seen that the intersection pairing gives $H^2(X,\mathbb{Z})$ a lattice structure. The lattice is even, and Poinar\'e duality implies that the lattice is unimodular. There is a unique even unimodular lattice of rank 22 (called the K3 lattice), and we have
\beq\label{eq:K3lattice}
H^2(X,\mathbb{Z}) \simeq E_8(-1)^{\oplus2}\oplus H^{\oplus3}\,,
\eeq
where $E_8(-1)$ is the (negative of the) lattice spanned by the simple roots of the exceptional Lie algebra $E_8$ and $H$ is the hyperbolic lattice. Said differently, there is a basis of $H^2(X,\mathbb{Z})$ such that the Gram matrix of the intersection pairing takes the block-diagonal form
\beq
\Sigma_{\textrm{K3}}=\left(\begin{smallmatrix} -\Sigma_{E_8} &0&\phantom{-}0&\phantom{-}0&\phantom{-}0\\0&-\Sigma_{E_8} &\phantom{-}0&\phantom{-}0&\phantom{-}0\\0&0&\phantom{-}\Sigma_H&\phantom{-}0&\phantom{-}0\\0&0&\phantom{-}0&\phantom{-}\Sigma_H&\phantom{-}0\\0&0&\phantom{-}0&\phantom{-}0&\phantom{-}\Sigma_H\end{smallmatrix}\right)\,,
\eeq
where $\Sigma_H$ is defined in eq.~\eqref{eq:Sigma_H} and $\Sigma_{E_8}$ is the Cartan matrix of $E_8$:
\beq
\Sigma_{E_8} = \left(\begin{smallmatrix}
\phantom{-}2&-1&\phantom{-}0&\phantom{-}0&\phantom{-}0&\phantom{-}0&\phantom{-}0&\phantom{-}0\\
-1&\phantom{-}2&-1&\phantom{-}0&\phantom{-}0&\phantom{-}0&\phantom{-}0&\phantom{-}0\\
\phantom{-}0&-1&\phantom{-}2&-1&\phantom{-}0&\phantom{-}0&\phantom{-}0&\phantom{-}0\\
\phantom{-}0&\phantom{-}0&-1&\phantom{-}2&-1&\phantom{-}0&\phantom{-}0&\phantom{-}0\\
\phantom{-}0&\phantom{-}0&\phantom{-}0&-1&\phantom{-}2&-1&\phantom{-}0&-1\\
\phantom{-}0&\phantom{-}0&\phantom{-}0&\phantom{-}0&-1&\phantom{-}2&-1&\phantom{-}0\\
\phantom{-}0&\phantom{-}0&\phantom{-}0&\phantom{-}0&\phantom{-}0&-1&\phantom{-}2&\phantom{-}0\\
\phantom{-}0&\phantom{-}0&\phantom{-}0&\phantom{-}0&-1&\phantom{-}0&\phantom{-}0&\phantom{-}2
\end{smallmatrix}\right)\,.
\eeq
It is easy to check that the K3 lattice has signature $(3,19)$.

The \emph{N\'eron-Severi lattice} of $X$ is defined to be the sublattice of those integer cohomology classes that lie in $H^{1,1}(X)$:
\beq
\NS(X) := H^2(X,\mathbb{Z}) \cap H^{1,1}(X)\,.
\eeq
The rank $\rho$ of $\NS(X)$ is called the \emph{Picard rank} of $X$. Note that we necessarily have $0\le \rho\le 20$, and one can show that the signature of the N\'eron-Severi lattice is $(1,\rho-1)$. The orthogonal complement of $\NS(X)$ in $H^2(X,\mathbb{Z})$ is the \emph{transcendental lattice} $\T(X)$, and we have a direct sum decomposition
\beq\label{eq:H2_NS_T}
H^2(X,\mathbb{Z}) = \NS(X)\oplus \T(X)\,.
\eeq
The N\'eron-Severi lattice has another characterisation. Consider the integration map
\beq
\pi: H_2(X,\mathbb{Z})\to \mathbb{C}; \qquad\Gamma\mapsto \int_{\Gamma}\Omega\,.
\eeq
Then $\NS(X)$ is the (dual of the) kernel of $\pi$,
\beq
\NS(X)^* = \Ker \pi\,.
\eeq
The transcendental lattice then corresponds to (the dual of) those cycles from $H_2(X,\mathbb{Z})$ that lead to non-zero periods after integration. The Gram matrix of the transcendental lattice can be identified with the matrix $\Sigma$ in eq.~\eqref{eq:sigma_def} (where we restrict the cycles to the transcendental cycles).

In the following we will be interested in families of K3 surfaces with a given transcendental lattice.
This leads to the notion of lattice-polarised K3 surfaces. More precisely, consider a lattice $\N$. An $\N$-polarised K3 surface is a K3 surface $X$ together with a primitive embedding of $\N$ into the N\'eron-Severi lattice of $X$. In applications, we will typically have access to the transcendental lattice rather than the N\'eron-Severi lattice (because $\T(X)$ corresponds to the intersection pairing $\Sigma$ on the periods, and so we can determine it from the computation of the periods). For K3 surfaces of Picard rank $\rho\ge 12$, there is always a unique\footnote{The embedding is unique up to isomorphism.} primitive embedding of $\T(X)$ into the K3 lattice, and this embedding in turn determines the N\'eron-Severi lattice (cf., e.g., Corollaries 3.5 and 3.6 of ref.~\cite{Huybrechts_2016}). So, for applications with $\rho\ge 12$, we may also think of a lattice-polarised K3 surface $X$ together with a specific transcendental lattice $\T(X)$ with Gram matrix $\Sigma$.

\subsubsection{The moduli space of K3 surfaces}

Let us now consider a family $X$ of K3 surfaces of Picard rank $\rho$ and depending on $m$ independent moduli $z$. We denote the $m$-dimensional moduli space by $\cM$. If we work with a basis that respects the decomposition in eq.~\eqref{eq:H2_NS_T}, then the vector of periods $\uPi(\uz)$ in eq.~\eqref{eq:uPi_def} has $22-\rho$ non-vanishing entries. The non-vanishing entries correspond to cycles in the transcendental lattice $\T$,\footnote{From now on, we will drop the dependence of the transcendental lattice on $X$, and we simply write $\T$ instead of $\T(X)$.} and, since the moduli space of CY varieties is unobstructed~\cite{MR915841,MR1027500}, we have
\beq
\dim \T = 2+m = 22-\rho\,,
\eeq
and the signature of the transcendental lattice is $(2,m) = (2,20-\rho)$. 
We will only consider families with a MUM-point, and $\Pi_0(\uz)$ denotes the unique period that is holomorphic at the MUM-point, while $\Pi_i(\uz)$, with $1\le i\le m$, diverges as a single power of a logarithm as we approach the MUM-point, cf.~eq.~\eqref{eq:MUM-basis}. Consequently, $\Pi_{m+1}(\uz)$ is the period that diverges double-logarithmically close to the MUM-point. In this basis the Gram matrix is given by
\beq
\Sigma = \left(\begin{smallmatrix} 0&0&1\\ 0&S&0\\ 1&0&0\end{smallmatrix}\right)\,,
\eeq
where $S$ is a symmetric matrix. Said differently, in this basis the transcendental lattice takes the form
\beq
\T = H \oplus \Lambda\,,
\eeq
where $H$ is the hyperbolic lattice and $\Lambda$ is an even lattice of signature $(1,m-1)$ and Gram matrix $S$. We could have chosen a different set of cycles that give rise to the same Gram matrix $\Sigma$ (or equivalently, we could have chosen a different basis for the transcendental lattice). If $\Pi'(z)$ is the period vector in this other basis, then the two vectors of periods are related by an orthogonal transformation,
\beq
\Pi'(z) = M\,\Pi(z)\,,\qquad M\in \OO(\T)\,.
\eeq

\paragraph{The period domain and the period map.} From the first Hodge-Riemann bilinear relation in eq.~\eqref{eq:Hodge-Riemann-1} it follows that the double-logarithmic period $\Pi_{m+1}(\uz)$ can be computed from the holomorphic period and the single-logarithmic periods:
\beq\label{eq:Pi_LogLog_general}
\Pi_{m+1}(\uz) = \frac{1}{\,\Pi_0(\uz)}\,S\!\left[\uPi^{(1)}(\uz)\right] = \Pi_0(\uz)\,S[\ut(\uz)]\,,
\eeq
where $t(z) = \Big(t_m(z),\ldots,t_1(z)\Big)^T$, the $t_k(z)$ are defined in eq.~\eqref{eq:t_def}, and we introduced the shorthand for a symmetric matrix $S$ acting on a vector $t$,
\beq
S[t] := -\frac{1}{2}\,t^TSt\,.
\eeq
An easy computation shows that the second Hodge-Riemann bilinear relation in eq.~\eqref{eq:Hodge-Riemann-2} can be cast in the form
\beq\label{eq:HRK3}
y^TSy > 0\,, \textrm{~~~with~~~} y := \Imag t\,.
\eeq

It is convenient to interpret the vector of periods as a point in complex projective space $\CC\PP^{m+1}$ with homogeneous coordinates
\beq
\Big[\Pi_{m+1}(\uz):\Pi^{(1)}(\uz):\Pi_{0}(\uz)\Big]^T = \Big[S[t(\uz)]:t(\uz):1\Big]^T\,.
\eeq
As a consequence of the Hodge-Riemann bilinear relations~\eqref{eq:Hodge-Riemann-1} and~\eqref{eq:Hodge-Riemann-2}, the period vector defines a point on the quadric
\beq\label{eq:per_dom}
D = \Big\{ \pi\in\CC\PP^{m+1}: \pi^T\Sigma\pi = 0\textrm{~~and~~} \pi^{\dagger}\Sigma\pi>0\Big\}\,.
\eeq
The quadric $D$ is called the \emph{period domain} and the map $\Pi:\cM\to D$ is called the \emph{period map}. Note that the period map is surjective for K3 surfaces, i.e., for every point $\pi \in D$ there is $z\in\cM$ such that $\uPi(z)=\pi$.

\paragraph{Orthogonal transformations as conformal transformations.} The orthogonal group $\OO(\T)$ acts linearly on $\CC\PP^{m+1}$.
It is easy to see that the period domain is invariant under the action of the orthogonal group. A point in the period domain is entirely determined by the mirror map $t(z)$, and the linear action of $\OO(\T)$ on $\CC\PP^{m+1}$ induces a non-linear action on $t(z)$. This non-linear action is well known from conformal field theories, where this construction is known as the \emph{embedding space formalism}~\cite{Mack:1969rr,Dirac:1936fq,Boulware:1970ty,Weinberg:2010fx,Simmons-Duffin:2012juh}. We thus conclude that $\OO(\T)$ acts on the $t(z)$  via conformal transformations in $m$-dimensional Minkowski space. We can describe this action explicitly in terms of generalised M\"obius transformations. Consider $g\in \OO(\T)$. We write
\beq
g = \left(\begin{smallmatrix}
\alpha & p^T & \beta\\ u & M  & v \\ \gamma & q^T & \delta\end{smallmatrix}\right)\,, \qquad \alpha,\beta,\gamma,\delta\in \ZZ,\quad p,q,u,v\in \ZZ^m\,,\quad M\in\ZZ^{m\times m}\,,
\eeq
and we define, for $t\in D$,
\beq\label{eq:generalised_Moebius}
g\cdot t := \frac{S[t]\,u + Mt + v}{S[t]\,\gamma + q^Tt + \delta}\,.
\eeq
One can check that this defines an action of $\OO(\T)$ on $D$, i.e., we have $g\cdot t \in D$ and $g_1\cdot(g_2\cdot t) = (g_1g_2)\cdot t$.

\paragraph{The moduli space of lattice-polarised K3 surfaces.} We now describe the moduli space of lattice-polarised K3 surfaces, which is the analogue of the moduli space of elliptic curves $\cM_{\textrm{ell}}$ defined in eq.~\eqref{eq:ell_moduli}. Mathematical details can be found, e.g., in refs.~\cite{Huybrechts_2016} and~\cite{SB_1982-1983__25__251_0}.

We start by noting that the period domain has two connected components, which are exchanged by complex conjugation,
\beq
D = D_+\cup D_-\,,\qquad D_+\cap D_-=\emptyset\,,\qquad D_-=\overline{D}_+\,.
\eeq
We define $\OO^+(\T)$ to be the subgroup of the orthogonal group $\OO(\T)$ that fixes $D_+$,
\beq
\OO^+(\T) := \big\{g\in \OO(\T): g\cdot D_+ = D_+\big\}\,.
\eeq
$\OO^+(\T)$ is clearly a subgroup of index two of $\OO(\T)$. We also define
\beq
\widetilde{\OO}{}^+(\T) := \OO^+(\T) \cap \widetilde{\OO}(\T)\,,
\eeq
where $\widetilde{\OO}(\T)$ is the discriminant kernel defined in eq.~\eqref{eq:disc_kern}.

Consider a lattice $\N$ of signature $(1,\rho-1)$ and its orthogonal complement $\T=\N^{\bot}$.
$\T$ has signature $(2,m)$ and the action of $\OO(\T)$ on $D$ is properly discontinuous (\cite{Huybrechts_2016}, Chapter 6, Remark 1.10). If $\Gamma\subseteq \OO(\T)$ is a torsion free subgroup, then it follows from the Bailey-Borel theorem that the quotient $\mfaktor{\Gamma}{D}$ is a smooth quasi-projective variety (\cite{Huybrechts_2016}, Chapter 6, Theorem 1.13). Such a subgroup always exists (\cite{Huybrechts_2016}, Chapter 6, Proposition 1.11). If $\Gamma$ is not torsion free, then $\mfaktor{\Gamma}{D}$ is still a quasi-projective variety (but not necessarily smooth). We can then take the quotient of the period domain by the discriminant kernel, and we get the moduli space of (pseudo-ample) $\N$-polarised K3 surfaces~\cite{SB_1982-1983__25__251_0},
\beq\label{eq:moduli_space_K3}
\cM_{\textrm{K3}}(\N) \simeq \mfaktor{\widetilde{\OO}(\T)}{D} \simeq \mfaktor{\widetilde{\OO}{}^+(\T)}{D_+}\,.
\eeq
$\cM_{\textrm{K3}}(\N)$ is the analogue for $\N$-polarised K3 surface of the moduli space $\cM_{\textrm{ell}}$ of elliptic curves in eq.~\eqref{eq:ell_moduli}. Said differently, (by the surjectivity of the period map), every point in $\cM_{\textrm{K3}}(\N)$ corresponds to an $\N$-polarised K3 surface. Let us comment on the
appearance of the discriminant kernel $\widetilde{\OO}(\T)$ rather than the orthogonal group $\OO(\T)$. Both the orthogonal group and its discriminant kernel are infinite groups, and their quotient is finite (because $\widetilde{\OO}(\T)$ has finite index in $\OO(\T)$). In other words, $\widetilde{\OO}(\T)$ is almost `as large' as $\OO(\T)$. The elements in the quotient are those that do not act trivially on the discriminant lattice $A_{\T}$. By taking the quotient by $\widetilde{\OO}(\T)$ in eq.~\eqref{eq:moduli_space_K3}, we consider K3 surfaces different if their periods are the same up to an orthogonal transformation, but this transformation acts non-trivially on the discriminant lattices.

\subsubsection{The monodromy group and automorphic properties of families of K3 surfaces} 
Just like in the case of elliptic curves, the moduli space $\cM_{\textrm{K3}}(\N)$ is too small to coincide with the moduli space of our family. Nevertheless, we know from the surjectivity of the period map that every point in $D_+$ corresponds to a member from our family. Just like in the case of elliptic curves, we want to identify points in $D_+$ that corresponds to a monodromy transformation, and just like in the construction of the moduli space $\cM_{\textrm{K3}}(\N)$, we only want to identify points that represent periods that are related by a monodromy transformation that acts trivially on the discriminant lattice. We then define
\beq\bsp
\widetilde{G}_{\!M} &\,= G_{\!M}\cap \widetilde{\OO}(\T)\,,\\
\widetilde{G}_{\!M}^+ &\,= G_{\!M}\cap \widetilde{\OO}{}^+(\T)\,.
\esp\eeq
The monodromy group $G_{\!M}$ of a family of lattice-polarised K3 surfaces always has finite index in $\OO(\T)$ (cf.,~e.g.,~the discussion in ref.~\cite{Huybrechts_2016}, section 4), and so $\widetilde{G}_{\!M}$ and $\widetilde{G}_{\!M}^+$ have finite index in $\widetilde{\OO}(\T)$ and $\widetilde{\OO}{}^+(\T)$ (because intersections of finite index-subgroups have themselves finite index, see appendix~\ref{app:groups}). 
By the same argument as before, the quotient space defines a quasi-projective variety, and
we have the isomorphism
\beq\label{eq:moduli_space}
\cM \simeq \mfaktor{\widetilde{G}_{\!M}}{D} = \mfaktor{\widetilde{G}_{\!M}^+}{D_+}\,.
\eeq
Equation~\eqref{eq:moduli_space} is the analogue of eq.~\eqref{eq:cM_quotient_elliptic} for elliptic curves. Just like in the case of elliptic curves, we may ask for the relationship between $\cM$ and $\cM_{\textrm{K3}}(\N)$. Since $\widetilde{G}_{\!M}^+$ has finite index in $\widetilde{\OO}{}^+(\T)$, we immediately see that $\cM$ is a finite cover of $\cM_{\textrm{K3}}(\N)$, and the degree of the covering is the index of $\widetilde{G}_{\!M}^+$ in $\widetilde{\OO}{}^+(\T)$, again in complete analogy with the elliptic case. In table~\ref{tab:summary} we summarize the correspondence of the well-known concepts from the elliptic case in section~\ref{sec:automorphic} to the case of K3 surfaces. We see that all concepts known from families of elliptic curves have a generalisation to families of lattice-polarised K3 surfaces. We stress that this is very specific to K3 surfaces, and does not necessarily generalise to families of higher-dimensional CY varieties.

The period map allows us to assign a point $\big[S[t(z)]:t(z):1\big] \in D$ to every $z\in\cM$. It is easy to see that every point in $D$ is uniquely determined by the vector $t(z)$, and since $D_+$ and $D_-$ are related by complex conjugation, we can focus without loss of generality on $D_+$. Conversely, by surjectivity of the period map, every point in $D_+$ corresponds to a member from our family. The monodromy group $\widetilde{G}_{\!M}^+$ acts on $D_+$ via generalised M\"obius transformations from eq.~\eqref{eq:generalised_Moebius}.
Using the same reasoning as in the elliptic case, we see that the mirror map $z(t)$ must be invariant under generalised M\"obius transformations,
\beq\label{eq:mirror_orthogonal_function}
z(g\cdot t) = z(t)\,,\qquad g\in \widetilde{G}_{\!M}^+\,.
\eeq
Similarly, the holomorphic period $\Pi_0(t)$ must transform as
\beq\label{eq:Pi_orthogonal_form}
\Pi_0(g\cdot t)  = j(g,t)\,\Pi_0(t)\,,
\eeq
with 
\beq\label{eq:automorphic}
j(g,t) := \big(S[t]\,\gamma + q^Tt + \delta\big)\,,\qquad g = \left(\begin{smallmatrix}
\alpha & p^T & \beta\\ u & M  & v \\ \gamma & q^T & \delta\end{smallmatrix}\right) \in \widetilde{G}_{\!M}^+\,.
\eeq
we arrive at the following conclusion:
\begin{quote}\emph{The mirror map $z(t)$ and the holomorphic period $\Pi_0(t)$ are respectively a modular function and a modular form for the group $\widetilde{G}_{\!M}^+$.}
\end{quote}

\begin{table}[!t]
\begin{center}
\begin{tabular}{|c|c|c|}
\hline
& Elliptic curves & K3 surfaces\\
\hline
$\dim \cM$ & 1 & $m=20-\rho$\\
Lattice & $J_2$ & $\T = H\oplus \Lambda(n)$\\
Connected component of the period domain & $\HH $ & $D_+$\\
Symmetry & $\SL(2,\ZZ)$ & $\widetilde{\OO}{}^+(\T)$\\
Moduli space & $\cM_{\textrm{ell}}$ & $\cM_{\textrm{K3}}(\N)$\\
Automorphic properties  & $G_{\!M}^+$ & $\widetilde{G}_{\!M}^+$\\
\hline
\end{tabular}
\end{center}
\caption{\label{tab:summary}Comparison of different concepts encountered for elliptic curves and $\N$-polarised K3 surfaces (with $\T=\N^{\bot}$).}
\end{table}

Modular forms for orthogonal groups are also called \emph{orthogonal modular forms}. There is a substantial body of mathematical literature on orthogonal modular forms \cite{bruinierbook,orthogonal_PhD,WANG2020107332,Wang_2021,Schaps2022FourierCO,Schaps2023}, including algorithms for computer codes~\cite{Assaf:2022aa}. They can often be constructed as certain integrals over ordinary modular forms using a procedure called the \emph{Borcherds lift}~\cite{Borcherds:1998aa}. It would be interesting to study the connection between Feynman integrals related to K3 surfaces and orthogonal modular forms in more detail, and we leave this for future work. In the remainder of this paper we focus on another aspect: it is well known that for small values of $m$, there are isomorphisms between the orthogonal groups $\OO(2,m)$ and other Lie groups. For small values of $m$, we therefore expect that we can express K3 periods in terms of other classes of modular forms. We will describe some of these cases, as well as the relevant modular forms in the next section.



\section{Exceptional isomorphisms and K3 surfaces of large Picard rank}
\label{sec:isom_K3}

In the previous section we have seen that for an $\N$-polarised family of K3 surfaces, the holomorphic period and the mirror map are modular forms for the monodomy group $\widetilde{G}_{\!M}^+\subseteq \widetilde{\OO}{}^+(\T)$ when expressed in terms of the canonical coordinate $t$ on the moduli space. In particular, the automorphic properties are tightly linked to the structure of the orthogonal group. 

Since the signature of the transcendental lattice is always $(2,m)$, we know that $\SO_0(\T) \subset \SO_0(2,m)$. It is well known that for $p+q\le 6$, the Lie groups $\SO_0(p,q)$ are isomorphic other Lie groups. Hence, for $m\le 4$, we expect that the automorphic properties of the periods may be related to those of other groups, potentially leading to other classes of modular forms that may have been studied in the literature. 

The goal of this section is to explore some of these isomorphisms and their consequences for families of K3 surfaces with Picard rank $\rho\ge 16$. We start by giving a very brief review of the exceptional isomorphisms of Lie algebras of small rank, mostly focusing on the case of complex Lie algebras (which is the case typically studied in physics). In later subsections we will indicate how these isomorphisms can be applied to K3 surfaces.

\subsection{Exceptional isomorphisms between Lie algebras and Lie groups}

The classification of simple Lie groups is easiest in the complex case. There are 4 infinite families of simple complex Lie algebras corresponding to the Dynkin diagrams $A_n$, $B_n$, $C_n$ and $D_n$ (as well as six exceptional Lie algebras, which will not play any role here). These Dynkin diagrams represent the root lattices of the Lie algebras $\sll(n+1,\mathbb{C})$, $\so(2n+1,\mathbb{C})$, $\spl(2n,\mathbb{C})$ and $\so(2n,\mathbb{C})$ for $n\ge 1$, which are the Lie algebras of the simple complex Lie groups $\SL(n+1,\mathbb{C})$, $\SO(2n+1,\mathbb{C})$ , $\Sp(2n,\mathbb{C})$ and $\SO(2n,\mathbb{C})$, respectively. The number of nodes of the Dynkin diagram corresponds to the rank of the Lie algebra (i.e., the dimension of its Cartan subalgebra).  

It is then easy to see from the Dynkin diagrams that for low ranks, we obtain the following isomorphisms:
\begin{itemize}
\item $A_1=B_1=C_1$, or equivalently $\sll(2,\mathbb{C}) = \so(3,\mathbb{C})  = \spl(2,\mathbb{C})$. At the level of Lie groups, we find that $\SL(2,\mathbb{C})$ and $\Sp(2,\mathbb{C})$ are isomorphic. These groups are not isomorphic to $\SO(3,\mathbb{C})$, but there is a 2-to-1 map. In later sections, we will be interested in the real form $\SO_0(\Sigma_1,\RR)\simeq \SO_0(2,1)$, which preserves the quadratic form
\beq
\Sigma_1 = \left(\begin{smallmatrix} 0&\phantom{-}0&1\\0&-2&0\\1&\phantom{-}0&0\end{smallmatrix}\right)\,.
\eeq
 We have the exact sequence
\beq\label{eq:SL2_to_SO21}
\mathbb{Z}_2\longrightarrow \SL(2,\mathbb{R})\overset{\phi_3}{\longrightarrow} \SO_0(\Sigma_1,\RR)\longrightarrow0\,.
\eeq
The first map is simply the inclusion, and $\phi_3$ is the symmetric square map,
\beq\label{eq:phi3_def}
\phi_3\!\left(\begin{smallmatrix}a&b\\c&d\end{smallmatrix}\right) = \left(\begin{smallmatrix} a^2 & 2 a b & b^2 \\
 a c & a d+b c & b d \\
 c^2 & 2 c d & d^2
\end{smallmatrix}\right)\,.
\eeq
\item $D_2=A_1\times A_1$, or equivalently $\so(4,\mathbb{C}) =\sll(2,\mathbb{C})\oplus \sll(2,\mathbb{C})$. We will in particular be interested in the real form $\SO_0(\Sigma_2,\RR)\simeq \SO_0(2,2)$, with
\beq\label{eq:Sigma2_def}
\Sigma_2 = \left(\begin{smallmatrix} 0&\phantom{-}0&\phantom{-}0&1\\0&\phantom{-}0&-1&0\\0&-1&\phantom{-}0&0\\1&\phantom{-}0&\phantom{-}0&0\end{smallmatrix}\right)\,.
\eeq
A the level of Lie groups we have an exact sequence
\beq
\ZZ_2\longrightarrow \SL(2,\mathbb{R})\times\SL(2,\mathbb{R}) \overset{\phi_4}{\longrightarrow} \SO_0(\Sigma_2,\mathbb{R})\longrightarrow0\,.
\eeq
The first map is the diagonal inclusion and
\beq\label{eq:phi4_def}
 {\phi}_4(\gamma_1,\gamma_2) = \left(\begin{smallmatrix} a_1 a_2 & a_1 b_2 & a_2 b_1 & b_1 b_2 \\
 a_1 c_2 & a_1 d_2 & b_1 c_2 & b_1 d_2 \\
 a_2 c_1 & b_2 c_1 & a_2 d_1 & b_2 d_1 \\
 c_1 c_2 & c_1 d_2 & c_2 d_1 & d_1 d_2\end{smallmatrix}\right)\,, \qquad \gamma_i= \left(\begin{smallmatrix}a_i & b_i\\ c_i&d_i\end{smallmatrix}\right)\,.
 \eeq
\item $B_2=C_2$, or equivalently $\so(5,\mathbb{C}) = \spl(4,\mathbb{C})$. We will encounter the real form $\SO_0(\Sigma_2,\RR)\simeq\SO_0(2,3)$, with
\beq
\Sigma_3 = \left(\begin{smallmatrix} 0&0&\phantom{-}0&0&1\\0&0&\phantom{-}0&1&0\\0&0&-2&0&0\\0&1&\phantom{-}0&0&0\\1&0&\phantom{-}0&0&0\end{smallmatrix}\right)\,,
\eeq
and we have
\beq\label{eq:Sp4C_SO5C}
\mathbb{Z}_2\longrightarrow \Sp(4,\mathbb{R})\overset{\phi_5}{\longrightarrow} \SO_0(\Sigma_3,\mathbb{R})\longrightarrow0\,.
\eeq
The first map is again the inclusion, and
\beq
\phi_5\!\left(\begin{smallmatrix}
 a_1 & a_2 & b_1 & b_2 \\
 a_3 & a_4 & b_3 & b_4 \\
 c_1 & c_2 & d_1 & d_2 \\
 c_3 & c_4 & d_3 & d_4
 \end{smallmatrix}\right)=
\left(\begin{smallmatrix}
 a_1 a_4-a_2 a_3 & a_3 b_2-a_1 b_4 & -a_3 b_1+a_4 b_2+a_1 b_3-a_2 b_4 \
& a_2 b_3-a_4 b_1 & b_2 b_3-b_1 b_4 \\
 a_2 c_3-a_1 c_4 & a_1 d_4-b_2 c_3 & -a_1 d_3+a_2 d_4+b_1 c_3-b_2 c_4 \
& b_1 c_4-a_2 d_3 & b_1 d_4-b_2 d_3 \\
 a_1 c_2-a_2 c_1 & b_2 c_1-a_1 d_2 & a_1 d_1-a_2 d_2-b_1 c_1+b_2 c_2 \
& a_2 d_1-b_1 c_2 & b_2 d_1-b_1 d_2 \\
 a_3 c_2-a_4 c_1 & b_4 c_1-a_3 d_2 & a_3 d_1-a_4 d_2-b_3 c_1+b_4 c_2 \
& a_4 d_1-b_3 c_2 & b_4 d_1-b_3 d_2 \\
 c_2 c_3-c_1 c_4 & c_1 d_4-c_3 d_2 & c_3 d_1-c_4 d_2-c_1 d_3+c_2 d_4 \
& c_4 d_1-c_2 d_3 & d_1 d_4-d_2 d_3
\end{smallmatrix}\right)\,.
\eeq
\item $D_3=A_3$, or equivalently $\so(6,\mathbb{C}) = \sll(4,\mathbb{C})$. 
We will need the real form $\SO_0(\Sigma_4,\RR)\simeq\SO_0(2,4)$ with
\beq\label{eq:Sigma_4}
\Sigma_4 = \left(\begin{smallmatrix} 0&0&\phantom{-}0&\phantom{-}0&0&1\\0&0&\phantom{-}0&\phantom{-}0&1&0\\0&0&-2&\phantom{-}1&0&0\\0&0&\phantom{-}1&-2&0&0\\0&1&\phantom{-}0&\phantom{-}0&0&0\\1&0&\phantom{-}0&\phantom{-}0&0&0\end{smallmatrix}\right)\,,
\eeq
and we have an exact sequence,
\beq\label{eq:SO(2,6)_sequence}
\mathbb{Z}_2\longrightarrow \SU(2,2,\ord_F)^*\overset{\phi_6}{\longrightarrow} \SO(\Sigma_4,\mathbb{R})\longrightarrow0\,,
\eeq
where $F=\QQ(i\sqrt{3})$ and $\SU(2,2,\ord_F)^*$ is an extension of the group~\cite{hermitian_thesis,HAUFFEWASCHBUSCH202122}
\beq\label{eq:SU22_def}
\SU(2,2,\ord_F) = \big\{M\in \SL(4,\ord_F): M^{\dagger}J_4M = J_4\big\}\,.
\eeq
Here $\ord_F$ is the ring of integers of $F$ and will be defined more carefully in section~\ref{sec:Hilbert}.
The map $\phi_6$ is more complicated than in the previous cases, but it can still be constructed explicitly. In the following, we will not need it explicit form, which can be found, e.g., in refs.~\cite{hermitian_thesis,HAUFFEWASCHBUSCH202122}.
\end{itemize}

\subsection{The isomorphism $\SO_0(2,1) \simeq {\SL(2,\mathbb{R})/\ZZ_2}$}
\label{sec:A1=B1=C1}
Consider a one-parameter family of K3 surfaces. Since the transcendental lattice is even and has signature $(2,1)$, it is determined by a single positive integer $n\ge 1$,
\beq
\T_n = H\oplus \langle 2n\rangle\,.
\eeq 
In an appropriate basis, the Gram matrix is
\beq
\Sigma = \left(\begin{smallmatrix} 0&0&1\\0&2n&0\\1&0&0\end{smallmatrix}\right)\,.
\eeq
Since $\SO_0(\T_n)\subseteq \SO_0(2,1,\mathbb{R}) \simeq\SL(2,\mathbb{R})/\ZZ_2$, we expect the orthogonal group of $\T_n$ to be essentially isomorphic to some subgroup of $\SL(2,\mathbb{R})$. We have an exact sequence~\cite{Dolgachev:1996aa},
\beq\label{eq:Dolgachev}
\mathbb{Z}_2\longrightarrow \Gamma_0(n)^+\overset{\phi_{n,3}}{\longrightarrow} \widetilde{\OO}{}^+(\T_n)\to0\,,
\eeq
where $\phi_{n,3} = \varphi_{\Delta_{-n,1}}\circ \phi_3$ is the composition of the map in eq.~\eqref{eq:phi3_def} and the group homomorphism 
\beq\label{eq:varphi_delta_def}
\varphi_{\Delta}(\gamma) := \Delta \gamma\Delta^{-1}\,,
\eeq
where we defined
\beq\label{eq:delta_def}
\Delta_{n,p} := \diag(\underbrace{n,\ldots,n}_{p},1,\ldots,1)\,.
\eeq 
$\Gamma_0(n)^+$ is the subgroup of $\SL(2,\mathbb{R})$ generated by the congruence subgroup
\beq
\Gamma_0(n) = \left\{\begin{mymatrix}a&b\\c&d\end{mymatrix}\in\SL(2,\ZZ):c = 0\!\!\!\mod n\right\}\,,
\eeq
and the Atkin-Lehner involution
\beq\label{eq:F_n_def}
F_n = \begin{mymatrix}0&-\tfrac{1}{\sqrt{n}}\\ \sqrt{n}&0\end{mymatrix}\,.
\eeq
The image of the subgroup $\Gamma_0(n)$ generates the subgroup $\cD(\T_n)$, and so we see that $\widetilde{\OO}{}^+(\T_n)$ is generated by $\cD(\T_n)$ augmented by the images of Atkin-Lehner involutions.

Let us interpret this result in the context of K3 surfaces (cf.,~e.g.,~ref.~\cite{Dolgachev:1996aa}). Let $\gamma = \begin{mymatrix}a&b\\c&d\end{mymatrix}\in \Gamma_0(n)^+$. By direct computation, one sees that the generalised M\"obius transformation and the automorphic factor for the orthogonal group in eqs.~\eqref{eq:generalised_Moebius} and~\eqref{eq:automorphic} take the form
\beq\label{eq:modular_HA1}
\phi_3(\gamma)\cdot t = \frac{at+b}{c t+d}\textrm{~~~and~~~} j(\phi_3(\gamma),t) = (ct+d)^2 = j_{\textrm{ell}}(\gamma,t)^2\,.
\eeq
In other words, the generalised M\"obius transformation and the automorphic factor for the orthogonal group reduce to those for $\SL(2,\RR)$. 
The second Hodge-Riemann bilinear relation in eq.~\eqref{eq:HRK3} gives
\beq\label{eq:modular_2}
2n\,(\Imag t)^2 > 0\,.
\eeq
This implies $\Imag t\neq 0$, and so the period domain is
\beq
D = D_+ \cup \overline{D}_+\textrm{~~~and~~~} D_+ = \HH\,,
\eeq
and we have the isomorphism
\beq
\cM_{\textrm{K3}}(\N) = \mfaktor{\Gamma_0(n)^+}{\HH}\,,
\eeq
with $N = E_8(-1)^{\oplus2}\oplus H\oplus\langle-2n\rangle$.

We know that the monodromy group $\widetilde{G}_{\!M}^+$ of our family is a finite-index subgroup of $\widetilde{\OO}{}^+(\T)$, and so via eq.~\eqref{eq:Dolgachev} we can identify a subgroup $\widetilde{\Gamma}_{\!M}^+\subseteq \Gamma_0(n)^+$ such that, $\phi_{n,3}(\widetilde{\Gamma}_{\!M}^+) = \widetilde{G}_{\!M}^+$. However we do not necessarily have equality between $\widetilde{\Gamma}_{\!M}^+$ and $\Gamma_0(n)^+$.\footnote{As an example, consider ref.~\cite{verrill1996}, where a family with $n=3$ is discussed, with monodromy group $\Gamma_0(6)^+$.} Combining eqs.~\eqref{eq:mirror_orthogonal_function} and~\eqref{eq:Pi_orthogonal_form} with eq.~\eqref{eq:modular_HA1}, we find that, for all $\gamma=\begin{mymatrix}a&b\\c&d\end{mymatrix}\in\widetilde{\Gamma}_{\!M}^+$ we have
\beq\bsp\label{eq:modular_1}
\Pi_0(\phi_{n,3}(\gamma)\cdot t) &\,= \Pi_0\!\left(\tfrac{at+b}{ct+d}\right) = j_{\textrm{ell}}(\gamma,t)^2\,\Pi_0(t)\,,\\
z(\phi_{n,3}(\gamma)\cdot t) &\,= z\!\left(\tfrac{at+b}{ct+d}\right) = z(t)\,.
\esp\eeq
From the previous equation, we see that the holomorphic period and the mirror map have the appropriate transformation properties to define modular forms for $\widetilde{\Gamma}_{\!M}^+$. However, at this point we only know that $\widetilde{\Gamma}_{\!M}^+$ is a finite-index subgroup of $\Gamma_0(n)^+\subseteq \SL(2,\RR)$. We will need to identify a finite-index subgroup $\widetilde{\Gamma}_{\!\cD}^+\subseteq \widetilde{\Gamma}_{\!M}^+\cap \SL(2,\ZZ)$ to conclude that $\Pi_0(t)$ and $z(t)$ are modular forms for $\widetilde{\Gamma}_{\!\cD}^+$. We now argue that such a subgroup always exists. The method used here serves as an example for the subsequent sections.

We start by noting that $\Gamma_0(n)^+$ contains $\Gamma_0(n)$ as a subgroup of index two. Indeed, the element $F_n$ in eq.~\eqref{eq:F_n_def} normalises $\Gamma_0(n)$, i.e., for all $\gamma\in\Gamma_0(n)$, $F_n\gamma F_n^{-1}\in\Gamma_0(n)$ (this can easily be checked by direct computation). It follows that we have
\beq
\Gamma_0(n)^+ = \Gamma_0(n)\cup F_n\Gamma_0(n)\,,
\eeq
and so $\Gamma_0(n)$ has index 2 in $\Gamma_0(n)^+$. Next we note that the monodromy group $G_{\!M}$ must be infinite. This can easily be seen in our case from the existence of a MUM-point (e.g., because the monodromy group of $\log z$ is $\pi_1(\CC^{\times})\simeq \ZZ$). As a consequence, $\widetilde{G}_{\!M}^+$ and $\widetilde{\Gamma}_{\!M}^+$ are also infinite (because $\widetilde{G}_{\!M}^+$ has finite index in ${G}_{\!M}$, and $\widetilde{\Gamma}_{\!M}^+$ is essentially isomorphic to $\widetilde{G}_{\!M}^+$). It follows that $\widetilde{\Gamma}^+_{\!\cD}:=\widetilde{\Gamma}_{\!M}^+\cap \Gamma_0(n)$ is non-trivial (and in fact itself infinite, see the derivation in appendix~\ref{app:groups}). Moreover, since $\widetilde{\Gamma}^+_{\!\cD}$ is the intersection of two subgroups of finite index, it has itself finite index in $\SL(2,\ZZ)$. Hence, we have identified a finite-index subgroup $\widetilde{\Gamma}^+_{\!\cD}$ of $\SL(2,\ZZ)$ such that eq.~\eqref{eq:modular_1} holds for all $\gamma\in\widetilde{\Gamma}^+_{\!\cD}$, and so $\Pi_0$ and $z$ define ordinary modular forms for $\widetilde{\Gamma}^+_{\!\cD}$. 

Let us make some comments at this point. First, while the previous argument shows the existence of a finite-index subgroup of $\SL(2,\ZZ)$ contained in $\widetilde{\Gamma}_{\!M}^+$, we do by no means claim that it is the largest group with that property! This is also irrelevant for our goal: we merely wanted to conclude that the holomorphic period and the mirror map admit a modular parametrisation for \emph{some} finite-index subgroup. A detailed study of those modular properties is a priori more complicated. Second, the key to our construction of $\widetilde{\Gamma}_{\!\cD}^+$ was the fact that we could identify the subgroup $\Gamma_0(n)$ that has finite index in both $\Gamma_0(n)^+$ and $\SL(2,\ZZ)$. The fact that this subgroup is $\Gamma_0(n)$ is not essential, and we could have started from any finite-index subgroup of $\Gamma\subseteq\Gamma_0(n)$. We now illustrate how we can find such a subgroup in other cases by using the result from section~\ref{sec:congruence}.
Take $\Gamma$ to be the subgroup of $\SL(2,\RR)$ such that 
\beq
\phi_{n,3}(\Gamma_2(n)) = \varphi_{\Delta_{-n,1}}\big(\cD(\T_1(n))\big)\,,
\eeq
where the map $\varphi_{\Delta}$ was defined in eq.~\eqref{eq:varphi_delta_def}.
We know that such a $\Gamma_2(n)$ exists from eq.~\eqref{eq:SL2_to_SO21}, and an easy computation shows that
\beq\label{eq:Gamma_2n_def}
\Gamma_2(n) := \big\{\begin{mymatrix}a&b\\c&d\end{mymatrix}\in\SL(2,\ZZ): a,d=1\!\!\!\!\mod n, b=0\!\!\!\!\mod n, c=0\!\!\!\!\mod n^2\big\}\,.
\eeq
Note that $\Gamma(n^2)\subseteq \Gamma_2(n)$, and so $\Gamma_2(n)$ is a congruence subgroup.
We could then equally-well have defined $\widetilde{\Gamma}^+_{\!\cD} :=\widetilde{\Gamma}^+_{\!M}\cap \Gamma_2(n)$, and all conclusions would have remained the same.

\subsection{The isomorphism $\SO_0(2,2) \simeq (\SL(2,\mathbb{R})\times \SL(2,\mathbb{R}))/\ZZ_2$}
\label{sec:D2=A1xA1}

We now turn to the discussion of two-parameter families of K3 surfaces whose transcendental lattice admits a Gram matrix of the form
\beq
\Sigma = \begin{mymatrix}
0 & 0 & 1\\
0&S&0\\
1&0&0\end{mymatrix}\,,
\eeq
where $S$ is a symmetric $2\times 2$ matrix that defines an even lattice of rank two and signature $(1,1)$, i.e., $S$ is a symmetric matrix that admits one positive and one negative eigenvalue and the entries on the diagonal are even. Unlike in the case of one-parameter families, there is more than one inequivalent choice for $S$ (even up to rescaling). In the following we discuss two cases which will turn out to be sufficient to discuss the three-loop banana integrals in section~\ref{sec:bananas}.

\subsubsection{Case 1: products of modular forms}
\label{sec:ExE}

We start by considering a family of K3 surfaces with  transcendental lattice
\beq\label{eq:T=H+H}
\T_n = H \oplus H(n)\,,
\eeq
which corresponds to the choice $S=\begin{mymatrix}0&n\\n&0\end{mymatrix}$. Note that for $n=-1$ we recover the Gram matrix $\Sigma_2$ in eq.~\eqref{eq:Sigma2_def}, and in general we have
\beq
 \Delta_{-n,1}^T\Sigma\Delta_{-n,1} = -n\Sigma_2\,.
 \eeq
From eq.~\eqref{eq:phi4_def}, we know that there is a subgroup $\widetilde{\Gamma}_{\!M}^+\subseteq \SL(2,\RR)\times \SL(2,\RR)$ such that $\varphi_{n,4}(\widetilde{\Gamma}_{\!M}^+)=\widetilde{G}^+_{\!M}\subseteq\widetilde{\OO}{}^+(\T_n)$, where $\phi_{n,4} = \varphi_{\Delta_{-n,1}}\circ \phi_4$ is the composition of the maps $\varphi_{\Delta_{-n,1}}$ and $\phi_4$ defined in eqs.~\eqref{eq:varphi_delta_def} and~\eqref{eq:phi4_def}. 

The Hodge-Riemann bilinear relation in eq.~\eqref{eq:HRK3} implies that
\beq\label{eq:modular_modular_2}
2n\,(\Imag t_1)\,(\Imag t_2) > 0\,.
\eeq
Hence, $\Imag t_1$ and $\Imag t_2$ are non zero and have the same sign, and we find
\beq
D = D_+ \cup \overline{D}_+\textrm{~~~and~~~}D_+ = \HH\times\HH\,.
\eeq
and so we can identify our moduli space with the quotient
\beq
\cM \simeq \mfaktor{\widetilde{\Gamma}_{\!M}^+}{(\HH\times\HH)}\,.
\eeq
An easy calculation shows that
 the generalised M\"obius transformation and the automorphic factor for the orthogonal group in eqs.~\eqref{eq:generalised_Moebius} and eq.~\eqref{eq:automorphic} now take the form (with $t=(t_2,t_1)$ and $(\gamma_2,\gamma_1)\in\widetilde{\Gamma}_{\!M}^+$))
\beq\bsp\label{eq:automorphic_SL2xSL2}
\phi_{n,4}(\gamma_2,\gamma_1)\cdot t &\,= \left(\tfrac{a_2t_2+b_2}{c_2 t_2+d_2},\tfrac{a_1t_1+b_2}{c_1 t_1+d_1}\right)^{\!T}\,,\\
 j(\phi_{n,4}(\gamma_2,\gamma_1),t) &\,= (c_2t_2+d_2)(c_1t_1+d_1)=j_{\textrm{ell}}(\gamma_2,t_2)j_{\textrm{ell}}(\gamma_1,t_1)\,.
\esp\eeq
We would like to conclude from the previous relations that the holomorphic period and the mirror map are modular forms in two variables for some subgroup $\Gamma
_2\times\Gamma_1\subseteq\widetilde{\Gamma}_{\!M}^+$ (see appendix~\ref{app:modular} for details). 
At this point we need to make a comment. We know that $\widetilde{\Gamma}_{\!M}^+$ is a subgroup of $\SL(2,\RR)\times \SL(2,\RR)$. Clearly, any direct product $\Gamma_2\times\Gamma_1$, with $\Gamma_1,\Gamma_2\subseteq\SL(2,\RR)$, is a a subgroup of $\SL(2,\RR)\times \SL(2,\RR)$, but not all subgroups have this form. For example, if $\Gamma_1$ is a subgroup of $\SL(2,\RR)$ and $\sigma:\Gamma_1\to \SL(2,\RR)$ is an embedding, then $\Gamma=\{(\gamma,\sigma(\gamma)): \gamma\in\Gamma_1\}$ is a subgroup of $\SL(2,\RR)\times \SL(2,\RR)$ (isormorphic to $\Gamma_1$). A complete description of the subgroups in terms of those of $\SL(2,\RR)$ is given by Goursat's lemma (see appendix~\ref{app:groups}). In the following we argue that in the case we are considering, we can always identify a subgroup that is a direct product.

We start by discussing the case $n=1$, which was analysed in detail in ref.~\cite{doranclingher1}, and it was shown that there is an isomorphism~\cite{doranclingher1,Hosono:2002yb}
\beq
\widetilde{\OO}{}^+(\T_1) \simeq\SL(2,\ZZ)\times \SL(2,\ZZ)\rtimes \ZZ_2\,,
\eeq
where the action of $\ZZ_2$ simply exchanges $t_1$ and $t_2$. In appendix~\ref{sec:congruence} we show that $\varphi_{\Delta_{-n,1}}\!\big(\cD(\T_1(n))\big)$ is a congruence subgroup of $\cD(\T_n)$, and so we have the following inclusion of finite-index subgroups:
\beq
\varphi_{\Delta_{-n,1}}\!\big(\cD(\T_1(n))\big) \subseteq \cD(\T_n) \subseteq \widetilde{\OO}{}^+(\T_n) \,.
\eeq
By direct calculation, we see that $\phi_{n,4}(\Gamma_2(n)\times \Gamma_2(n))\subseteq \cD(\T_1(n))$, where $\Gamma_2(n)$ was defined in eq.~\eqref{eq:Gamma_2n_def}. We define
\beq
\widetilde{\Gamma}_{\!\cD}^+ := \widetilde{\Gamma}_{\!M}^+\cap(\Gamma_2(n)\times \Gamma_2(n))\,,
\eeq
We see that $ \widetilde{\Gamma}_{\!M}^+$ contains a finite-index subgroup of $\SL(2,\ZZ)\times\SL(2,\ZZ)$. However, this subgroup may not take the form of a direct product. In appendix~\ref{app:groups} we show that, as a consequence of Goursat's lemma, $\widetilde{\Gamma}_{\!\cD}^+$ always contains a finite-index subgroup of the form $\widetilde{\Gamma}_{\!\cD,2}^+\times \widetilde{\Gamma}_{\!\cD,1}^+$, where $\widetilde{\Gamma}_{\!\cD,i}^+$ are finite-index subgroups of $\SL(2,\ZZ)$. This fact combined with eq.~\eqref{eq:automorphic_SL2xSL2} immediately shows that the holomorphic period and the mirror map are respectively a modular form in two variables of weights $(1,1)$ and a modular function for the subgroup $\widetilde{\Gamma}_{\!\cD,2}^+\times \widetilde{\Gamma}_{\!\cD,1}^+$. In appendix~\ref{app:modular} we show that any such modular form or function necessarily factorises into product of two ordinary, single-variable, modular forms for $\widetilde{\Gamma}_{\!\cD,1}^+$ and $\widetilde{\Gamma}_{\!\cD,2}^+$. 
We can summarise this by saying that if $\T_n = H\oplus H(n)$, then the period and the mirror map can be expressed in terms of ordinary modular forms. This agrees with the analysis for $n=1$ in ref.~\cite{doranclingher,doranclingher1}.

\subsubsection{Case 2: Hilbert modular forms}
\label{sec:Hilbert}
We now consider the situation of a transcendental lattice
\beq
\T_n = H \oplus\langle 2n\rangle\oplus \langle -2dn\rangle\,,
\eeq
which corresponds to $S = \begin{mymatrix} 2n & 0\\0&-2dn\end{mymatrix}$, where $d$ and $n$ are positive integers, and we assume $d$ squarefree. Note that this lattice is not equivalent to the one in eq.~\eqref{eq:T=H+H}, because it is not possible to rotate the Gram matrices into each via a unimodular transformation. Instead, we have
\beq
(R^{-1})^T \begin{mymatrix} 2n & 0\\0&-2dn\end{mymatrix}R^{-1} =  \begin{mymatrix} 0 & n\\n&0\end{mymatrix}\,,\textrm{~~~with~~~} R =  \begin{mymatrix} 1 &-\sqrt{d}\\1&\phantom{-}\sqrt{d}\end{mymatrix}\,.
\eeq
The appearance of $\sqrt{d}$ can be understood by noting that the discriminant of $\T_n$ is $4dn^2$, which is not a perfect square if $d$ is squarefree, while the discriminant of the lattice in eq.~\eqref{eq:T=H+H} is a perfect square.
We now discuss the orthogonal group for this lattice, and the ensuing modular properties for the periods and the mirror map. We proceed using exactly the same steps as in the previous case. We will therefore not discuss all steps in detail, but only focus on the main differences.

The orthogonal group for the lattice $\T_n$ is connected to the so-called Hilbert modular group. For a pedagogical introduction, see ref.~\cite{Bruinier2008}. Let $F:=\QQ(\sqrt{d})$, and consider its subring of integers,
\beq
\ord_F = \left\{\begin{array}{ll}\ZZ +\tfrac{1+\sqrt{d}}{2}\ZZ\,,& \textrm{if } d=1\!\!\!\!\mod 4\,,\\
\ZZ +\sqrt{d}\ZZ\,,& \textrm{if } d=2,3\!\!\!\!\mod 4\,. 
\end{array}\right.
\eeq
The Hilbert modular group is the group $\SL(2,\ord_F)$. For $d\neq 1\!\!\!\mod 4$ (which will be sufficient to understand the three-loop banana integrals), we can describe the Hilbert modular group as the group of matrices with unit determinant of the form
\beq
\gamma=\begin{mymatrix}
a_1&b_1\\c_1&d_1\end{mymatrix} +\sqrt{d} \begin{mymatrix}a_2&b_2\\c_2&d_2\end{mymatrix} \,,\quad \begin{mymatrix}a_i&b_i\\c_i&d_i\end{mymatrix} \in \ZZ^{2\times2}\,.
\eeq
Clearly $\SL(2,\ord_F)$ is a subgroup of $\SL(2,\RR)$. In the following it will be important that $\SL(2,\ord_F)$ can be embedded into  $\SL(2,\RR)$ in two different ways. The first embedding is simply the identity, whereas the second one is given by conjugation:
\beq
\sigma: \SL(2,\ord_F) \to \SL(2,\RR); \quad\gamma \mapsto \overline{\gamma}\,,
\eeq
where the conjugate $\overline{\gamma}$ is obtained from $\gamma$ by changing $\sqrt{d}\to -\sqrt{d}$. We can therefore identify $\SL(2,\ord_F)$ with a subgroup of $\SL(2,\RR)\times\SL(2,\RR)$:
\beq
\SL(2,\ord_F)\simeq\big\{(\gamma,
\overline{\gamma}): \gamma\in \SL(2,\ord_F)\big\}\subset \SL(2,\RR)\times\SL(2,\RR)\,.
\eeq
Under $\phi_{4}$, we can then identify (the double cover of) $\cD(\T_1)$ with $\SL(2,\ord_F)$ (cf.,~e.g., ref.~\cite{Bruinier2008,hauffe-waschbusch_hilbert_2022}). The full group $\widetilde{\OO}{}^+(\T_1)$ is obtained by adding the $\phi_{4}$ images of appropriate Atkin-Lehner involutions~\cite{hauffe-waschbusch_hilbert_2022}. The latter do not play any role in the following, so we will not define them explicitly. We define $\psi_{n,4}(\gamma) := \phi_{n,4}(\gamma,\overline\gamma)$.

Let us define $\tau=Rt$, or more explicitly
\beq\bsp
\tau_2 &\,= t_2-t_1\sqrt{d}\,,\\
\tau_1 &\,= t_2+t_1\sqrt{d}\,.
\esp\eeq
The generalised M\"obius transformation becomes
\beq
\psi_{n,4}(\gamma)\cdot t = \big(\gamma\cdot\tau_2,\overline{\gamma}\cdot \tau_1\big)\,,
\eeq
where the action of $\SL(2,\ord_F)$ on $\tau_i$ is simply by ordinary M\"orbius transformation. The automorphic factor for the orthogonal group reduces to
\beq
j(\psi_{n,4}(\gamma),t) = j_{\textrm{ell}}(\gamma,\tau_2)j_{\textrm{ell}}(\overline{\gamma},\tau_1)\,.
\eeq
The Hodge-Riemann bilinear relation in eq.~\eqref{eq:HRK3} again implies that
\beq\label{eq:modular_modular_2}
2n\,(\Imag \tau_1)\,(\Imag \tau_2) > 0\,,
\eeq
and so by the same argument as in the previous case, the period domain $D_+$ can be identified as $\HH\times\HH$, provided we work in the variables $\tau$. We can of course interpret the holomorphic period and the mirror map as functions of $\tau$, and we define
\beq\bsp
\widetilde{\Pi}_0(\tau) &\,=  \Pi_0(Rt)\,,\\
\tilde{z}(\tau) &\,=  z(Rt)\,.
\esp\eeq
We use again the result that $\cD(\T_1(n))$ is a congruence subgroup of $\cD(\T_n)$ (See appendix~\ref{sec:congruence}). An explicit description of $\cD(\T_1(n))$ as the image under $\psi_{n,4}$ of a congruence subgroup of the Hilbert modular group can be found in ref.~\cite[Corollary 2]{hauffe-waschbusch_hilbert_2022}. More specifically, in ref.~\cite{hauffe-waschbusch_hilbert_2022} it is shown that
\beq
\cD(\T_1(n)) = \psi_{n,4}\big(\Gamma_{\ord_F,n}\big)\,,
\eeq
where $\Gamma_{\ord_F,n}$ is the congruence subgroup,
\beq
\Gamma_{\ord_F,n} = \big\{\gamma\in\SL(2,\ord_F): \gamma = \varepsilon\,\mathds{1}\!\!\!\!\mod n\ord_F, \,\,\varepsilon\in\ZZ,\,\,\varepsilon^2=1\!\!\!\!\mod n\big\}\,.
\eeq

Let $\widetilde{\Gamma}_{\!M}^+$ be the subgroup of $\SL(2,\ord_F)$ such that $\psi_{n,4}\big(\widetilde{\Gamma}_{\!M}^+\big) = \widetilde{G}_{\!M}^+$. Since $\Gamma_{\ord_F,n}$ is a congruence subgroup, it has finite index in $\SL(2,\ord_F)$. We define
\beq
\widetilde{\Gamma}_{\!\cD}^+ := \widetilde{\Gamma}_{\!M}^+\cap \Gamma_{\ord_F,n}\,.
\eeq
Since $\widetilde{\Gamma}_{\!M}+$ is infinite and has finite index, $\widetilde{\Gamma}_{\!\cD}^+$ is non trivial, and it has finite index in both $\widetilde{\Gamma}_{\!M}^+$ and $\SL(2,\ord_F)$, because it is the intersection of finite-index subgroups (see appendix~\ref{app:groups}). In fact, $\widetilde{\Gamma}_{\!\cD}^+$ is a congruence subgroup of $\SL(2,\ord_F)$, because the congruence subgroup problem has a positive answer for $\SL(2,\ord_F)$ (see appendix~\ref{app:groups}). Then, for all $\gamma\in\widetilde{\Gamma}_{\!\cD}^+$, we have
\beq\bsp
\Pi_0(\psi_{n,4}(\gamma)\cdot t) &\,= \widetilde{\Pi}_0(\gamma\cdot\tau_2,\overline{\gamma}\cdot\tau_1)\\
&\, = j_{\textrm{ell}}(\gamma,\tau_2)j_{\textrm{ell}}(\overline\gamma,\tau_1)\widetilde{\Pi}_0(\tau)\\
&\,=j_{\textrm{ell}}(\gamma,\tau_2)j_{\textrm{ell}}(\overline\gamma,\tau_1)\Pi_0(t)\,,\\
z(\psi_{n,4}(\gamma)\cdot t) &\,=\tilde{z}(\gamma\cdot\tau_2,\overline{\gamma}\cdot\tau_1) = \tilde{z}(\tau) = z(t)\,.
\esp\eeq
Let us interpret this result. If $\Gamma$ is some finite-index subgroup of the Hilbert modular group, then a (holomorphic or meromorphic) \emph{Hilbert modular form of weight $(k_2,k_1)$} for $\Gamma$ is a (holomorphic or meromorphic) function $\tilde f:\HH\times\HH\to \CC$ such that for all $\gamma\in\Gamma\subseteq\SL(2,\ord_F)$
\beq
\tilde{f}(\gamma\cdot \tau_2,\overline{\gamma}\cdot\tau_1) = j_{\textrm{ell}}(\gamma,\tau_2)^{k_2}j_{\textrm{ell}}(\overline\gamma,\tau_1)^{k_1}\tilde{f}(\tau_2,\tau_1)\,.
\eeq
See, e.g.,~refs.~\cite{Bruinier2008,freitag} for an introduction to Hilbert modular forms. We then conclude that $\widetilde{\Pi}_0(\tau)$ is a holomorphic Hilbert modular form of weight $(1,1)$, and $\tilde{z}(\tau)$ is a meromorphic Hilbert modular form of weight $(0,0)$ (also known as a Hilbert modular function) for the congruence subgroup $\widetilde{\Gamma}_{\!\cD}^+$. In other words, if $\T_n=H\oplus\langle 2n\rangle\oplus\langle-2dn\rangle$, then the holomorphic period and the mirror map admit a parametrisation in terms of Hilbert modular forms.

\subsection{The isomorphism $\SO_0(2,3) \simeq \Sp(4,\mathbb{R})/\ZZ_2$}
\label{sec:B2=C2}

We now consider families of K3 surfaces of Picard rank 17. We can pick a basis of cycles such that the Gram matrix of the intersection pairing takes the form
\beq
\Sigma = \begin{mymatrix}
0&0&1\\0&S&0\\1&0&0
\end{mymatrix}\,,
\eeq
where $S$ is a symmetric integral $3\times3$ matrix of full rank with signature $(1,2)$ and with even numbers on the diagonal. Just like in section~\ref{sec:D2=A1xA1}, the form of $S$ is not unique, and we restrict ourselves to one example which will be sufficient to understand the three-loop banana integrals. We consider a transcendental lattice of the form
\beq
\T_n = H \oplus H \oplus \langle-2n\rangle\,,
\eeq
where $n>0$ is an integer. This corresponds to the choice 
\beq
S = \begin{mymatrix}
0&0&1\\0&-2n&0\\1&0&0
\end{mymatrix}\,.
\eeq

We proceed in exactly the same manner as in the previous examples, so we will be as brief as possible. We again start by analysing the case $n=1$. In refs.~\cite{Clingher2010LatticePK,HAUFFEWASCHBUSCH202122} it is shown that (a double cover of) $\cD(\T_1)$ is the Siegel modular group $\Sp(4,\ZZ)$. An explicit description of $\cD(\T_1(n))$ as a congruence subgroup of $\Sp(4,\ZZ)$ is given in ref.~\cite[Theorem 1]{HAUFFEWASCHBUSCH202122}:
\beq
\cD(\T_1(n)) = \phi_{n,5}\big(\Gamma_{\textrm{Sp},n}\big)\,,
\eeq
where we defined $\phi_{n,5}=\varphi_{\Delta_{n,2}}\circ \phi_5$ to be the composition of the maps in eqs.~\eqref{eq:varphi_delta_def} and~\eqref{eq:Sp4C_SO5C}, and $\Gamma_{\textrm{Sp},n}$ is the congruence subgroup
\beq
\Gamma_{\textrm{Sp},n} := \big\{\gamma\in\Sp(4,\ZZ):\gamma = \varepsilon\,\mathds{1}\!\!\!\!\mod n, \,\,\varepsilon\in\ZZ,\,\,\varepsilon^2=1\!\!\!\!\mod n\big\}\,.
\eeq
We again define 
\beq
\widetilde{\Gamma}_{\!\cD}^+:=\widetilde{\Gamma}_{\!M}^+\cap\Gamma_{\textrm{Sp},n}\,.
\eeq
$\widetilde{\Gamma}_{\!\cD}^+$ is non trivial, because
$\widetilde{\Gamma}_{\!M}^+$ is infinite and $\Gamma_{\textrm{Sp},n}$ has finite index, and it is a finite-index subgroup of $\Sp(4,\ZZ)$ and $\widetilde{\Gamma}_{\!M}^+$, because it is the intersection of finite-index subgroups. Moreover, since the congruence subgroup problem has a positive answer for the $\Sp(4,\ZZ)$, $\widetilde{\Gamma}_{\!\cD}^+$ is a congruence subgroup of $\Sp(4,\ZZ)$.

Let us define the symmetric matrix
\beq
\Omega_n(t) := \begin{mymatrix}\tfrac{1}{n}t_3 & t_2\\ \phantom{\tfrac{1}{n}}t_2&t_1\end{mymatrix}\,.
\eeq
The Hodge-Riemann bilinear relation in eq.~\eqref{eq:HRK3} then implies
\beq\label{eq:modular_siegel_2}
2n\,\det(\Imag\Omega_n(t)) > 0\,.
\eeq
It follows that $\det(\Imag\Omega_n(t))>0$, and so $\Imag\Omega_n(t)$ must be either positive or negative definite. We can then identify the connected component $D_+$ of the period domain with the Siegel upper half-space:
\beq
D_+ = \HH_2 = \big\{\omega\in\CC^{2\times2}: \omega^T=\omega\textrm{ and $\Imag\omega$ is positive definite}\}\,.
\eeq

Consider now an element $\gamma=\begin{mymatrix}A&B\\C&D\end{mymatrix}\in\widetilde{\Gamma}_{\!\cD}^+$. The generalised M\"obius transformation and the automorphic factor for the orthogonal group reduce to the standard action and automorphic factor of $\Sp(4,\ZZ)$ on $\HH_2$:
\beq\bsp
\phi_{n,5}(\gamma)\cdot t &\,= (A\Omega_n(t)+B)(C\Omega_n(t)+D)^{-1}\,,\\
j(\phi_{n,5}(\gamma), t) &\, = \det(C\Omega_n(t)+D)\,.
\esp\eeq
Putting everything together, and defining the following two functions on $\HH_2$,
\beq\bsp
\widetilde{\Pi}_0(\Omega) &\,= \Pi_0(\Omega_{22},\Omega_{12},n\Omega_{11})\,,\\
\tilde{z}(\Omega) &\,= z(\Omega_{22},\Omega_{12},n\Omega_{11})\,,
\esp\eeq
we see that the holomorphic period and the mirror map transform as
\beq\bsp\label{eq:Siegel_trafo}
\Pi_0(\phi_{n,5}(\gamma)\cdot t) &\,= \widetilde{\Pi}_0\big((A\Omega_n(t)+B)(C\Omega_n(t)+D)^{-1}\big)\\
&\, = \det(C\Omega_n(t)+D)\widetilde{\Pi}_0(\Omega_n(t))\\
&\, = \det(C\Omega_n(t)+D)\Pi_0(t)\,,\\
z(\phi_{n,5}(\gamma)\cdot t) &\,= \tilde{z}_0\big((A\Omega_n(t)+B)(C\Omega_n(t)+D)^{-1}\big) = \tilde{z}_0(\Omega_n(t)) =z(t)\,.
\esp\eeq
 In other words, we see that the holomorphic period and the mirror map, when seen as functions on the Siegel upper half-space $\HH_2$, transform as classical Siegel modular forms of weight 1 and 0 for the congruence subgroup $\widetilde{\Gamma}_{\!\cD}^+$ respectively. This is in agreement with the result of ref.~\cite{Clingher2010LatticePK}. 
 
 Let us conclude with a comment. Consider the a family of K3 surfaces with transcendental lattice 
 \beq
 \T_n = H\oplus H(n)\oplus \langle-2n\rangle\,.
 \eeq
 we can repeat exactly the same steps as before (up to replacing $\varphi_{\Delta_{n,2}}$ by $\varphi_{\Delta_{n,1}}$), and we see that exactly the same conclusion holds, and also in this case the holomorphic period and the mirror map can be expressed in terms of Siegel modular forms.
 
 
 \subsection{The isomorphism $\SO_0(2,4)\simeq \SU(2,2,\RR)/\ZZ_2$}
 \label{sec:D3=A3}
Finally, let us discuss an exceptional isomorphism for $m=4$, which will be the one encountered for the three-loop banana integrals. Consider a family of K3 surfaces with transcendental lattice
\beq
\T_n = H\oplus H \oplus A_2(-n)\,,
\eeq
where $ A_2(-n)$ is the rescaled lattice generated by the simple roots of the Dynkin diagram $A_2$. The Gram matrix of $\T_n$ is
\beq
\Sigma = \begin{mymatrix}
\ph0&\ph0&\ph0&\ph0&\ph1\\
\ph0&\ph0&\ph0&\ph1&\ph0\\
\ph0&\ph0&\ph S&\ph0&\ph0\\
\ph0&\ph1&\ph0&\ph0&\ph0\\
\ph1&\ph0&\ph0&\ph0&\ph0\\
\end{mymatrix}
\eeq
with $S=-n\Sigma_{A_2}$, and $\Sigma_{A_2} = \begin{mymatrix}\ph 2 & -1\\ -1&\ph2\end{mymatrix}$ is the Cartan matrix of $A_2$. 

We proceed just like in the previous cases. We start by discussing the case $n=1$, in which case we recover the Gram matrix in eq.~\eqref{eq:Sigma_4}. The explicit expression for the map $\phi_6$ in eq.~\eqref{eq:SO(2,6)_sequence} can be found in refs.~\cite{hermitian_thesis,HAUFFEWASCHBUSCH202122}, where it was also shown that $\cD(T_1)$ is the image under $\phi_6$ of the group $\SU(2,2,\ord_F)$ defined in eq.~\eqref{eq:SU22_def}. If we consider the composition $\phi_{n,6} = \varphi_{\Delta_{n,2}}\circ\phi_6$, then $\cD(T_1(n)) = \phi_{n,6}(\Gamma_{\SU,n})$ is a congruence subgroup of $\cD(\T_1)$, with $\Gamma_{\SU,n}$ a congruence subgroup of $\SU(2,2,\ord_F)$~\cite{HAUFFEWASCHBUSCH202122}:
\beq
\Gamma_{\SU,n}:= \big\{\gamma\in\SU(2,2,\ord_F):\gamma = \varepsilon\,\mathds{1}\!\!\!\!\mod n\ord_F, \,\,\varepsilon\in\ZZ,\,\,\varepsilon^2=1\!\!\!\!\mod n\big\}\,.
\eeq
Just like in the previous cases, using this subgroup we can again construct a non-trivial finite-index subgroup of $\SU(2,2,\ord_F)$:
\beq
\widetilde{\Gamma}_{\!\cD}^+ := \widetilde{\Gamma}_{\!M}^+\cap\Gamma_{\SU,n}\,.
\eeq
Note that every finite-index subgroup of $\SU(2,2,\ord_F)$ is a congruence subgroup, so that $\widetilde{\Gamma}_{\!\cD}^+$ is a congruence subgroup.

Let us define
\beq
\Omega_n(t) := \begin{mymatrix}\tfrac{1}{n}t_4 & t_2+\omega t_3 \\ t_2+\overline{\omega} t_3 &t_1\end{mymatrix}\,,
\eeq
where $\omega:=-\tfrac{1}{2}+i\tfrac{\sqrt{3}}{2}$, and $\overline{\omega}$ is its complex conjugate. The Hodge-Riemann bilinear relation in eq.~\eqref{eq:HRK3} implies
\beq\label{eq:modular_hermitian_2}
2n\,\det(\him(\Omega_n(t))) > 0\,,
\eeq
where we defined
\beq
\him(\Omega) := \tfrac{1}{2i}\left(\Omega- \Omega^\dagger\right)\,.
\eeq
It follows that $\det(\him(\Omega_n(t))>0$, and so $\him(\Omega_n(t))$ must be either positive or negative definite. We then see that the period domain $D_+$ can be identified with the hermitian half-space of degree 2,
\beq
D_+=\cH_2 = \big\{\omega\in\CC^{2\times2}: \him(\omega) \textrm{ positive definite}\big\}\,.
\eeq
Then, if $\gamma=\begin{mymatrix}A&B\\C&D\end{mymatrix}\in\widetilde{\Gamma}_{\!\cD}^+$, the generalised M\"obius transformation and the automorphic factor for the orthogonal group reduce to:
\beq\bsp
\phi_{n,6}(\gamma)\cdot t &\,= (A\Omega_n(t)+B)(C\Omega_n(t)+D)^{-1}\,,\\
j(\phi_{n,6}(\gamma), t) &\, = \det(C\Omega_n(t)+D)\,.
\esp\eeq
In order to interpret these results, we define the following two functions on $\cH_2$,
\beq\bsp
\widetilde{\Pi}_0(\Omega) &\,= \Pi_0\big(\widetilde{\Omega}\big)\,,\\
\tilde{z}(\Omega) &\,= z\big(\widetilde{\Omega}\big)\,,
\esp\eeq
where we defined 
\beq
\widetilde{\Omega} := \left(\Omega_{22},-\tfrac{\overline\omega}{\omega-\overline{\omega}}\Omega_{12}+\tfrac{{\omega}}{\omega-\overline{\omega}}\Omega_{21},\tfrac{1}{\omega-\overline{\omega}}\Omega_{12}-\tfrac{1}{\omega-\overline{\omega}}\Omega_{21},n\Omega_{11}\right)\,.
\eeq
The holomorphic period and the mirror map then transform as
\beq\bsp
\Pi_0(\phi_{n,6}(\gamma)\cdot t) &\,= \widetilde{\Pi}_0\big((A\Omega_n(t)+B)(C\Omega_n(t)+D)^{-1}\big)\\
&\, = \det(C\Omega_n(t)+D)\widetilde{\Pi}_0(\Omega_n(t))\\
&\, = \det(C\Omega_n(t)+D)\Pi_0(t)\,,\\
z(\phi_{n,6}(\gamma)\cdot t) &\,= \tilde{z}_0\big((A\Omega_n(t)+B)(C\Omega_n(t)+D)^{-1}\big) = \tilde{z}_0(\Omega_n(t)) =z(t)\,.
\esp\eeq
We see that the previous equation is identical to the transformation properties of a Siegel modular form in eq.~\eqref{eq:Siegel_trafo}. However, it would be wrong to conclude that also in this case we obtain Siegel modular forms, because the functions are defined on the hermitian half-space $\cH_2$ rather than the Siegel half-space $\HH_2$. If $\Gamma$ is a finite-index subgroup of $\SU(2,2,\ord_F)$, then a (holomorphic, meromorphic) function $\tilde{f} :\cH_2 \to\CC$ is called a \emph{hermitian modular form} if ~\cite{c9e0d978-3042-3cf9-a67c-48eb598b7003}
\beq
\tilde{f}\big((A\Omega+B)(C\Omega+D)^{-1}\big) = \det(C\Omega+D)^k\,\tilde{f}(\Omega)\,,
\textrm{~~~for all }\gamma=\begin{mymatrix}A&B\\C&D\end{mymatrix}\in\Gamma\,.
\eeq
We then see that the holomorphic period and the mirror map, when seen as functions on the hermitian half-space $\HH_2$, transform as hermitian modular forms (of weight 1 and 0) for the congruence subgroup $\widetilde{\Gamma}_{\!\cD}^+$. This result is in agreement with the result for $n=1$ in ref.~\cite{Nagano:2024aa}. Finally, we mention that, using exactly the same arguments, we see that the same conclusions hold for a family of K3 surfaces with the transcendental lattice
\beq
\cT = H\oplus H(n)\oplus A_2(-n)\,.
\eeq



\section{Automorphic properties of three-loop banana integrals}
\label{sec:bananas}

\subsection{Review of banana integrals}
\label{eq:review_banana}

In this section we apply the mathematical concepts from the previous section to the three-loop banana integrals in $D=2$ dimensions (see figure~\ref{fig:banana}),
\beq
I(p^2,m_1^2,m_2^2 ,m_3^2 ,m_4^2 ) = -\frac{1}{\pi^3}\int\prod_{j=1}^3\left(\frac{\rd^D k_j}{k_j^2-m_j^2}\right)\frac{1}{(p-k_1-k_2-k_3)^2-m_4^2}\,.
\eeq
Note that the only non-trivial functional dependence is in the ratios $z_i=m_i^2/p^2$, and we can put $p^2=1$ without loss of generality.
These integrals have been studied extensively in the literature. Most of the studies, however, have focused on the equal-mass case $m_1=m_2=m_3=m_4$~\cite{Bloch:2014qca,MR3780269,Primo:2017ipr,Broedel:2019kmn,Broedel:2021zij,Pogel:2022yat,Mishnyakov:2023wpd,Cacciatori:2023tzp,Mishnyakov:2023sly,Mishnyakov:2024rmb,delaCruz:2024xit}, and there are only very few studies of banana integrals depending on different masses~\cite{Klemm:2019dbm,Bonisch:2020qmm,Bonisch:2021yfw,Kreimer:2022fxm}.

\begin{figure}[!th]
\begin{center}
\includegraphics[scale=0.25]{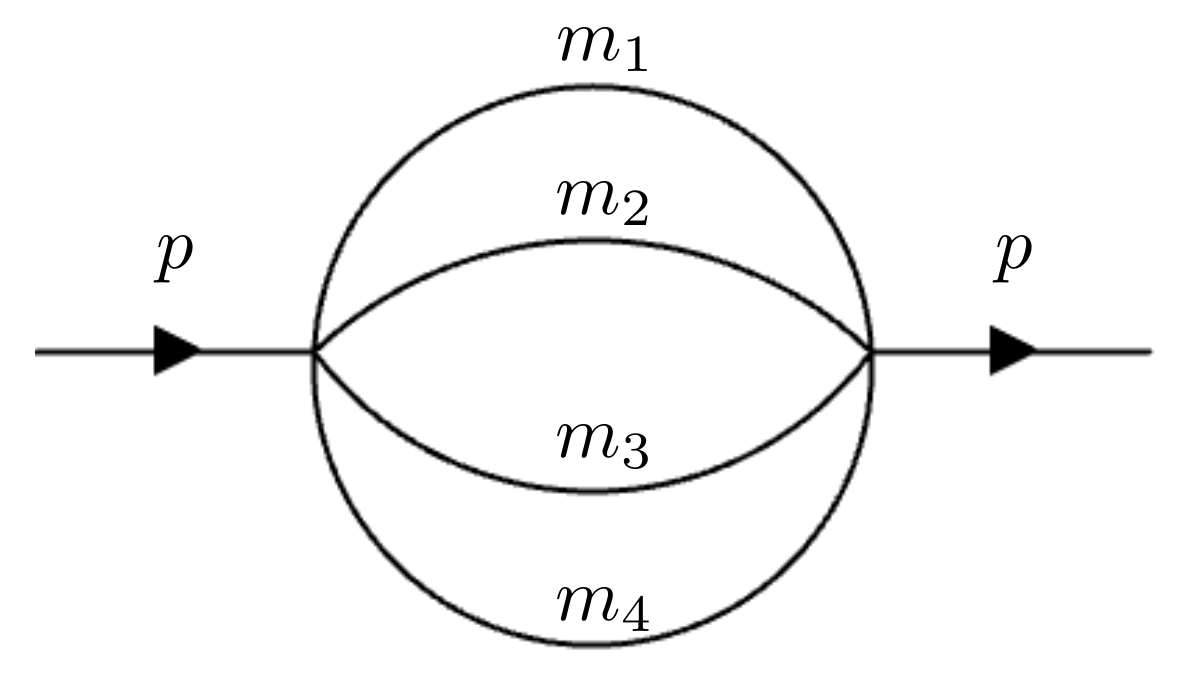}
\caption{\label{fig:banana}The three-loop banana integral.}
\end{center}
\end{figure}

It is known that, independently of the mass configuration, the maximal cuts of the three-loop banana integral in tow dimensions obtained by putting all four propagators on shell compute the periods of some family of K3 surfaces. This can for example be seen by analysing the parametric or Baikov representations of the integral. The resulting family of K3 surfaces is a family of hyperplane sections in a toric ambient space. It was pointed out~\cite{Kerr,Bonisch:2020qmm} that there is another family of K3 surfaces, realised as a complete intersection in weighted projective space, that has the same periods. The advantage of this model is that the number of free moduli of this family is four, and so it agrees with the number of independent mass parameters. This four-parameter family has a MUM-point at $z_i = m_i^2=0$. Close to the MUM-point the holomorphic and the single-logarithmic periods admit the expansion~\cite{Bonisch:2020qmm},
\beq\bsp
\Pi_0(z) &\,= 4\pi^2\Big[1+(z_1+z_2+z_3+z_4) + (z_1^2+4 z_2 z_1+4 z_3 z_1+4 z_4 z_1+z_2^2+z_3^2+z_4^2\\
&\,+4 z_2 z_3+4 z_2 z_4+4 z_3 z_4)+\ord(z_i^3)\Big]\,,\\
\Pi_{1i}(z)&\,=\Pi_0(z)\frac{\log z_i}{-2\pi i} + 4\pi i(z_1+z_2+z_3+z_4-z_i) + \ord(z_i^2)\,.
\esp\eeq
The concrete expressions will not be required in the following. The double-logarithmic period $\Pi_2(z)$ can then be obtained as the solution to the bilinear relation
\beq\label{eq:banana_HR}
\uPi(z)^T\Sigma_{4}\uPi(z) = 0\,,
\eeq
where 
\beq\label{eq:P_vec_banana}
\uPi(z) = \big(\Pi_2(z),\Pi_{14}(z),\Pi_{13}(z),\Pi_{12}(z),\Pi_{11}(z),\Pi_{0}(z)\big)^T\,, 
\eeq
and the intersection form is given by
\beq\label{eq:Sigma_ban}
\Sigma_{4} = \begin{mymatrix}
0&0&0&0&0&1\\
0&0&1&1&1&0\\
0&1&0&1&1&0\\
0&1&1&0&1&0\\
0&1&1&1&0&0\\
1&0&0&0&0&0
\end{mymatrix}\,.
\eeq
This Gram matrix determines the transcendental lattice $\T=H\oplus \Lambda$ of discriminant $\det\Sigma=-3$. At first glance, the lattice does not seem to match any of the special case cases discussed in section~\ref{sec:isom_K3}. In the remainder of this section, we show that the different possible mass configurations precisely match the cases discussed in section~\ref{sec:isom_K3}.


\subsection{The equal-mass case $m_1=m_2=m_3=m_4$}
We start by discussing the equal-mass case. The four single-logarithmic periods $\Pi_{1,i}(z)$ become equal, so that we only need to consider the three-dimensional vector of periods 
\beq
\uPi^{(1)}(z) = \big(\Pi_2(z),\Pi_{11}(z),\Pi_0(z)\big)^T_{|z_1=z_2=z_3=z_4}\,.
\eeq
The bilinear relation~\eqref{eq:banana_HR} reduces to
\beq
\uPi^{(1)}(z)^T\Sigma_{1}\uPi^{(1)}(z) = 0\,,\textrm{~~~with~~~}\Sigma_1 = \begin{mymatrix}0&0&1\\0&12&0\\1&0&0\end{mymatrix}\,.
\eeq
We see that the transcendental lattice in the equal-mass case is $\T = H\oplus\langle12\rangle$. This is not really a surprise, because it is well known that the maximal cuts of the equal-mass banana integrals are products of maximal cuts of the sunrise integral~\cite{Bloch:2014qca,MR3780269,Primo:2017ipr,ABE1973348,verrill1996}, and we see that we immediately recover this result from our analysis. In particular, we conclude that the periods and the mirror map admit a modular parametrisation, as expected.

\subsection{The three-equal-mass case $m_2=m_3=m_4$}
Next we analysis the case where three masses are equal, say $m_2=m_3=m_4$. Three single-logarithmic periods become equal, and we need to consider the period vector
\beq
\uPi^{(2a)}(z) = \big(\Pi_2(z),\Pi_{12}(z),\Pi_{11}(z),\Pi_0(z)\big)^T_{|z_2=z_3=z_4}\,,
\eeq
 which satisfies the bilinear relation
 \beq
\uPi^{(2a)}(z)^T\Sigma_{2a}\uPi^{(2a)}(z) = 0\,,\textrm{~~~with~~~}\Sigma_{2a} = \begin{mymatrix}0&\ph0&\ph0&\ph1\\0&\ph6&\ph3&\ph0\\0&\ph3&\ph0&\ph0\\1&\ph0&\ph0&\ph0\end{mymatrix}\,.
\eeq
At first glance, the Gram matrix does not correspond to any of the Gram matrices discussed in section~\ref{sec:isom_K3}. However, we may use the freedom to redefine the basis of periods via a $\GL(4,\ZZ)$ transformation that preserves the Hodge structure and the intersection product. Let us define a new basis of periods $\uPi^{\textrm{new}}(z)$ via
\beq
\uPi^{(2a)}(z) = R_{2a}\uPi^{(2a)\textrm{new}}(z)\,,\textrm{~~~with~~~}R_{2a} = \begin{mymatrix}1&\ph0&\ph0&\ph0\\
0&\ph0&-1&\ph0\\0&-1&\ph1&\ph0\\0&\ph0&\ph0&\ph1\end{mymatrix}\,.
\eeq
Note that $\det R_{2a} = 1$, and so $R_{2a}\in \GL(4,\ZZ)$, and the block-diagonal structure implies that this change of basis respects the Hodge structure.
We then find
\beq
R_{2a}^T\Sigma_{2a}R_{2a} = \begin{mymatrix} 0&\ph0&\ph0&\ph1\\0&\ph0&\ph3&\ph0\\0&\ph3&\ph0&\ph0\\1&\ph0&\ph0&\ph0\end{mymatrix}\,.
\eeq
In this basis it is manifest that the transcendental lattice is $\T=H\oplus H(3)$, and so it precisely matches the type of families of K3 surfaces studied in section~\ref{sec:ExE}. We can therefore immediately conclude that the periods and the mirror map admit a modular parametrisation in terms of ordinary modular forms. The modular expressions for the holomorphic period and mirror map will be presented elsewhere~\cite{inprep}.
 
\subsection{The pairwise-equal-mass case $m_1=m_3$ and $m_2=m_4$}
We now discuss the case where the masses are pairwise equal, say $m_1=m_3$ and $m_2=m_4$. The period vector is 
\beq
\uPi^{(2b)}(z) = \big(\Pi_2(z),\Pi_{12}(z),\Pi_{11}(z),\Pi_0(z)\big)^T_{|z_1=z_3,z_2=z_4}\,.
\eeq
It satisfies the bilinear relation
 \beq
\uPi^{(2b)}(z)^T\Sigma_{2b}\uPi^{(2b)}(z) = 0\,,\textrm{~~~with~~~}\Sigma_{2b} = \begin{mymatrix}0&\phantom{-}0&\phantom{-}0&\phantom{-}1\\0&\ph2&\ph4&\phantom{-}0\\0&\ph4&\ph2&\phantom{-}0\\1&\phantom{-}0&\phantom{-}0&\phantom{-}0\end{mymatrix}\,.
\eeq
We can define a new basis via the $\GL(4,\ZZ)$ transformation
\beq
\uPi^{(2b)}(z) = R_{2b}\uPi^{(2b)\textrm{new}}(z)\,,\textrm{~~~with~~~}R_{2b} = \begin{mymatrix}1&\ph0&\ph0&\ph0\\
0&\ph0&-1&\ph0\\0&-1&\ph2&\ph0\\0&\ph0&\ph0&\ph1\end{mymatrix}\,,
\eeq
and we find
\beq
R_{2b}^T\Sigma_{2b}R_{2b} = \begin{mymatrix} 0&\ph0&\ph0&\ph1\\0&\ph2&\ph0&\ph0\\0&\ph0&-6&\ph0\\1&\ph0&\ph0&\ph0\end{mymatrix}\,.
\eeq
We see that the transcendental lattice is $\T = H\oplus\langle2\rangle\oplus\langle-6\rangle$, which corresponds to the case studied in section~\ref{sec:Hilbert}. In particular, this implies that the holomorphic period and the mirror map are Hilbert modular forms.


\subsection{The two-equal-mass case $m_2=m_3$}
In the case of two equal masses $m_2=m_3$, the vector of periods reduces to
\beq
\uPi^{(3)}(z) = \big(\Pi_2(z),\Pi_{14}(z),\Pi_{12}(z),\Pi_{11}(z),\Pi_0(z)\big)^T_{|z_3=z_2}\,.
\eeq
It satisfies the bilinear relation
\beq
\uPi^{(3)}(z)^T\Sigma_{3}\uPi^{(3)}(z) = 0\,,\textrm{~~~with~~~}\Sigma_{3} = \begin{mymatrix}0&\phantom{-}0&\phantom{-}0&\ph0&\phantom{-}1\\
0&\ph0&\ph2&\ph1&\phantom{-}0\\0&\ph2&\ph2&\ph2&\phantom{-}0\\
0&\ph1&\ph2&\ph0&\phantom{-}0\\1&\phantom{-}0&\phantom{-}0&\phantom{-}0&\ph0\end{mymatrix}\,.
\eeq
We define the new basis 
\beq
\uPi^{(3)}(z) = R_{3}\uPi^{(3)\textrm{new}}(z)\,,\textrm{~~~with~~~}R_{3} = \begin{mymatrix}1&\ph0&\ph0&\ph0&\ph0\\
0&\ph0&-2&-1&\ph0\\
0&\ph0&\ph1&\ph0&\ph0\\
0&-1&-2&\ph0&\ph0\\0&\ph0&\ph0&\ph0&\ph1\end{mymatrix}\,,
\eeq
with $R_3\in\GL(5,\ZZ)$, and we get
\beq
R_{3}^T\Sigma_{3}R_{3} = \begin{mymatrix} 0&\ph0&\ph0&\ph0&\ph1\\
0&\ph0&\ph0&\ph1&\ph0\\0&\ph0&-6&\ph0&\ph0\\
0&\ph1&\ph0&\ph0&\ph0\\1&\ph0&\ph0&\ph0&\ph0\end{mymatrix}\,.
\eeq
We find the transcendental lattice $\T=H\oplus H\oplus \langle-6\rangle$, which was studied in section~\ref{sec:B2=C2}. Hence, the mirror map is a classical Siegel modular function, and holomorphic period is a classical Siegel modular form of weight 1. At genus two every point in the Siegel upper half-space $\HH_2$ corresponds to a curve of genus two, and every curve of genus two is hyperelliptic. As a consequence, there is a family of hyperelliptic curves of genus two such that the determinant of the matrix of $A$-cycle periods equals $\Pi_0^{(3)}$. 

\subsection{The case of four different masses}
Finally, let us return to the general case of four different masses. The basis of periods and the intersection form were given in eqs.~\eqref{eq:P_vec_banana} and~\eqref{eq:Sigma_ban}. We define the matrix
\beq
R_{4} = \begin{mymatrix}1&\ph0&\ph0&\ph0&\ph0&\ph0\\
0&\ph0&\ph1&-1&\ph0&\ph0\\
0&1&-1&\ph0&\ph0&\ph0\\
0&\ph0&-1&\ph0&\ph1&\ph0\\
0&\ph0&\ph0&\ph1&\ph0&\ph0\\
0&\ph0&\ph0&\ph0&\ph0&\ph1\end{mymatrix}\,,
\eeq
with $\det R_4=1$. We can then define the new basis
\beq
\uPi(z) = R_{4}\uPi^{\textrm{new}}(z)\,,
\eeq
and we obtain the Gram matrix
\beq
R_{4}^T\Sigma_{4}R_{4} = \begin{mymatrix} 0&\ph0&\ph0&\ph0&\ph0&\ph1\\
 0&\ph0&\ph0&\ph0&\ph1&\ph0\\
0&\ph0&\ph-2&\ph1&\ph0&\ph0\\
0&\ph0&\ph1&-2&\ph0&\ph0\\
 0&\ph1&\ph0&\ph0&\ph0&\ph0\\
 1&\ph0&\ph0&\ph0&\ph0&\ph0\end{mymatrix}\,.
\eeq
This is the Gram matrix of the lattice $\T=H\oplus H\oplus A_2(-1)$ studied in section~\ref{sec:D3=A3}, and so the holomorphic period and the mirror map are hermitian modular forms.


\section{Conclusions}
\label{sec:conclusions}

The goal of this paper was to initiate the study of Feynman integrals associated to families of K3 surfaces depending on $m$ parameters. While the case $m=1$ is by now relatively well understood, there has not been any systematic approach in physics to the cases $m\ge 2$. An important first step in understanding Feynman integrals is understanding their cuts, which in this case correspond to the periods of the K3 surface. In a first part of this paper we have reviewed the mathematical background relevant to understanding the period geometry of families of K3 surfaces, highlighting in particular the parallels to the case of families of elliptic curves, which is by now relatively well understood in physics. In particular, the modular group $\SL(2,\ZZ)$ and modular forms in the elliptic case get replaced in the K3 case by the orthogonal group of the transcendental lattice and the corresponding orthogonal modular forms. Orthogonal modular forms are an active area of research in mathematics (cf.,~e.g.,~refs.~\cite{bruinierbook,orthogonal_PhD,WANG2020107332,Wang_2021,Schaps2022FourierCO,Schaps2023}, including algorithms for computer codes~\cite{Assaf:2022aa}). For the future, it would be interesting to study in how far the mathematics of orthogonal modular forms can be leveraged to compute Feynman integrals in the same way that ordinary modular forms have become an important tool for multi-loop integrals.

In a second part of this paper we have studied examples where one can use the transcendental lattice to identify other classes of modular form that arise from K3 periods. The key observation is the well-known fact that for $m\le 4$, the orthogonal groups can be $\OO(2,m)$ are isomorphic to other real Lie groups. In this way we can relate the automorphic properties of the periods and the mirror map to those of other groups. We have worked out various examples which allow us to uncover ordinary modular forms as well as Hilbert, Siegel and hermitian modular forms. We summarise these correspondences in table~\ref{tab:summary}. We recover in this way the well-known result that the periods and the mirror map of one-parameter families of K3 admit a modular parametrisation, and our results should seen as a generalisation to families depending on more moduli. 

\begin{table}[!t]
\begin{center}\begin{tabular}{|c|c|c|c|}
\hline
$m$&$\T$ &  Modularity & $D_+$ \\
\hline\hline
1&$H\oplus\langle 2n\rangle$ &Elliptic &$\HH$ \\ 
\hline
2&$H\oplus H(n)$ & Elliptic &$\HH\times \HH$\\ 
\hline
2&$H\oplus \langle 2n\rangle\oplus \langle -2dn\rangle$ & Hilbert &$\HH\times \HH$ \\
\hline
\multirow{ 2}{*}{3}&$H\oplus H\oplus \langle -2n\rangle$ &\multirow{ 2}{*}{Siegel} &\multirow{ 2}{*}{$\HH_2$}  \\ 
&$H\oplus H(n)\oplus \langle -2n\rangle$&&\\
\hline
\multirow{ 2}{*}{4}&$H\oplus H\oplus A_2(-n)$ & \multirow{ 2}{*}{Hermitian} &\multirow{ 2}{*}{$\cH_2$} \\
&$H\oplus H(n)\oplus A_2(-n)$ &&\\
\hline
\end{tabular}
\caption{\label{tab:summary} Summary of the different transcendental lattice and their automorphic properties studied in section~\ref{sec:isom_K3}.}
\end{center}
\end{table}

Finally, we have applied these ideas to study the periods that arise from the maximal cuts of three-loop banana integrals depending on any configuration of non-zero masses. Our key observation is that in all cases the K3 periods for banana integrals can be expressed in terms of other classes of modular forms, see table~\ref{tab:banana} for a summary. For the future, it would be interesting to understand if one can identify the explicit expression for these modular forms for banana integrals, and for Feynman integrals attached to K3 surfaces more generally.

\begin{table}[!th]
\begin{center}\begin{tabular}{|c|c|c|c|}
\hline
Masses&$\T$ &  Modularity & $D_+$ \\
\hline\hline
all equal&$H\oplus\langle 12\rangle$ & elliptic &$\HH$ \\
\hline
$m_2=m_3=m_4$&$H\oplus H(3)$& elliptic &$\HH\times \HH$ \\
\hline
$m_1=m_3$&\multirow{2}{*}{$H\oplus \langle 2\rangle\oplus \langle -6\rangle$} & \multirow{2}{*}{Hilbert} &\multirow{2}{*}{$\HH\times \HH$} \\
$m_2=m_4$&&&  \\ 
\hline
$m_2=m_3$&$H\oplus H\oplus \langle -6\rangle$ & Siegel &$\HH_2$ \\ 
\hline
all different&$H\oplus H\oplus A_2(-3)$ &Hermitian &${\cH}_2$\\
\hline
\end{tabular}
\caption{\label{tab:banana} Summary of the transcendental lattices and automorphic properties of the maximal cuts of three-loop banana integrals in the basis of periods defined in section~\ref{eq:review_banana}.}
\end{center}
\end{table}

\subsection*{Acknowledgments}
We are grateful to Albrecht Klemm for discussions an to Sara Maggio for collaboration on related topics.
This work was funded by the European Union (ERC Consolidator Grant LoCoMotive 101043686). Views and opinions expressed are however those of the author(s) only and do not necessarily reflect those of the European Union or the European Research Council. Neither the European Union nor the granting authority can be held responsible for them.

\appendix

\section{Some group theory}
\label{app:groups}
In this appendix we review some standard material from group theory which plays an important role in this paper. We specifically focus on subgroups of finite index. We only consider discrete groups, which is sufficient for our purposes.

\subsection{Cosets and the index of a subgroup}
Consider a group $G$ and a subgroup $H$. The \emph{index} of $H$ in $G$ is defined to be the number of cosets of $H$ in $G$:
\beq
[G:H]  = \left|\faktor{G}{H}\right|\,.
\eeq
If $K\subseteq H$ is a subgroup, then we have
\beq\label{eq:GKH}
[G:K] = [G:H] \, [H:K]\,.
\eeq
This relation has an important consequence. Assume that $K$ has finite index in $G$, $[G:K]<\infty$. Then we see from eq.~\eqref{eq:GKH} that neither of the two factors on the right-hand side may be infinite, and so we arrive at the following conclusion:
\begin{lemma}
\emph{Let $K\subseteq H\subseteq G$ be groups, and assume that $K$ has finite index in $G$. Then $K$ has finite index in $H$, and $H$ has finite index in $G$.}
\end{lemma}
 Another useful property is the following:
\begin{lemma}\label{eq:finite_index_intersection}
\emph{Let $H, K\subseteq G$ be groups, and assume that both $H$ and $K$ have finite index in $G$. Then also $H\cap K$ has finite index in $G$, and $[H:(H\cap K)] \le [G:K]$.}
\end{lemma}

Finally, the following result plays an important role in the main text. 
\begin{lemma}
\emph{Let $H, K\subseteq G$ be groups, and assume that $H$ has finite index in $G$. If $H\cap K=\{e\}$, then $K$ is a finite group.}
\end{lemma}
Since we have not seen the proof anywhere in the literature, we include it below. We mostly need the following corollary:
\begin{corollary}
\emph{Let $H, K\subseteq G$ be groups, and assume that $H$ has finite index in $G$. If $K$ is an infinite group, then $H\cap K\neq\{e\}$.}
\end{corollary}
\begin{proof}
Let $r:=[G:H]$, and let $g_1=e,g_2,\ldots,g_r$ be a set of coset representatives of $H$, i.e., the cosets of $H$ are $g_1H=H, g_2H,\ldots$, and we have $g_iH=g_jH$ if and only if $g_i=g_j$. 

Assume that $H\cap K=\{e\}$. If $K$ is the trivial group, then the claim is obvious. Otherwise there is $g\in K$ with $g\neq e$, and there is a unique coset representative $g_i$ with $i>1$ such that $g\in g_iH$. If $|K|=2$, we are done. Else, there is $g'\in K$ with $g'\notin \{e,g\}$, and there is a unique coset representative $g_j$ such that $g'\in g_jH$. If $i=j$, then there is $h,h'\in H$ such that $g = g_ih$ and $g' = g_ih'$. This implies that 
\beq
g'g^{-1} = h'h^{-1} \in  H\cap K = \{e\}\,.
\eeq
This implies $g'g^{-1}=e$, which is impossible for $g'\neq g$. Hence we must have $i\neq j$. We thus see that $g$ and $g'$ must lie in distinct $H$-cosets. 

We can continue this way, and we see that all elements of $K$ must lie in distinct $H$-cosets, which implies
\beq
|K| \le [G:H] < \infty\,.
\eeq
\end{proof}

\subsection{Finite-index subgroups and direct products}
In section~\ref{sec:D2=A1xA1} we needed to understand the finite-index subgroups of the direct product of to copies $\SL(2,\ZZ)$. The subgroups of a direct product are described by \emph{Goursat's lemma}:
\begin{lemma}
\emph{Let $G_1$ and $G_2$ be two groups. Then there is a bijection between:
\begin{enumerate}
\item subgroups of $G_1\times G_2$,
\item quintuples $(H_1,N_1,H_2,N_2,\varphi)$, where $H_i\subseteq G_i$ are subgroups, $N_i$ is a normal subgroup of $H_i$, and $\varphi:\faktor{H_1}{N_1}\to\faktor{H_2}{N_2}$ is an isomormphism.
\end{enumerate}}
\end{lemma}

Goursat's lemma allows one to write down the subgroups of $G_1\times G_2$ explicitly. Let $Q=(H_1,N_1,H_2,N_2,\varphi)$ be a quintuple as described by the lemma, and let $\{h_{1,i}\}$ and $\{h_{2,j}\}$ be coset representatives of $N_1$ and $N_2$ in $H_1$ and $H_2$ respectively. 
Note that, since $\faktor{H_1}{N_1}\simeq\faktor{H_2}{N_2}$, $N_1$ and $N_2$ have the same number of cosets in $H_1$ and $H_2$ respectively.
We then associate to $Q$ the subgroup
\beq
\Gamma_Q := \big\{ (h_{1,i}n,h_{2,i}n') \in H_1\times H_2: (n,n')\in N_1\times N_2\big\}\subseteq G_1\times G_2\,.
\eeq
In the following we choose the convention that $h_{i,1}=e_i$ is the unit element of $G_i$. We then immediately see that $N_1\times N_2$ is always a subgroup of $\Gamma_Q$.

We can use Goursat's lemma to obtain a very useful property of  finite-index subgroups of $G_1\times G_2$:
\begin{lemma}
\emph{$\Gamma_Q$ has finite index in $G_1\times G_2$, then $N_1\times N_2$ has finite index in $\Gamma_Q$.}
\end{lemma}
\begin{proof}
Let $r:=[H_1:N_1] = [H_2:N_2]$. If $r=1$, then $N_i=H_i$, and the claim is trivial. Assume therefore that $r>1$.
Since $\Gamma_Q\subseteq H_1\times H_2\subseteq G_1\times G_2$, $\Gamma_Q$ has finite index in $H_1\times H_2$.

Assume $r<\infty$. Then
\beq
[H_1\times H_2: N_1\times N_2] = r^2<\infty\,,
\eeq
and since $N_1\times N_2\subseteq \Gamma_Q\subseteq H_1\times H_2$, the claim follows.

It remains to analyse the case $r=\infty$. We have $(h_{1,i},e_2)\in \Gamma_Q$ if and only if $i=1$ (and $h_{1,1}=e_1$). Indeed, $(h_{1,i},e_2)\in \Gamma_Q$ if and only if
\beq
\varphi(h_{1,i}N_1) = e_2N_2 = h_{2,1}N_2 = \varphi(h_{1,1}N_1)\,.
\eeq
Since $\varphi$ is an isomorphism, this can only happen if $h_{1,i}N_1 = h_{1,1}N_1$, i.e., if $i=1$.

Let $i,j\neq 1$, and let $x\in (h_{1,i},e_2)\Gamma_Q\cap(h_{1,j},e_2)\Gamma_Q$. Then there must be $(h_{1,a}n,h_{2,a}n')$ and $(h_{1,b}m,h_{2,b}m')$ in $\Gamma_Q$, with $(n,n'),(m,m')\in N_1\times N_2$, such that
\beq
x = (h_{1,i}h_{1,a}n,h_{2,a}n') = (h_{1,j}h_{1,b}m,h_{2,b}m')\,.
\eeq
From this we conclude that $h_{2,a}n' = h_{2,b}m'$, and so $a=b$ (because the $h_{2,k}$ are coset representatives). Moreover, since $N_1$ is normal in $H_1$, there is $\tilde{n},\tilde{m}\in N_1$ such that $h_{1,a}n=\tilde{n}h_{1,a}$ and $h_{1,b}m=\tilde{m}h_{1,b}$. Putting everyting together, we get the constraint
\beq
h_{1,i}\tilde{n}h_{1,a} = h_{1,j}\tilde{m}h_{1,a}\,.
\eeq
From this follows that $h_{1,i}=h_{1,j}$ (because the $h_{1,k}$ are coset representatives). Hence,
\beq
(h_{1,i},e_2)\Gamma_Q\cap(h_{1,j},e_2)\Gamma_Q = \emptyset \textrm{~~~if } i\neq j\,.
\eeq
In other words, $(h_{1,i},e_2)\Gamma_Q$ define distinct cosets in $H_1\times H_2$. There are then $r=\infty$ distinct cosets $(h_{1,i},e_2)\Gamma_Q$, which is impossible if $\Gamma_Q$ has finite index in $H_1\times H_2$. Hence, $r=\infty$ is impossible, which finishes the proof.
\end{proof}

\subsection{Congruence subgroups}
In the context of periods for elliptic curves, one typically encounters modular forms for congruence subgroups of $\SL(2,\ZZ)$. The latter are defined as subgroups $\Gamma\subseteq\SL(2,
\ZZ)$ that contain a principal congruence subgroup:
\beq
\Gamma(n) = \big\{\gamma\in\SL(2,\ZZ):\gamma=\mathds{1}\!\!\!\mod n\big\}\,.
\eeq
Typical congruence subgroups encountered are
\beq\bsp
\Gamma_0(n) &\,= \left\{\begin{mymatrix}a&b\\c&d\end{mymatrix}\in\SL(2,\ZZ):c=0\!\!\!\mod n\right\}\,,\\
\Gamma_1(n) &\,= \left\{\begin{mymatrix}a&b\\c&d\end{mymatrix}\in\SL(2,\ZZ):a,d=1\!\!\!\!\mod n\textrm{~and~}c=0\!\!\!\mod n\right\}\,.
\esp\eeq
Our aim here is to indicate how these definitions extend to other groups, in particular those encountered in this paper.

Let $R$ be a commutative ring, and $\pf\subseteq R$ an ideal such that $\left|{R}/{\pf}\right|<\infty$. Let $p:\GL(m,R)\to \GL(m,R/\pf)$ the reduction mod $\pf$ (i.e., we replace each entry in a matrix by its reduction mod $\pf$). We define the \emph{principal congruence subgroup}
\beq
\GL(m,R,\pf) := \Ker(p) = \left\{\gamma\in \GL(m,R) : \gamma=\mathds{1}\!\!\!\mod \pf\right\}\,.
\eeq

Consider a group $G\subseteq \GL(m,R)$. We say that a subgroup $H\subseteq G$ is a congruence subgroup if $H$ contains $G\cap\GL(m,R,\pf)$ for some $\pf$. The principal congruence subgroups of $G$ are precisely the subgroups $G\cap\GL(m,R,\pf)$. Note that this definition captures the previous one for $\SL(2,\ZZ)$. Indeed, in that case we have $m=2$, $R=\ZZ$, and $G=\SL(2,\ZZ)$. All ideals of $\ZZ$ have the form $n\ZZ$ for some integer $n$. The principal congruence subgroups are then 
\beq
\SL(2,\ZZ) \cap \GL(2,\ZZ,n\ZZ) = \Gamma(n)\,.
\eeq

Note that a congruence subgroup necessarily has finite index. To see this, first note that it is sufficient to prove this for the principal congruence subgroups, because every subgroup that contains a subgroup of finite index has itself finite index. We have the usual exact sequence
\beq
0\to H = G\cap \GL(m,R,\pf) \to G \to G\cap \im(p) \to 0\,,
\eeq
where $G\cap \im(p) = G/H$ is a group of $m\times m$ matrices with entries in $R/\pf$.
Hence
\beq
[G:H] = \left|G\cap \im(p)\right| \le \left|\left(\faktor{R}{\pf}\right)^{m\times m}\right|=m^2\,\left|\faktor{R}{\pf}\right|<\infty\,.
\eeq

Since every congruence subgroup has finite index, it is natural to ask if the converse is also true. This is the famous \emph{congruence subgroup problem}. It is known to be false for $\SL(2,\ZZ)$, but it has a positive answer for linear groups $\GL(n,\ZZ)$ of higher rank $n\ge 3$~\cite{bams/1183526018} (see, e.g., ref.~\cite{Raghunathan:2004aa} for a recent review). In particular it holds for $\SO_0(p,q,\ZZ)$ with $p,q\ge 2$ and for $\Sp(2n,\ZZ)$. It also holds for the Hilbert modular group $\SL(2,\ord_F)$ (if $F\neq \QQ$).

\section{Some properties of the discriminant kernel}
\subsection{The index of $\cD(\Lambda)$ in $\widetilde{\OO}(\Lambda)$ and $\SO_0(\Lambda)$}
\label{app:groups_1}

In this section we present the proof of the claim from section~\ref{sec:orthogonal_groups} that the group $\cD(\Lambda)$ defined in eq.~\eqref{eq:cD_def} has finite index in both $\widetilde{\OO}(\Lambda)$ and $\SO_0(\Lambda)$. The proof is a repeated application of Lemma~\ref{eq:finite_index_intersection}.

Assume that $\Lambda$ has signature $(p,q)$.
It is well known that the real Lie group $\OO(p,q)$ has four connected components, and so $[\OO(p,q):\SO_0(p,q)]=4$.
If we apply Lemma~\ref{eq:finite_index_intersection} with $G=\OO(p,q)$, $K=\SO_0(p,q)$ and $H=\OO(\Lambda)$, we find
\beq
[\OO(\Lambda):\SO_0(\Lambda)] = [\OO(\Lambda):(\OO(\Lambda)\cap\SO_0(p,q))] \le [\OO(p,q):\SO_0(p,q)] =4\,,
\eeq
and so $\SO_0(\Lambda)$ has finite index in $\OO(p,q)$.

Next, we apply Lemma~\ref{eq:finite_index_intersection} with $G=\OO(\Lambda)$, $K=\widetilde{\OO}(\Lambda)$ and $H=\SO_0(\Lambda)$. Then
\beq
[\SO_0(\Lambda):\cD(\Lambda)] = [\SO_0(\Lambda):(\SO_0(\Lambda)\cap\widetilde{\OO}(\Lambda))] \le [\OO(\Lambda):\widetilde{\OO}(\Lambda)]<\infty\,,
\eeq
and so $\cD(\Lambda)$ has finite index in $\SO_0(\Lambda)$, as claimed. For the last inequality, see the discussion in section~\ref{sec:orthogonal_groups}.

Finally, we apply Lemma~\ref{eq:finite_index_intersection} with $G=\OO(\Lambda)$, $K=\SO_0(\Lambda)$ and $H=\widetilde{\OO}(\Lambda)$. Then
\beq
[\widetilde{\OO}(\Lambda):\cD(\Lambda)] = [\widetilde{\OO}(\Lambda):(\widetilde{\OO}(\Lambda)\cap\SO_0(\Lambda))] \le [\OO(\Lambda):\SO_0(\Lambda)]\le4\,,
\eeq
and so $\cD(\Lambda)$ has finite index in $\widetilde{\OO}(\Lambda)$, as claimed.


\subsection{Congruence subgroups}
\label{sec:congruence}

Let us consider a lattice $\T_n := H^{\oplus p}\oplus \Lambda(n)$ for some  integers $n,p\ge1$. 
Our goal is to derive a relation between the discriminant kernels of $\cD(\T_n)$ and $\cD(\T_1(n))$. Recall the definition of $\varphi_{\Delta}$ from eq.~\eqref{eq:varphi_delta_def}.
\begin{lemma}\emph{
$\varphi_{\Delta_{n,p}}\big(\cD(\T_1(n))\big)$ is a congruence subgroup of $\cD(\T_n)$ and of $\widetilde{\OO}(\T_n)$.}
\end{lemma}

\begin{proof}
Assume that $\dim T_n = r$. We choose bases of $\T_1(n)$ and $\T_n$ such that the Gram matrices are respectively given by
\beq
\Sigma' = \begin{mymatrix}& &&&&&n\\
&&&&&\iddots&\\
&&&&n&&\\
&&&nS&&&\\
&&n&&&&\\
&\iddots&&&&&\\
n&&&&&&\end{mymatrix}\textrm{~~~and~~~}
\Sigma = \begin{mymatrix}& &&&&&1\\
&&&&&\iddots&\\
&&&&1&&\\
&&&nS&&&\\
&&1&&&&\\
&\iddots&&&&&\\
1&&&&&&\end{mymatrix}
\eeq
We clearly have
\beq
\Delta_{n,p}^T\Sigma\Delta_{n,p}=\Sigma'\,.
\eeq
In the following, in order to simplify the notations, we will write $\Delta$ instead of $\Delta_{n,p}$.
Let $R\in\cD(T_1(n))\subseteq \OO(T_1(n))$. We have
\beq\label{eq:proof_1}
R^T\Sigma'R=\Sigma'\,.
\eeq
A straightforward calculation shows that eq.~\eqref{eq:proof_1} implies
\beq
\varphi_{\Delta}(R)^T\Sigma\varphi_{\Delta}(R)=\Sigma\,.
\eeq
Hence, $\varphi_{\Delta}(R)$ preserves the bilinear form. The entries of this matrix, however, are not integers, but we have $\varphi_{\Delta}(R)\in\SO_0(\T_n,\ZZ[\tfrac{1}{n^p}])$. 
Our goal is now to show that the image $\im\varphi_{\Delta}$ lies in $\cD(T_n) \subseteq \SO_0(\T_n,\ZZ[\tfrac{1}{n^p}])$, and it actually defines a congruence subgroup of $\cD(T_n)$.

Let us start by showing that, for all $R\in \cD(T_1(n))$, $\varphi_{\Delta}(R)$ lies in $\cD(T_n)$.
By definition there is $M\in\ZZ^{r\times r}$ such that $R=\mathds{1}+M\Sigma'$, and so 
\beq\bsp\label{eq:proof_3}
\varphi_{\Delta}(R) &\,= \mathds{1} + \Delta M\Sigma'\Delta^{-1} \\
&\,=\mathds{1} + \Delta M\Delta^T(\Delta^{-1})^T\Sigma'\Delta^{-1}\\
&\, = \mathds{1} + \Delta M\Delta^{T}\Sigma\,,
\esp\eeq
Clearly $ \Delta M\Delta^{T}$ has only integer entries, and so $\varphi_{\Delta}(R)\in\cD(\T_n)$. The previous considerations also provide us with an explicit expression for $\im\varphi_{\Delta}$. We have
\beq
\im\varphi_{\Delta} = \left\{ \mathds{1} + \Delta M\Delta^{T}\Sigma \in \cD(\T_n) : M\in \ZZ^{r\times r}\right\}\,.
\eeq

It remains to show that $\im\varphi_{\Delta}$ is a congruence subgroup, i.e., we need to show that there is an integer  $N\ge 1$ such that $\im\varphi_{\Delta}$ contains the subgroup defined as the kernel of the projection modulo $N$ (see appendix~\ref{app:groups}),
\beq
\cK_N := \Ker\big(\cD(\T_n)\to \cD(\T_n,\ZZ_N)\big)\,,
\eeq
where we defined 
\beq
\cD(\T_n,\ZZ_N) := \left\{\mathds{1}+M\Sigma\in\SO_0(\Sigma,\ZZ_N) : M\in \ZZ_N\right\}\,.
\eeq
Let $d:=\det\Sigma$ and set $N:=|d|n^{2p}$. We have $d\mathds{1} = \Sigma^{\#}\,\Sigma$, where $\Sigma^\#$ is the matrix of cofactors of $\Sigma$. For every $R\in\cK_N$,there is $M\in \ZZ^{r\times r}$ such that (with $s:=\textrm{sign}(d)$)
\beq
R=\mathds{1}+NM =1+\big(sn^{2p}M\Sigma^\#\big) \Sigma\,.
\eeq
It is easy to see that there is $M'\in\ZZ^{r\times r}$ such that $n^{2p}M\Sigma^\#=\Delta^TM'\Delta$, and so we have
\beq
R = \mathds{1}+\Delta^T M'\Delta\Sigma\in\im \varphi_{\Delta}\,.
\eeq
Hence, $\cK_N \subseteq \im\varphi_{\Delta}$, and so $\im\varphi_{\Delta}$ is a congruence subgroup.
\end{proof}


\section{Modular forms in two variables}
\label{app:modular}

There is no commonly accepted definition of modular forms of more than one variable  for subgroups of $\SL(2,\mathbb{Z})$; see, e.g., refs.~\cite{10.1215/ijm/1258138437,bimodular,WangYang,Manschot:2021qqe,Aspman:2021evt}. Here we review the definition that is relevant for this paper.

Let $\Gamma_1$ and $\Gamma_2$ be two subgroups of finite index of $\SL(2,\mathbb{Z})$. A modular form in two variables of weight $(k_1,k_2)$ for $(\Gamma_1,\Gamma_2)$ is a function $f:\mathbb{H}\times\HH\to \mathbb{C}$ holomorphic on $\mathbb{H}\times\HH$ and at the cusps, such that for all $\gamma_i\in\Gamma_i$ (cf.,~e.g.,~ref.~\cite{10.1215/ijm/1258138437})
\beq
f(\gamma_1\cdot \tau_1,\gamma_2\cdot \tau_2) = j(\gamma_1,\tau_1)^{k_1}\,j(\gamma_2,\tau_2)^{k_2}\,f(\tau_1,\tau_2)\,.
\eeq
We denote the $\mathbb{C}$-vector space of holomorphic and meromorphic modular forms in two variables of weight $(k_1,k_2)$ for $(\Gamma_1,\Gamma_2)$ by $\cM_{k_1,k_2}(\Gamma_1,\Gamma_2)$ and $\cM_{k_1,k_2}^{\textrm{mer}}(\Gamma_1,\Gamma_2)$, respectively. 


Let us make some comment about the structure of $\cM_{k_1,k_2}^{\textrm{mer}}(\Gamma_1,\Gamma_2)$. We start by discussing the case of modular functions, $(k_1,k_2)=(0,0)$. Then $\cM_{0,0}^{\textrm{mer}}(\Gamma_1,\Gamma_2)$ is the field of meromorphic functions on $(\Gamma_1\times \Gamma_2)\setminus(\mathbb{H}\times\HH)= Y_{\Gamma_1}\times Y_{\Gamma_2}$, where $Y_{\Gamma_i}=\Gamma_i\setminus \HH$ is the modular curve for $\Gamma_i$. Since $Y_{\Gamma_i}$ is a Riemann surface, its field of meromorphic functions is generated by two generators $(x_i,y_i)$, and so $\cM_{0,0}^{\textrm{mer}}(\Gamma_1,\Gamma_2)$ is a (quotient of a) field of rational functions in four generators. 

Next let us discuss the case of $\cM_{k_1,k_2}^{\textrm{mer}}(\Gamma_1,\Gamma_2)$ with $(k_1,k_2)\neq (0,0)$. Let $f\in\cM_{k_1,k_2}^{\textrm{mer}}(\Gamma_1,\Gamma_2)$. If $g_i\in \cM_{k_i}^{\textrm{mer}}(\Gamma_i)$, $i=1,2$, then $f/(g_1g_2)\in \cM_{0,0}^{\textrm{mer}}(\Gamma_1,\Gamma_2)$, and so 
\beq
\cM_{k_1,k_2}^{\textrm{mer}}(\Gamma_1,\Gamma_2) = \cM_{0,0}^{\textrm{mer}}(\Gamma_1,\Gamma_2)\otimes \cM_{k_1}^{\textrm{mer}}(\Gamma_1) \otimes\cM_{k_2}^{\textrm{mer}}(\Gamma_2) \,.
\eeq
%
The only way this argument could fail is if $ \cM_{k_1}^{\textrm{mer}}(\Gamma_1)$ or $ \cM_{k_2}^{\textrm{mer}}(\Gamma_2)$ were trivial, while $\cM_{k_1,k_2}^{\textrm{mer}}(\Gamma_1,\Gamma_2)$ is not (because then there would be no $g_1$ and $g_2$). As we show now, this can never happen. Indeed, consider a non-zero meromorphic modular form of weight $(k_1,k_2)$ for $(\Gamma_1,\Gamma_2)$, and assume $k_1\neq 0$. Pick $\tau_0\in\mathbb{H}$ such that $\tau_0$ is not an elliptic point for $\Gamma_2$, nor such that $f(\tau_1,\tau_2)$ has no pole or zero for $\tau_2=\tau_0$ (such a $\tau_0$ always exists). Then we have, for all $\gamma_1\in\Gamma_1$,
\beq
f(\gamma_1\cdot \tau_1,\tau_0) = j(\gamma_1,\tau_1)^{k_1}\,f(\tau_1,\tau_0)\,,
\eeq
and so $f(\tau_1,\tau_0)$ is a non-zero meromorphic modular form of weight $k_1$ for $\Gamma_1$. Hence, $ \cM_{k_1}^{\textrm{mer}}(\Gamma_1)$ is not trivial. A similar argument shows that $ \cM_{k_2}^{\textrm{mer}}(\Gamma_2)$ is not trivial.

Finally, we show that for holomorphic modular forms, we simply have
\beq
\cM_{k_1,k_2}(\Gamma_1,\Gamma_2) = \cM_{k_1}(\Gamma_1) \otimes\cM_{k_2}(\Gamma_2) \,.
\eeq
From the previous argument we know that for every $f\in \cM_{k_1,k_2}(\Gamma_1,\Gamma_2)$ there is $g_i\in \cM_{k_i}(\Gamma_i)$ and a rational function $\rho(x_1,y_1,x_2,y_2) = P(x_1,y_1,x_2,y_2)/Q(x_1,y_1,x_2,y_2)$ such that ($P$ and $Q$ are polynomials)
\beq\bsp
f&\,= \rho(x_1,y_1,x_2,y_2)\,g_1\,g_2=\frac{P(x_1,y_1,x_2,y_2)}{Q(x_1,y_1,x_2,y_2)}\,g_1\,g_2\,,
\esp\eeq
where $(x_i,y_i)$ are two generators for $\cM_{k_i}(\Gamma_i)$. We now argue that the polynomial $Q(x_1,y_1,x_2,y_2)$ must take the factorised form 
\beq\label{eq:Q_fac}
Q(x_1,y_1,x_2,y_2) = Q_1(x_1,y_1)\,Q_2(x_2,y_2)\,.
\eeq
Indeed, since $f$ is holomorphic, the poles coming from the zeroes in $Q$ must be canceled by zeroes of $g_1$ and $g_2$. Assume that $Q(x_1,y_1,x_2,y_2) $ contains an irreducible factor $B(x_1,y_1,x_2,y_2)$. The pole coming from $B(x_1,y_1,x_2,y_2)=0$ cannot be cancelled by the factorised expression $g_1\,g_2$. Hence, $Q(x_1,y_1,x_2,y_2) $ must take the factorised form in eq.~\eqref{eq:Q_fac}, and we have
\beq\bsp
f(\tau_1,\tau_2) 
&\,=P(x_1,y_1,x_2,y_2)\,\frac{g_1}{Q(x_1,y_1)}\,\frac{g_2}{Q_2(x_2,y_2)} = P(x_1,y_1,x_2,y_2)\,\tilde{g}_1\,\tilde{g}_2\,,
\esp\eeq
with $\tilde{g}_i:=g_i/Q_i(x_i,y_i)\in \cM_{k_i}(\Gamma_i)$. The degree of $P$ is bounded by the requirement that $f$ is holomorphic, and we then obtain a finite linear combination of products of holomorphic modular forms.

\bibliographystyle{jhep}
\bibliography{twistedbib}

\providecommand{\href}[2]{#2}\begingroup\raggedright\begin{thebibliography}{100}

\bibitem{Mpls1}
A.~B. Goncharov, \emph{Multiple polylogarithms and mixed {Tate} motives},  May,
  2001.
\newblock 10.48550/arXiv.math/0103059.

\bibitem{Mpls2}
A.~B. Goncharov, \emph{Galois symmetries of fundamental groupoids and
  noncommutative geometry},  June, 2004.
\newblock 10.48550/arXiv.math/0208144.

\bibitem{Remiddi:1999ew}
E.~Remiddi and J.~A.~M. Vermaseren, \emph{{Harmonic polylogarithms}},
  \href{https://doi.org/10.1142/S0217751X00000367}{\emph{Int. J. Mod. Phys. A}
  {\bfseries 15} (2000) 725}
  [\href{https://arxiv.org/abs/hep-ph/9905237}{{\ttfamily hep-ph/9905237}}].

\bibitem{Sabry}
A.~Sabry, \emph{{Fourth order spectral functions for the electron propagator}},
  {\emph{Nucl. Phys.} {\bfseries 33} (1962) 401}.

\bibitem{Caffo:1998du}
M.~Caffo, H.~Czyz, S.~Laporta and E.~Remiddi, \emph{{The Master differential
  equations for the two loop sunrise selfmass amplitudes}}, {\emph{Nuovo Cim.}
  {\bfseries A111} (1998) 365}
  [\href{https://arxiv.org/abs/hep-th/9805118}{{\ttfamily hep-th/9805118}}].

\bibitem{Laporta:2004rb}
S.~Laporta and E.~Remiddi, \emph{{Analytic treatment of the two loop equal mass
  sunrise graph}},
  \href{https://doi.org/10.1016/j.nuclphysb.2004.10.044}{\emph{Nucl. Phys. B}
  {\bfseries 704} (2005) 349}
  [\href{https://arxiv.org/abs/hep-ph/0406160}{{\ttfamily hep-ph/0406160}}].

\bibitem{Laporta:2008sx}
S.~Laporta, \emph{{Analytical expressions of 3 and 4-loop sunrise Feynman
  integrals and 4-dimensional lattice integrals}},
  \href{https://doi.org/10.1142/S0217751X08042869}{\emph{Int. J. Mod. Phys. A}
  {\bfseries 23} (2008) 5007}
  [\href{https://arxiv.org/abs/0803.1007}{{\ttfamily 0803.1007}}].

\bibitem{Muller-Stach:2011qkg}
S.~M\"uller-Stach, S.~Weinzierl and R.~Zayadeh, \emph{{A Second-Order
  Differential Equation for the Two-Loop Sunrise Graph with Arbitrary Masses}},
  \href{https://doi.org/10.4310/CNTP.2012.v6.n1.a5}{\emph{Commun. Num. Theor.
  Phys.} {\bfseries 6} (2012) 203}
  [\href{https://arxiv.org/abs/1112.4360}{{\ttfamily 1112.4360}}].

\bibitem{Muller-Stach:2012tgj}
S.~M\"uller-Stach, S.~Weinzierl and R.~Zayadeh, \emph{{Picard-Fuchs equations
  for Feynman integrals}},
  \href{https://doi.org/10.1007/s00220-013-1838-3}{\emph{Commun. Math. Phys.}
  {\bfseries 326} (2014) 237}
  [\href{https://arxiv.org/abs/1212.4389}{{\ttfamily 1212.4389}}].

\bibitem{ell2}
S.~Bloch and P.~Vanhove, \emph{The elliptic dilogarithm for the sunset graph},
  \href{https://doi.org/10.1016/j.jnt.2014.09.032}{\emph{Journal of Number
  Theory} {\bfseries 148} (2015) 328}.

\bibitem{LevinRacinet}
A.~Levin and G.~Racinet, \emph{{Towards multiple elliptic polylogarithms}},
  \href{https://arxiv.org/abs/math/0703237}{{\ttfamily math/0703237}}.

\bibitem{MR1265553}
A.~Be\u{\i}linson and A.~Levin, \emph{The elliptic polylogarithm},  in
  \emph{Motives ({S}eattle, {WA}, 1991)}, vol.~55 of \emph{Proc. Sympos. Pure
  Math.}, pp.~123--190, Amer. Math. Soc., Providence, RI, (1994),
  \href{https://doi.org/10.1007/s00208-018-1645-4}{DOI}.

\bibitem{BrownLevin}
F.~Brown and A.~Levin, \emph{{Multiple Elliptic Polylogarithms}},
  \href{https://arxiv.org/abs/1110.6917}{{\ttfamily 1110.6917}}.

\bibitem{Broedel:2014vla}
J.~Broedel, C.~R. Mafra, N.~Matthes and O.~Schlotterer, \emph{{Elliptic
  multiple zeta values and one-loop superstring amplitudes}},
  \href{https://doi.org/10.1007/JHEP07(2015)112}{\emph{JHEP} {\bfseries 07}
  (2015) 112} [\href{https://arxiv.org/abs/1412.5535}{{\ttfamily 1412.5535}}].

\bibitem{ell15}
J.~Broedel, C.~Duhr, F.~Dulat and L.~Tancredi, \emph{Elliptic polylogarithms
  and iterated integrals on elliptic curves. part i: general formalism},
  \href{https://doi.org/10.1007/jhep05(2018)093}{\emph{Journal of High Energy
  Physics} {\bfseries 2018} (2018) }.

\bibitem{EnriquezZerbini}
B.~Enriquez and F.~Zerbini, \emph{Elliptic hyperlogarithms},
  \href{https://arxiv.org/abs/2307.01833}{{\ttfamily 2307.01833}}.

\bibitem{ManinModular}
Y.~I. Manin, \emph{{Iterated integrals of modular forms and noncommutative
  modular symbols}},  in \emph{Algebraic geometry and number theory}, vol.~253
  of \emph{Progr. Math.}, (Boston), pp.~565--597, Birkh\"auser Boston, 2006,
  \href{https://arxiv.org/abs/math/0502576}{{\ttfamily math/0502576}}.

\bibitem{Brown2014MultipleMV}
F.~Brown, \emph{Multiple modular values and the relative completion of the
  fundamental group of \$m\_\{1,1\}\$}, {\emph{arXiv: Number Theory} (2014) }.

\bibitem{Adams:2017ejb}
L.~Adams and S.~Weinzierl, \emph{{Feynman integrals and iterated integrals of
  modular forms}},
  \href{https://doi.org/10.4310/CNTP.2018.v12.n2.a1}{\emph{Commun. Num. Theor.
  Phys.} {\bfseries 12} (2018) 193}
  [\href{https://arxiv.org/abs/1704.08895}{{\ttfamily 1704.08895}}].

\bibitem{ell14}
J.~Broedel, C.~Duhr, F.~Dulat, B.~Penante and L.~Tancredi, \emph{Elliptic
  symbol calculus: from elliptic polylogarithms to iterated integrals of
  eisenstein series},
  \href{https://doi.org/10.1007/jhep08(2018)014}{\emph{Journal of High Energy
  Physics} {\bfseries 2018} (2018) }.

\bibitem{Bourjaily:2022bwx}
J.~L. Bourjaily et~al., \emph{{Functions Beyond Multiple Polylogarithms for
  Precision Collider Physics}},  in \emph{{Snowmass 2021}}, 3, 2022,
  \href{https://arxiv.org/abs/2203.07088}{{\ttfamily 2203.07088}}.

\bibitem{Schimmrigk:2024xid}
R.~Schimmrigk, \emph{{Special Fano geometry from Feynman integrals}},
  \href{https://arxiv.org/abs/2412.20236}{{\ttfamily 2412.20236}}.

\bibitem{Huang:2013kh}
R.~Huang and Y.~Zhang, \emph{{On Genera of Curves from High-loop Generalized
  Unitarity Cuts}}, \href{https://doi.org/10.1007/JHEP04(2013)080}{\emph{JHEP}
  {\bfseries 04} (2013) 080} [\href{https://arxiv.org/abs/1302.1023}{{\ttfamily
  1302.1023}}].

\bibitem{Hauenstein:2014mda}
J.~D. Hauenstein, R.~Huang, D.~Mehta and Y.~Zhang, \emph{{Global Structure of
  Curves from Generalized Unitarity Cut of Three-loop Diagrams}},
  \href{https://doi.org/10.1007/JHEP02(2015)136}{\emph{JHEP} {\bfseries 02}
  (2015) 136} [\href{https://arxiv.org/abs/1408.3355}{{\ttfamily 1408.3355}}].

\bibitem{Doran:2023yzu}
C.~F. Doran, A.~Harder, P.~Vanhove and E.~Pichon-Pharabod, \emph{{Motivic
  Geometry of two-Loop Feynman Integrals}},
  \href{https://doi.org/10.1093/qmath/haae015}{\emph{Quart. J. Math. Oxford
  Ser.} {\bfseries 75} (2024) 901}
  [\href{https://arxiv.org/abs/2302.14840}{{\ttfamily 2302.14840}}].

\bibitem{Abreu:2024jde}
S.~Abreu, A.~Behring, A.~J. McLeod and B.~Page, \emph{{Four-loop two-mass
  tadpoles and the $\rho$ parameter}},
  \href{https://doi.org/10.22323/1.467.0008}{\emph{PoS} {\bfseries LL2024}
  (2024) 008} [\href{https://arxiv.org/abs/2407.21700}{{\ttfamily
  2407.21700}}].

\bibitem{Enriquez_hyperelliptic}
B.~Enriquez, \emph{Flat connections on configuration spaces and braid groups of
  surfaces}, {\emph{Adv. Math.} {\bfseries 252} (2014) 204}
  [\href{https://arxiv.org/abs/1112.0864}{{\ttfamily 1112.0864}}].

\bibitem{DHoker:2023vax}
E.~D'Hoker, M.~Hidding and O.~Schlotterer, \emph{{Constructing polylogarithms
  on higher-genus Riemann surfaces}},
  \href{https://arxiv.org/abs/2306.08644}{{\ttfamily 2306.08644}}.

\bibitem{DHoker:2024ozn}
E.~D'Hoker and O.~Schlotterer, \emph{{Fay identities for polylogarithms on
  higher-genus Riemann surfaces}},
  \href{https://arxiv.org/abs/2407.11476}{{\ttfamily 2407.11476}}.

\bibitem{Baune:2024biq}
K.~Baune, J.~Broedel, E.~Im, A.~Lisitsyn and F.~Zerbini,
  \emph{{Schottky\textendash{}Kronecker forms and hyperelliptic
  polylogarithms}}, \href{https://doi.org/10.1088/1751-8121/ad8197}{\emph{J.
  Phys. A} {\bfseries 57} (2024) 445202}
  [\href{https://arxiv.org/abs/2406.10051}{{\ttfamily 2406.10051}}].

\bibitem{Baune:2024ber}
K.~Baune, J.~Broedel, E.~Im, A.~Lisitsyn and Y.~Moeckli, \emph{{Higher-genus
  Fay-like identities from meromorphic generating functions}},
  \href{https://arxiv.org/abs/2409.08208}{{\ttfamily 2409.08208}}.

\bibitem{DHoker:2025szl}
E.~D'Hoker, B.~Enriquez, O.~Schlotterer and F.~Zerbini, \emph{{Relating flat
  connections and polylogarithms on higher genus Riemann surfaces}},
  \href{https://arxiv.org/abs/2501.07640}{{\ttfamily 2501.07640}}.

\bibitem{Ichikawa}
T.~Ichikawa, \emph{{Higher genus polylogarithms on families of Riemann
  surfaces}},  \href{https://arxiv.org/abs/2209.05006}{{\ttfamily 2209.05006}}.

\bibitem{Duhr:2024uid}
C.~Duhr, F.~Porkert and S.~F. Stawinski, \emph{{Canonical differential
  equations beyond genus one}},
  \href{https://doi.org/10.1007/JHEP02(2025)014}{\emph{JHEP} {\bfseries 02}
  (2025) 014} [\href{https://arxiv.org/abs/2412.02300}{{\ttfamily
  2412.02300}}].

\bibitem{Brown:2010bw}
F.~Brown and O.~Schnetz, \emph{{A K3 in $\phi^4$}},
  \href{https://doi.org/10.1215/00127094-1644201}{\emph{Duke Math. J.}
  {\bfseries 161} (2012) 1817}
  [\href{https://arxiv.org/abs/1006.4064}{{\ttfamily 1006.4064}}].

\bibitem{Bourjaily:2018ycu}
J.~L. Bourjaily, Y.-H. He, A.~J. Mcleod, M.~Von~Hippel and M.~Wilhelm,
  \emph{{Traintracks through Calabi-Yau Manifolds: Scattering Amplitudes beyond
  Elliptic Polylogarithms}},
  \href{https://doi.org/10.1103/PhysRevLett.121.071603}{\emph{Phys. Rev. Lett.}
  {\bfseries 121} (2018) 071603}
  [\href{https://arxiv.org/abs/1805.09326}{{\ttfamily 1805.09326}}].

\bibitem{Bourjaily:2018yfy}
J.~L. Bourjaily, A.~J. McLeod, M.~von Hippel and M.~Wilhelm, \emph{{Bounded
  Collection of Feynman Integral Calabi-Yau Geometries}},
  \href{https://doi.org/10.1103/PhysRevLett.122.031601}{\emph{Phys. Rev. Lett.}
  {\bfseries 122} (2019) 031601}
  [\href{https://arxiv.org/abs/1810.07689}{{\ttfamily 1810.07689}}].

\bibitem{Bourjaily:2019hmc}
J.~L. Bourjaily, A.~J. McLeod, C.~Vergu, M.~Volk, M.~Von~Hippel and M.~Wilhelm,
  \emph{{Embedding Feynman Integral (Calabi-Yau) Geometries in Weighted
  Projective Space}},
  \href{https://doi.org/10.1007/JHEP01(2020)078}{\emph{JHEP} {\bfseries 01}
  (2020) 078} [\href{https://arxiv.org/abs/1910.01534}{{\ttfamily
  1910.01534}}].

\bibitem{Bloch:2014qca}
S.~Bloch, M.~Kerr and P.~Vanhove, \emph{{A Feynman integral via higher normal
  functions}}, \href{https://doi.org/10.1112/S0010437X15007472}{\emph{Compos.
  Math.} {\bfseries 151} (2015) 2329}
  [\href{https://arxiv.org/abs/1406.2664}{{\ttfamily 1406.2664}}].

\bibitem{MR3780269}
S.~Bloch, M.~Kerr and P.~Vanhove, \emph{Local mirror symmetry and the sunset
  {F}eynman integral},
  \href{https://doi.org/10.4310/ATMP.2017.v21.n6.a1}{\emph{Adv. Theor. Math.
  Phys.} {\bfseries 21} (2017) 1373}
  [\href{https://arxiv.org/abs/1601.08181}{{\ttfamily 1601.08181}}].

\bibitem{Klemm:2019dbm}
A.~Klemm, C.~Nega and R.~Safari, \emph{{The $l$-loop Banana Amplitude from GKZ
  Systems and relative Calabi-Yau Periods}},
  \href{https://doi.org/10.1007/JHEP04(2020)088}{\emph{JHEP} {\bfseries 04}
  (2020) 088} [\href{https://arxiv.org/abs/1912.06201}{{\ttfamily
  1912.06201}}].

\bibitem{Bonisch:2020qmm}
K.~B\"onisch, F.~Fischbach, A.~Klemm, C.~Nega and R.~Safari, \emph{{Analytic
  structure of all loop banana integrals}},
  \href{https://doi.org/10.1007/JHEP05(2021)066}{\emph{JHEP} {\bfseries 05}
  (2021) 066} [\href{https://arxiv.org/abs/2008.10574}{{\ttfamily
  2008.10574}}].

\bibitem{Bonisch:2021yfw}
K.~B\"onisch, C.~Duhr, F.~Fischbach, A.~Klemm and C.~Nega, \emph{{Feynman
  integrals in dimensional regularization and extensions of Calabi-Yau
  motives}}, \href{https://doi.org/10.1007/JHEP09(2022)156}{\emph{JHEP}
  {\bfseries 09} (2022) 156}
  [\href{https://arxiv.org/abs/2108.05310}{{\ttfamily 2108.05310}}].

\bibitem{Driesse:2024feo}
M.~Driesse, G.~U. Jakobsen, A.~Klemm, G.~Mogull, C.~Nega, J.~Plefka et~al.,
  \emph{{High-precision black hole scattering with Calabi-Yau manifolds}},
  \href{https://arxiv.org/abs/2411.11846}{{\ttfamily 2411.11846}}.

\bibitem{Klemm:2024wtd}
A.~Klemm, C.~Nega, B.~Sauer and J.~Plefka, \emph{{Calabi-Yau periods for black
  hole scattering in classical general relativity}},
  \href{https://doi.org/10.1103/PhysRevD.109.124046}{\emph{Phys. Rev. D}
  {\bfseries 109} (2024) 124046}
  [\href{https://arxiv.org/abs/2401.07899}{{\ttfamily 2401.07899}}].

\bibitem{Frellesvig:2023bbf}
H.~Frellesvig, R.~Morales and M.~Wilhelm, \emph{{Calabi-Yau Meets Gravity: A
  Calabi-Yau Threefold at Fifth Post-Minkowskian Order}},
  \href{https://doi.org/10.1103/PhysRevLett.132.201602}{\emph{Phys. Rev. Lett.}
  {\bfseries 132} (2024) 201602}
  [\href{https://arxiv.org/abs/2312.11371}{{\ttfamily 2312.11371}}].

\bibitem{Frellesvig:2024rea}
H.~Frellesvig, R.~Morales, S.~P\"ogel, S.~Weinzierl and M.~Wilhelm,
  \emph{{Calabi-Yau Feynman integrals in gravity: $\varepsilon$-factorized form
  for apparent singularities}},
  \href{https://arxiv.org/abs/2412.12057}{{\ttfamily 2412.12057}}.

\bibitem{Frellesvig:2024zph}
H.~Frellesvig, R.~Morales and M.~Wilhelm, \emph{{Classifying post-Minkowskian
  geometries for gravitational waves via loop-by-loop Baikov}},
  \href{https://doi.org/10.1007/JHEP08(2024)243}{\emph{JHEP} {\bfseries 08}
  (2024) 243} [\href{https://arxiv.org/abs/2405.17255}{{\ttfamily
  2405.17255}}].

\bibitem{Dlapa:2024cje}
C.~Dlapa, G.~K\"alin, Z.~Liu and R.~A. Porto, \emph{{Local in Time Conservative
  Binary Dynamics at Fourth Post-Minkowskian Order}},
  \href{https://doi.org/10.1103/PhysRevLett.132.221401}{\emph{Phys. Rev. Lett.}
  {\bfseries 132} (2024) 221401}
  [\href{https://arxiv.org/abs/2403.04853}{{\ttfamily 2403.04853}}].

\bibitem{Forner:2024ojj}
F.~Forner, C.~Nega and L.~Tancredi, \emph{{On the photon self-energy to three
  loops in QED}},  \href{https://arxiv.org/abs/2411.19042}{{\ttfamily
  2411.19042}}.

\bibitem{Duhr:2022pch}
C.~Duhr, A.~Klemm, F.~Loebbert, C.~Nega and F.~Porkert,
  \emph{{Yangian-Invariant Fishnet Integrals in Two Dimensions as Volumes of
  Calabi-Yau Varieties}},
  \href{https://doi.org/10.1103/PhysRevLett.130.041602}{\emph{Phys. Rev. Lett.}
  {\bfseries 130} (2023) 041602}
  [\href{https://arxiv.org/abs/2209.05291}{{\ttfamily 2209.05291}}].

\bibitem{Duhr:2023eld}
C.~Duhr, A.~Klemm, F.~Loebbert, C.~Nega and F.~Porkert, \emph{{The Basso-Dixon
  formula and Calabi-Yau geometry}},
  \href{https://doi.org/10.1007/JHEP03(2024)177}{\emph{JHEP} {\bfseries 03}
  (2024) 177} [\href{https://arxiv.org/abs/2310.08625}{{\ttfamily
  2310.08625}}].

\bibitem{Duhr:2024hjf}
C.~Duhr, A.~Klemm, F.~Loebbert, C.~Nega and F.~Porkert, \emph{{Geometry from
  integrability: multi-leg fishnet integrals in two dimensions}},
  \href{https://doi.org/10.1007/JHEP07(2024)008}{\emph{JHEP} {\bfseries 07}
  (2024) 008} [\href{https://arxiv.org/abs/2402.19034}{{\ttfamily
  2402.19034}}].

\bibitem{Pogel:2022ken}
S.~P\"ogel, X.~Wang and S.~Weinzierl, \emph{{Taming Calabi-Yau Feynman
  Integrals: The Four-Loop Equal-Mass Banana Integral}},
  \href{https://doi.org/10.1103/PhysRevLett.130.101601}{\emph{Phys. Rev. Lett.}
  {\bfseries 130} (2023) 101601}
  [\href{https://arxiv.org/abs/2211.04292}{{\ttfamily 2211.04292}}].

\bibitem{Pogel:2022vat}
S.~P\"ogel, X.~Wang and S.~Weinzierl, \emph{{Bananas of equal mass: any loop,
  any order in the dimensional regularisation parameter}},
  \href{https://doi.org/10.1007/JHEP04(2023)117}{\emph{JHEP} {\bfseries 04}
  (2023) 117} [\href{https://arxiv.org/abs/2212.08908}{{\ttfamily
  2212.08908}}].

\bibitem{Pogel:2022yat}
S.~P\"ogel, X.~Wang and S.~Weinzierl, \emph{{The three-loop equal-mass banana
  integral in \ensuremath{\varepsilon}-factorised form with meromorphic modular
  forms}}, \href{https://doi.org/10.1007/JHEP09(2022)062}{\emph{JHEP}
  {\bfseries 09} (2022) 062}
  [\href{https://arxiv.org/abs/2207.12893}{{\ttfamily 2207.12893}}].

\bibitem{Gorges:2023zgv}
L.~G\"orges, C.~Nega, L.~Tancredi and F.~J. Wagner, \emph{{On a procedure to
  derive \ensuremath{\epsilon}-factorised differential equations beyond
  polylogarithms}}, \href{https://doi.org/10.1007/JHEP07(2023)206}{\emph{JHEP}
  {\bfseries 07} (2023) 206}
  [\href{https://arxiv.org/abs/2305.14090}{{\ttfamily 2305.14090}}].

\bibitem{Primo:2017ipr}
A.~Primo and L.~Tancredi, \emph{{Maximal cuts and differential equations for
  Feynman integrals. An application to the three-loop massive banana graph}},
  \href{https://doi.org/10.1016/j.nuclphysb.2017.05.018}{\emph{Nucl. Phys. B}
  {\bfseries 921} (2017) 316}
  [\href{https://arxiv.org/abs/1704.05465}{{\ttfamily 1704.05465}}].

\bibitem{Broedel:2019kmn}
J.~Broedel, C.~Duhr, F.~Dulat, R.~Marzucca, B.~Penante and L.~Tancredi,
  \emph{{An analytic solution for the equal-mass banana graph}},
  \href{https://doi.org/10.1007/JHEP09(2019)112}{\emph{JHEP} {\bfseries 09}
  (2019) 112} [\href{https://arxiv.org/abs/1907.03787}{{\ttfamily
  1907.03787}}].

\bibitem{Broedel:2021zij}
J.~Broedel, C.~Duhr and N.~Matthes, \emph{{Meromorphic modular forms and the
  three-loop equal-mass banana integral}},
  \href{https://doi.org/10.1007/JHEP02(2022)184}{\emph{JHEP} {\bfseries 02}
  (2022) 184} [\href{https://arxiv.org/abs/2109.15251}{{\ttfamily
  2109.15251}}].

\bibitem{magnetic1}
D.~J. Broadhurst and W.~Zudilin, \emph{{A magnetic double integral}}, {\emph{J.
  Austral. Math. Soc.} {\bfseries 107} (2019) 9}.

\bibitem{magnetic3}
V.~Pasol and W.~Zudilin, \emph{Magnetic (quasi-)modular forms}, {\emph{Nagoya
  Math. J.} {\bfseries 248} (2022) 849}.

\bibitem{Bonisch:2024nru}
K.~B\"onisch, C.~Duhr and S.~Maggio, \emph{{Some conjectures around magnetic
  modular forms}},  \href{https://arxiv.org/abs/2404.04085}{{\ttfamily
  2404.04085}}.

\bibitem{ChenSymbol}
K.~T. Chen, \emph{{Iterated path integrals}}, {\emph{Bull.\ Amer.\ Math.\ Soc.}
  {\bfseries 83} (1977) 831}.

\bibitem{Henn:2013pwa}
J.~M. Henn, \emph{{Multiloop integrals in dimensional regularization made
  simple}}, \href{https://doi.org/10.1103/PhysRevLett.110.251601}{\emph{Phys.
  Rev. Lett.} {\bfseries 110} (2013) 251601}
  [\href{https://arxiv.org/abs/1304.1806}{{\ttfamily 1304.1806}}].

\bibitem{Dlapa:2022wdu}
C.~Dlapa, J.~M. Henn and F.~J. Wagner, \emph{{An algorithmic approach to
  finding canonical differential equations for elliptic Feynman integrals}},
  \href{https://doi.org/10.1007/JHEP08(2023)120}{\emph{JHEP} {\bfseries 08}
  (2023) 120} [\href{https://arxiv.org/abs/2211.16357}{{\ttfamily
  2211.16357}}].

\bibitem{bruinierbook}
J.~H. Bruinier, \emph{Borcherds Products on $O(2,l)$ and Chern Classes of
  Heegner Divisors}, Lecture Notes in Mathematics. Springer Berlin Heidelberg,
  2002.

\bibitem{orthogonal_PhD}
I.~H. Kl\"ocker, \emph{Modular Forms for the Orthogonal Group O(2,5)}, Ph.D.
  thesis, RWTH Aachen, November, 2005.

\bibitem{WANG2020107332}
H.~Wang and B.~Williams, \emph{On some free algebras of orthogonal modular
  forms},
  \href{https://doi.org/https://doi.org/10.1016/j.aim.2020.107332}{\emph{Advances
  in Mathematics} {\bfseries 373} (2020) 107332}.

\bibitem{Wang_2021}
H.~Wang, \emph{The classification of free algebras of orthogonal modular
  forms}, \href{https://doi.org/10.1112/S0010437X21007429}{\emph{Compositio
  Mathematica} {\bfseries 157} (2021) 2026}.

\bibitem{Schaps2022FourierCO}
F.~Schaps, \emph{Fourier coefficients of eisenstein series for $o^+(2, n +
  2)$},  \href{https://arxiv.org/abs/2212.09341}{{\ttfamily 2212.09341}}.

\bibitem{Schaps2023}
H.~R. Aloys~Krieg, Felix~Schaps, \emph{Eisenstein series for $o(2, n + 2)$},
  \href{https://arxiv.org/abs/2308.10709}{{\ttfamily 2308.10709}}.

\bibitem{Assaf:2022aa}
E.~Assaf, D.~Fretwell, C.~Ingalls, A.~Logan, S.~Secord and J.~Voight,
  \emph{Definite orthogonal modular forms: computations, excursions, and
  discoveries},
  \href{https://doi.org/10.1007/s40993-022-00373-2}{\emph{Research in Number
  Theory} {\bfseries 8} (2022) 70}.

\bibitem{Marzucca:2023gto}
R.~Marzucca, A.~J. McLeod, B.~Page, S.~P\"ogel and S.~Weinzierl, \emph{{Genus
  drop in hyperelliptic Feynman integrals}},
  \href{https://doi.org/10.1103/PhysRevD.109.L031901}{\emph{Phys. Rev. D}
  {\bfseries 109} (2024) L031901}
  [\href{https://arxiv.org/abs/2307.11497}{{\ttfamily 2307.11497}}].

\bibitem{Jockers:2024uan}
H.~Jockers, S.~Kotlewski, P.~Kuusela, A.~J. McLeod, S.~P\"ogel, M.~Sarve
  et~al., \emph{{A Calabi-Yau-to-curve correspondence for Feynman integrals}},
  \href{https://doi.org/10.1007/JHEP01(2025)030}{\emph{JHEP} {\bfseries 01}
  (2025) 030} [\href{https://arxiv.org/abs/2404.05785}{{\ttfamily
  2404.05785}}].

\bibitem{Bezuglov:2021jou}
M.~A. Bezuglov, \emph{{Integral representation for three-loop banana graph}},
  \href{https://doi.org/10.1103/PhysRevD.104.076017}{\emph{Phys. Rev. D}
  {\bfseries 104} (2021) 076017}
  [\href{https://arxiv.org/abs/2104.14681}{{\ttfamily 2104.14681}}].

\bibitem{Kreimer:2022fxm}
D.~Kreimer, \emph{{Bananas: multi-edge graphs and their Feynman integrals}},
  \href{https://doi.org/10.1007/s11005-023-01660-4}{\emph{Lett. Math. Phys.}
  {\bfseries 113} (2023) 38}
  [\href{https://arxiv.org/abs/2202.05490}{{\ttfamily 2202.05490}}].

\bibitem{Duhr:2022dxb}
C.~Duhr, A.~Klemm, C.~Nega and L.~Tancredi, \emph{{The ice cone family and
  iterated integrals for Calabi-Yau varieties}},
  \href{https://doi.org/10.1007/JHEP02(2023)228}{\emph{JHEP} {\bfseries 02}
  (2023) 228} [\href{https://arxiv.org/abs/2212.09550}{{\ttfamily
  2212.09550}}].

\bibitem{Vergu:2020uur}
C.~Vergu and M.~Volk, \emph{{Traintrack Calabi-Yaus from Twistor Geometry}},
  \href{https://doi.org/10.1007/JHEP07(2020)160}{\emph{JHEP} {\bfseries 07}
  (2020) 160} [\href{https://arxiv.org/abs/2005.08771}{{\ttfamily
  2005.08771}}].

\bibitem{McLeod:2023doa}
A.~J. McLeod and M.~von Hippel, \emph{{Traintracks All the Way Down}},
  \href{https://arxiv.org/abs/2306.11780}{{\ttfamily 2306.11780}}.

\bibitem{MR915841}
G.~Tian, \emph{Smoothness of the universal deformation space of compact
  {C}alabi-{Y}au manifolds and its {P}etersson-{W}eil metric},  in
  \emph{Mathematical aspects of string theory ({S}an {D}iego, {C}alif., 1986)},
  vol.~1 of \emph{Adv. Ser. Math. Phys.}, pp.~629--646, World Sci. Publishing,
  Singapore, (1987).

\bibitem{MR1027500}
A.~N. Todorov, \emph{The {W}eil-{P}etersson geometry of the moduli space of
  {${\rm SU}(n\geq 3)$} ({C}alabi-{Y}au) manifolds. {I}}, {\emph{Comm. Math.
  Phys.} {\bfseries 126} (1989) 325}.

\bibitem{Adams:2018yfj}
L.~Adams and S.~Weinzierl, \emph{{The $\varepsilon$-form of the differential
  equations for Feynman integrals in the elliptic case}},
  \href{https://doi.org/10.1016/j.physletb.2018.04.002}{\emph{Phys. Lett. B}
  {\bfseries 781} (2018) 270}
  [\href{https://arxiv.org/abs/1802.05020}{{\ttfamily 1802.05020}}].

\bibitem{Broedel:2018rwm}
J.~Broedel, C.~Duhr, F.~Dulat, B.~Penante and L.~Tancredi, \emph{{From modular
  forms to differential equations for Feynman integrals}},  in \emph{{KMPB
  Conference}: {Elliptic Integrals, Elliptic Functions and Modular Forms in
  Quantum Field Theory}}, pp.~107--131, 2019,
  \href{https://doi.org/10.1007/978-3-030-04480-0_6}{DOI}
  [\href{https://arxiv.org/abs/1807.00842}{{\ttfamily 1807.00842}}].

\bibitem{doran}
C.~Doran, \emph{{Picard--Fuchs Uniformization and Modularity of the Mirror
  Map}}, {\emph{Comm. Math. Phys.} {\bfseries 212} (2000) 625}.

\bibitem{BognerThesis}
M.~Bogner, \emph{{On differential operators of Calabi-Yau type}}, Ph.D. thesis,
  Johannes Gutenberg-Universit\"at Mainz, 2012.

\bibitem{BognerCY}
M.~Bogner, \emph{{Algebraic characterization of differential operators of
  Calabi-Yau type}},  \href{https://arxiv.org/abs/1304.5434}{{\ttfamily
  1304.5434}}.

\bibitem{Huybrechts_2016}
D.~Huybrechts, \emph{Lectures on K3 Surfaces}, Cambridge Studies in Advanced
  Mathematics. Cambridge University Press, 2016.

\bibitem{Mack:1969rr}
G.~Mack and A.~Salam, \emph{{Finite component field representations of the
  conformal group}},
  \href{https://doi.org/10.1016/0003-4916(69)90278-4}{\emph{Annals Phys.}
  {\bfseries 53} (1969) 174}.

\bibitem{Dirac:1936fq}
P.~A.~M. Dirac, \emph{{Wave equations in conformal space}},
  \href{https://doi.org/10.2307/1968455}{\emph{Annals Math.} {\bfseries 37}
  (1936) 429}.

\bibitem{Boulware:1970ty}
D.~G. Boulware, L.~S. Brown and R.~D. Peccei, \emph{{Deep-inelastic
  electroproduction and conformal symmetry}},
  \href{https://doi.org/10.1103/PhysRevD.2.293}{\emph{Phys. Rev. D} {\bfseries
  2} (1970) 293}.

\bibitem{Weinberg:2010fx}
S.~Weinberg, \emph{{Six-dimensional Methods for Four-dimensional Conformal
  Field Theories}},
  \href{https://doi.org/10.1103/PhysRevD.82.045031}{\emph{Phys. Rev. D}
  {\bfseries 82} (2010) 045031}
  [\href{https://arxiv.org/abs/1006.3480}{{\ttfamily 1006.3480}}].

\bibitem{Simmons-Duffin:2012juh}
D.~Simmons-Duffin, \emph{{Projectors, Shadows, and Conformal Blocks}},
  \href{https://doi.org/10.1007/JHEP04(2014)146}{\emph{JHEP} {\bfseries 04}
  (2014) 146} [\href{https://arxiv.org/abs/1204.3894}{{\ttfamily 1204.3894}}].

\bibitem{SB_1982-1983__25__251_0}
I.~Dolgachev, \emph{Integral quadratic forms : applications to algebraic
  geometry},  in \emph{S\'eminaire Bourbaki : volume 1982/83, expos\'es
  597-614}, no.~105-106 in Ast\'erisque, pp.~251--278, Soci\'et\'e
  math\'ematique de France, (1983).

\bibitem{Borcherds:1998aa}
R.~E. Borcherds, \emph{Automorphic forms with singularities on grassmannians},
  \href{https://doi.org/10.1007/s002220050232}{\emph{Inventiones mathematicae}
  {\bfseries 132} (1998) 491}.

\bibitem{hermitian_thesis}
A.~Wernz, \emph{{On Hermitian modular groups and modular forms}}, Ph.D. thesis,
  RWTH Aachen, June, 2019.

\bibitem{HAUFFEWASCHBUSCH202122}
A.~Hauffe-Waschb{\"u}sch and A.~Krieg, \emph{Congruence subgroups and
  orthogonal groups},
  \href{https://doi.org/https://doi.org/10.1016/j.laa.2021.01.011}{\emph{Linear
  Algebra and its Applications} {\bfseries 618} (2021) 22}.

\bibitem{Dolgachev:1996aa}
I.~V. Dolgachev, \emph{Mirror symmetry for lattice polarizedk3 surfaces},
  \href{https://doi.org/10.1007/BF02362332}{\emph{Journal of Mathematical
  Sciences} {\bfseries 81} (1996) 2599}.

\bibitem{verrill1996}
H.~A. Verrill, \emph{Root lattices and pencils of varieties},
  \href{https://doi.org/10.1215/kjm/1250518557}{\emph{J. Math. Kyoto Univ.}
  {\bfseries 36} (1996) 423}.

\bibitem{doranclingher1}
A.~Clingher and C.~Doran, \emph{{Modular invariants for lattice polarized K3
  surfaces}}, \href{https://doi.org/10.1307/mmj/1187646999}{\emph{Michigan
  Mathematical Journal} {\bfseries 55} (2007) 355 }.

\bibitem{Hosono:2002yb}
S.~Hosono, B.~H. Lian, K.~Oguiso and S.-T. Yau, \emph{{Classification of c = 2
  rational conformal field theories via the Gauss product}},
  \href{https://doi.org/10.1007/s00220-003-0927-0}{\emph{Commun. Math. Phys.}
  {\bfseries 241} (2003) 245}
  [\href{https://arxiv.org/abs/hep-th/0211230}{{\ttfamily hep-th/0211230}}].

\bibitem{doranclingher}
C.~Doran, A.~Clingher, J.~Lewis and U.~Whitcher, \emph{Normal forms, k3 surface
  moduli, and modular parametrizations}, {\emph{CRM-AMS proceedings and lecture
  notes} {\bfseries 47} (2009) 81}.

\bibitem{Bruinier2008}
J.~H. Bruinier, \emph{Hilbert modular forms and their applications},  in
  \emph{The 1-2-3 of Modular Forms: Lectures at a Summer School in
  Nordfjordeid, Norway}, K.~Ranestad, ed., (Berlin, Heidelberg), pp.~105--179,
  Springer Berlin Heidelberg, (2008),
  \href{https://doi.org/10.1007/978-3-540-74119-0_2}{DOI}.

\bibitem{hauffe-waschbusch_hilbert_2022}
A.~Hauffe-Waschb{\"u}sch and A.~Krieg, \emph{The {Hilbert} modular group and
  orthogonal groups},
  \href{https://doi.org/10.1007/s40993-022-00346-5}{\emph{Research in Number
  Theory} {\bfseries 8} (2022) 47}.

\bibitem{freitag}
E.~Freitag, \emph{{Hilbert Modular Forms}}. Springer, 1990.

\bibitem{Clingher2010LatticePK}
A.~Clingher and C.~F. Doran, \emph{Lattice polarized k3 surfaces and siegel
  modular forms}, {\emph{Advances in Mathematics} {\bfseries 231} (2010) 172}.

\bibitem{c9e0d978-3042-3cf9-a67c-48eb598b7003}
H.~Braun, \emph{Hermitian modular functions}, {\emph{Annals of Mathematics}
  {\bfseries 50} (1949) 827}.

\bibitem{Nagano:2024aa}
A.~Nagano and H.~Shiga, \emph{Geometric interpretation of hermitian modular
  forms via burkhardt invariants},
  \href{https://doi.org/10.1007/s00031-021-09681-w}{\emph{Transformation
  Groups} {\bfseries 29} (2024) 253}.

\bibitem{Mishnyakov:2023wpd}
V.~Mishnyakov, A.~Morozov and P.~Suprun, \emph{{Position space equations for
  banana Feynman diagrams}},
  \href{https://doi.org/10.1016/j.nuclphysb.2023.116245}{\emph{Nucl. Phys. B}
  {\bfseries 992} (2023) 116245}
  [\href{https://arxiv.org/abs/2303.08851}{{\ttfamily 2303.08851}}].

\bibitem{Cacciatori:2023tzp}
S.~L. Cacciatori, H.~Epstein and U.~Moschella, \emph{{Banana integrals in
  configuration space}},
  \href{https://doi.org/10.1016/j.nuclphysb.2023.116343}{\emph{Nucl. Phys. B}
  {\bfseries 995} (2023) 116343}
  [\href{https://arxiv.org/abs/2304.00624}{{\ttfamily 2304.00624}}].

\bibitem{Mishnyakov:2023sly}
V.~Mishnyakov, A.~Morozov and M.~Reva, \emph{{On factorization hierarchy of
  equations for banana Feynman integrals}},
  \href{https://doi.org/10.1016/j.nuclphysb.2024.116746}{\emph{Nucl. Phys. B}
  {\bfseries 1010} (2025) 116746}
  [\href{https://arxiv.org/abs/2311.13524}{{\ttfamily 2311.13524}}].

\bibitem{Mishnyakov:2024rmb}
V.~Mishnyakov, A.~Morozov, M.~Reva and P.~Suprun, \emph{{From equations in
  coordinate space to Picard-Fuchs and back}},
  \href{https://arxiv.org/abs/2404.03069}{{\ttfamily 2404.03069}}.

\bibitem{delaCruz:2024xit}
L.~de~la Cruz and P.~Vanhove, \emph{{Algorithm for differential equations for
  Feynman integrals in general dimensions}},
  \href{https://doi.org/10.1007/s11005-024-01832-w}{\emph{Lett. Math. Phys.}
  {\bfseries 114} (2024) 89}
  [\href{https://arxiv.org/abs/2401.09908}{{\ttfamily 2401.09908}}].

\bibitem{Kerr}
M.~Kerr, \emph{{private communication}}, .

\bibitem{ABE1973348}
Y.~Abe and S.~Katsura, \emph{Lattice green's function for the simple cubic and
  tetragonal lattices at arbitrary points},
  \href{https://doi.org/https://doi.org/10.1016/0003-4916(73)90073-0}{\emph{Annals
  of Physics} {\bfseries 75} (1973) 348}.

\bibitem{inprep}
C.~Duhr and S.~Maggio, \emph{in preparation.}, .

\bibitem{bams/1183526018}
H.~Bass, M.~Lazard and J.-P. Serre, \emph{{Sous-groupes d'indice fini dans
  $SL\left( {n,Z} \right)$}}, {\emph{Bulletin of the American Mathematical
  Society} {\bfseries 70} (1964) 385 }.

\bibitem{Raghunathan:2004aa}
M.~S. Raghunathan, \emph{The congruence subgroup problem},
  \href{https://doi.org/10.1007/BF02829437}{\emph{Proceedings Mathematical
  Sciences} {\bfseries 114} (2004) 299}.

\bibitem{10.1215/ijm/1258138437}
Y.~Yang and N.~Yui, \emph{{Differential equations satisfied by modular forms
  and $K3$ surfaces}},
  \href{https://doi.org/10.1215/ijm/1258138437}{\emph{Illinois Journal of
  Mathematics} {\bfseries 51} (2007) 667 }.

\bibitem{bimodular}
J.~Stienstra and D.~Zagier, \emph{{Bimodular forms and holomorphic anomaly
  equation}},  in \emph{{Workshop on Modular Forms and String Duality}}, Banff
  International Research Station, 2006.

\bibitem{WangYang}
L.~Wang and Y.~Yang, \emph{{Ramanujan-type $1/\pi$-series from bimodular
  forms}}, {\emph{Ramanujan J.} {\bfseries 59} (2022) 831}.

\bibitem{Manschot:2021qqe}
J.~Manschot and G.~W. Moore, \emph{{Topological correlators of $SU(2) N=2$ SYM
  on four-manifolds}},
  \href{https://doi.org/10.4310/atmp.240914021307}{\emph{Adv. Theor. Math.
  Phys.} {\bfseries 28} (2024) 407}
  [\href{https://arxiv.org/abs/2104.06492}{{\ttfamily 2104.06492}}].

\bibitem{Aspman:2021evt}
J.~Aspman, E.~Furrer and J.~Manschot, \emph{{Four flavors, triality, and
  bimodular forms}},
  \href{https://doi.org/10.1103/PhysRevD.105.025017}{\emph{Phys. Rev. D}
  {\bfseries 105} (2022) 025017}
  [\href{https://arxiv.org/abs/2110.11969}{{\ttfamily 2110.11969}}].

\end{thebibliography}\endgroup

\end{document}